\documentclass[10pt, a4paper, onecolumn, external]{deepmind}

\usepackage[authoryear, sort&compress, round]{natbib}
\usepackage{calc}
\usepackage{tikz}
\usepackage{pgfplots}
\pgfplotsset{compat=1.17}
\usetikzlibrary{arrows.meta}
\usetikzlibrary{patterns}
\usetikzlibrary{fadings,shadings}
\usepackage{transparent}
\usepackage{pgffor}
\usepackage[skins,breakable]{tcolorbox}
\usepackage{ifthen}
\usepackage{contour}

\usepackage{twoxtwo}

\usepackage[noend]{algorithm2e}
\RestyleAlgo{ruled}
\SetKwComment{Comment}{/* }{ */}

\usepgfplotslibrary{ternary}
\usetikzlibrary{matrix}
\usetikzlibrary{positioning}

\usepackage{amsthm, amsmath, amsfonts, amssymb}
\usepackage{mathtools}

\newtheorem{theorem}{Theorem}

\newtheorem{lemma}{Lemma}
\newtheorem{remark}{Remark}
\newtheorem{definition}{Definition}

\usepackage{caption}
\usepackage{subcaption}

\DeclareMathOperator*{\sort}{sort}
\DeclareMathOperator*{\lexsort}{lexsort}

\newcount\Comments  
\newcommand{\kibitz}[2]{\ifnum\Comments=1{\color{#1}{#2}}\fi}

\title{Equilibrium-Invariant Embedding, Metric Space, and Fundamental Set of 2×2 Normal-Form Games}

\correspondingauthor{marris@deepmind.com}

\paperurl{}

\reportnumber{} 

\renewcommand{\today}{}

\author[1,2]{Luke Marris}
\author[1]{Ian Gemp}
\author[1]{Georgios Piliouras}

\affil[1]{DeepMind}
\affil[2]{University College London}

\begin{abstract}
Equilibrium solution concepts of normal-form games, such as Nash equilibria, correlated equilibria, and coarse correlated equilibria, describe the joint strategy profiles from which no player has incentive to unilaterally deviate. They are widely studied in game theory, economics, and multiagent systems. Equilibrium concepts are invariant under certain transforms of the payoffs. We define an equilibrium-inspired distance metric for the space of all normal-form games and uncover a distance-preserving equilibrium-invariant embedding. Furthermore, we propose an additional transform which defines a better-response-invariant distance metric and embedding. To demonstrate these metric spaces we study 2×2 games. The equilibrium-invariant embedding of 2×2 games has an efficient two variable parameterization (a reduction from eight), where each variable geometrically describes an angle on a unit circle. Interesting properties can be spatially inferred from the embedding, including: equilibrium support, cycles, competition, coordination, distances, best-responses, and symmetries. The best-response-invariant embedding of 2×2 games, after considering symmetries, rediscovers a set of 15 games, and their respective equivalence classes. We propose that this set of game classes is fundamental and captures all possible interesting strategic interactions in 2×2 games. We introduce a directed graph representation and name for each class. Finally, we leverage the tools developed for 2×2 games to develop game theoretic visualizations of large normal-form and extensive-form games that aim to fingerprint the strategic interactions that occur within.
\end{abstract}

\begin{document}

\DeclareRobustCommand{\colorsquarergb}[3]{%
\begin{tikzpicture}%
    \node[fill={rgb:red,#1;green,#2;blue,#3},rectangle,minimum size=0.6em] {};
\end{tikzpicture}%
}

\DeclareRobustCommand{\colorsquare}[1]{%
\begin{tikzpicture}%
    \node[fill=#1,rectangle,minimum size=0.6em] {};
\end{tikzpicture}%
}

\DeclareRobustCommand{\northwestpattern}{%
\begin{tikzpicture}%
    \path[pattern=south east lines,pattern color=gray!20]
        (-0.45em, -0.45em) --
        (-0.45em, 0.45em) --
        (0.45em, 0.45em) --
        (0.45em, -0.45em) --
        cycle;
\end{tikzpicture}%
}

\DeclareRobustCommand{\northeastpattern}{%
\begin{tikzpicture}%
    \path[pattern=south west lines,pattern color=gray!20]
        (-0.45em, -0.45em) --
        (-0.45em, 0.45em) --
        (0.45em, 0.45em) --
        (0.45em, -0.45em) --
        cycle;
\end{tikzpicture}%
}

\DeclareRobustCommand{\verticalpattern}{%
\begin{tikzpicture}%
    \path[pattern=vertical lines,pattern color=gray!20]
        (-0.45em, -0.45em) --
        (-0.45em, 0.45em) --
        (0.45em, 0.45em) --
        (0.45em, -0.45em) --
        cycle;
\end{tikzpicture}%
}

\definecolor{tab1color}{RGB}{102,194,165}
\definecolor{tab2color}{RGB}{252,141,98}
\definecolor{tab3color}{RGB}{141,160,203}
\definecolor{tab4color}{RGB}{231,138,195}
\definecolor{tab5color}{RGB}{166,216,84}
\definecolor{tab6color}{RGB}{255,217,47}

\definecolor{tabB1color}{RGB}{31,119,180}
\definecolor{tabB2color}{RGB}{255,127,14}
\definecolor{tabB3color}{RGB}{44,160,44}

\tikzset{ 
    table/.style={
        matrix of nodes,
        row sep=-\pgflinewidth,
        column sep=-\pgflinewidth,
        nodes={
            rectangle,
            draw=black,
            align=center,
        },
        minimum height=0.6em,
        text depth=0.2ex,
        text height=0.6ex,
        text width=0.40em,
        align=center,
        font=\tiny,
        nodes in empty cells,
        column 3/.style={
            nodes={fill=gray!10}
        },
        column 4/.style={
            nodes={fill=gray!10}
        },
    }
}

\maketitle

\section{Introduction}

Equilibrium solutions to normal-form games, such as Nash equilibrium (NE) \citep{nash1951_neq}, correlated equilibrium (CE) \citep{aumann1974_ce}, and coarse correlated equilibrium (CCE) \citep{hannan1957_cce,moulin1978_cce,young2004_strategic_cce} are ubiquitously used to model the strategic behaviour of rational utility maximizing players in games. In some classes of games, such as two-player zero-sum games, the Nash equilibrium solution concept is considered fundamental \citep{vonneumann1947_game_theory_book}, because it is unexploitable and interchangeable in this class. CEs and CCEs are important in n-player general-sum games, which may require coordination facilitated by a correlation device. The set of equilibria in a game is invariant to certain transforms of the payoffs \citep{morris2004_best_response_equivalence}. The most well known is the affine transform (offset and positive scale of each player's payoff). We call such transforms \emph{equilibrium-invariant}. In addition, there are symmetries in payoffs which result in equivalent symmetries in the set of equilibria (for example, the order of strategies or players). We call transforms over symmetries \emph{equilibrium-symmetric}. Finally, there is a weaker notion of better-response invariance, where transforms do not change a player's preference over responses to other player's joint strategies. We call such transforms \emph{better-response-invariant}. This work studies transforms to produce metric spaces and embeddings over n-player general-sum normal-form games. To motivate the importance of these metric spaces we then focus on 2×2 games: exploring their properties, visualizing their structure, and rediscovering a set of 15 fundamental games.

2×2 normal-form games have two players, each with two strategies. Players take one of the two strategies (possibly at random) simultaneously. The resulting joint strategy triggers a payoff for each player. The game is played only once. The table of payoffs determines rational behaviour of the players (i.e. the set of equilibria). Popular games are given names based on the payoffs and the resulting behaviour. 2×2 games are utilized so frequently that their names have entered popular culture (Figure~\ref{fig:canonical_2x2_games}), for example: \chickengame~Chicken, \prisonersgame~Prisoner's Dilemma, \huntgame~Stag Hunt \citep{skyrms2004_stag_hunt}, \battlegame~Bach or Stravinsky (battle of the sexes), and \penniesgame~Matching Pennies\footnote{In this work, we accompany game names with a graphical representation which describes either each player's preference over joint payoffs for ordinal games (e.g. \penniesgame~Matching Pennies) or their best-response preferences for best-response-invariant embeddings (e.g. \cyclegame~Cycle). We describe these representations more thoroughly in later sections.}. The study of such games \citep{harold2002_atlas_of_interpersonal_situations} is crucial to understanding cooperation \citep{gauthier1986_morals_by_agreement}, competition, coordination, nature \citep{wilkinson1984_vampire_bat}, incentive structures \citep{sugden2005_economics_of_rights_cooperation_and_welfare}, social dilemmas \citep{bruns2021_archetypal_games}, utilitarian behaviour, rational behaviour \citep{gintis2014_bounds_of_reason}, and seemingly irrational behaviour. Games are used to inform economic policy \citep{ostrom1994_rules}, social structure \citep{skyrms2004_stag_hunt,bicchieri2005_gammar_of_society,binmore1994_social_contract}, foreign policy \citep{schelling1966_arms_and_influence}, pandemic response, and environmental treaties \citep{breton2006_environmental_projects,branzei2021_economies,schosser2022_fairness_pandemic}.

As a result, great effort has been expended in creating parameterizations, taxonomies, topologies, and names for 2×2 games \citep{rapoportandguyer1966_taxonomy_of_2x2_games,rapoport1976_the_2x2_game_book,kilgour1988_taxonomy_of_all_ordinal_2x2_games,robinsonandgoforth2005_topology_of_2x2_games_book,brams1993_theory_of_moves,bruns2015_names_for_games,robinson2007_toward_a_topological_treatment_of_the_nonstrictly_ordered_2x2_games,walliser1988_simplifed_taxonomy_2x2}. The most common of these focus on ordinal games, a type of equivalence that recognises that resulting joint strategies in games are only meaningfully different up to the partial ordering of the elements of each player's payoff. There are 726 such games \citep{fraser1986_non_strict_2x2_games}. Most other work \citep{rapoportandguyer1966_taxonomy_of_2x2_games,goforth2005_periodic_table_of_games,harris1969_geometric_classification_of_symmetric_2x2,huertasrosero2003_cartography_of_symmetric_2x2,boors2022_2x2_game_classification_by_decomposition} only considers subsets (e.g. symmetric or strictly ordinal) of 2×2 games. \cite{borm1987_classification_of_2x2_games} classified all 2×2 games into 15 distinct classes studying their NE best-responses. \cite{fishburn1990_binary_2x2_games} later showed that these classes can be represented with binary games. Neither provide a notion of closeness, distance metric, or satisfying parameterization for these games.

This work uncovers metric spaces and embeddings for general-sum n-player normal-form games. We demonstrate these embeddings in 2×2 games (Figure~\ref{fig:summary}). All nontrivial (Definition~\ref{def:trivial_payoff}) 2×2 games can be transformed to an equilibrium-invariant embedding which can be parameterized using only two variables, a reduction from the eight needed to represent the original payoffs (Table~\ref{tab:num_variables}). Remarkably, this does not change the set of equilibria in the game. Geometrically, each of the two variables describes an angle on a circle and has spatial meaning: similar games are situated near each other. Properties of a game, like equilibrium support, zero-sum-invariance, common-payoff-invariance, whether it is clockwise or anti-clockwise cyclic, whether it is coordination or anti-coordination, and the best-response dynamics, can be easily deduced from this embedding. Using symmetry, the area of equilibrium-invariant embedding can be reduced by a factor of eight to result in the equilibrium-symmetric embedding. Additionally this space can be further reduced to a set of games with a cardinality of 15: the best-response-invariant embedding. Because many equilibrium-invariant embeddings map to the same best-response-invariant embedding, they are also part of equivalence classes. These are the same classes found by \cite{borm1987_classification_of_2x2_games}, but derived through a different but related argument. We improve upon Borm's classification by situating them in the equilibrium-invariant embedding, illuminating the relationship between the games. We also provide names and an elegant graphical visualization of the classes.

Finally, we leverage the tools developed in 2×2 games to study games with more players and strategies, including large normal-form and extensive-form games. Historically, these games have been considered intractable to visualize or summarize. We develop game-theoretic visualizations that fingerprint strategic interactions within these games. The field of machine learning has enjoyed such simplified visualizations of high dimensional complex data. Principled techniques like PCA \citep{pearson1901_pca,hotelling1936_pca} aim to reduce dimensionality, while maintaining the maximum amount of information. Other less principled techniques like t-SNE \citep{vandermaaten2008_tsne,hintonroweis2003_sne} are also very popular. We hope that our visualizations prove useful for game theory practitioners, where we believe such visualization tools are underdeveloped.

\begin{figure}[t]
    \centering
    \footnotesize
    \begin{subfigure}[b]{0.20\textwidth}
        \centering
        \begin{tabular}{c|cc}
            {\large \ordinalgame{1423}{1243}} & C & S \\ \hline
            C & $-9,-9$ & $+1,-1$ \\
            S & $-1,+1$ & $~0,~~0$ \\
        \end{tabular}
        \caption{\centering Chicken}
        \label{fig:chicken}
    \end{subfigure}
    \hfill
    \begin{subfigure}[b]{0.22\textwidth}
        \centering
        \begin{tabular}{c|cc}
            {\large \ordinalgame{3142}{3412}}  & C & D \\ \hline
            C & $-2,-2$ & $-9,+0$ \\
            D & $+0,-9$ & $-5,-5$ \\
        \end{tabular}
        \caption{\centering Prisoner's Dilemma}
        \label{fig:prisoners_dilemma}
    \end{subfigure}
    \hfill
    \begin{subfigure}[b]{0.16\textwidth}
        \centering
        \begin{tabular}{c|cc}
            {\large \ordinalgame{4132}{4312}} & S & H \\ \hline
            S & $4,4$ & $1,3$ \\
            H & $3,1$ & $2,2$ \\
        \end{tabular}
        \caption{\centering Stag Hunt}
        \label{fig:stag_hunt}
    \end{subfigure}
    \hfill
    \begin{subfigure}[b]{0.19\textwidth}
        \centering
        \begin{tabular}{c|cc}
            {\large \ordinalgame{4113}{3114}} & M & F \\ \hline
            M & $3,2$ & $0,0$ \\
            F & $0,0$ & $2,3$ \\
        \end{tabular}
        \caption{\centering  Bach or Stravinsky}
        \label{fig:battle_of_the_sexes}
    \end{subfigure}
    \hfill
    \begin{subfigure}[b]{0.21\textwidth}
        \centering
        \begin{tabular}{c|cc}
            {\large \ordinalgame{3113}{1331}}  &       H &       T \\ \hline
            H & $+1,-1$ & $-1,+1$ \\
            T & $-1,+1$ & $+1,-1$ \\
        \end{tabular}
        \caption{\centering Matching Pennies}
        \label{fig:game_matching_pennies}
    \end{subfigure}
    
    \caption{Payoff tables of common 2×2 normal-form games. Player 1 selects a row strategy and player 2 selects a column strategy. Each player respectively receives one of the payoffs in the tuple. The joint payoff preference ordering is shown in the top-left for each player.}
    \label{fig:canonical_2x2_games}
\end{figure}

\begin{figure}
    \centering
        \centering
        \begin{tikzpicture}[
        squarenode/.style={rectangle, draw=black!60, fill=black!5, very thick, minimum size=2.6cm,text width=2.5cm,align=center}]
        \footnotesize
        
        \node[squarenode] (game_space) { ~ \\Game Space\\ ~ \\$G_p(a_1, a_2) \in \mathbb{R}$\\$\forall p, a_1, a_2$\\(8 variables)};
        \node[squarenode,right=1.65cm of game_space] (equilibrium_invariant_embedding) {Equilibrium\\Invariant\\Embedding\\$\theta_1 \in [-\pi, +\pi]$\\$\theta_2 \in [-\pi, +\pi]$\\(2 variables)};
        \node[squarenode,right=1.65cm of equilibrium_invariant_embedding] (equilibrium_symmetric_embedding) {Equilibrium\\Symmetric\\ Embedding\\$\phi_\text{coord}, \phi_\text{cycle}$\\(eighth of area)\\(2 variables)};
        \node[squarenode,right=2.1cm of equilibrium_symmetric_embedding] (best_response_invariant_embedding) {Best-Response\\Invariant\\Embedding\\(11 nontrivial)\\(3 partially trivial)\\(1 trivial)};
        
        \draw[->] (game_space.east) -- (equilibrium_invariant_embedding.west) node[midway,above] {Theorem~\ref{thm:equilibrium_invariant_embedding}};
        \draw[->] (equilibrium_invariant_embedding.east) -- (equilibrium_symmetric_embedding.west) node[midway,above] {Theorem~\ref{thm:equilibrium_symmetric_embedding}};
        \draw[->] (equilibrium_symmetric_embedding.east) -- (best_response_invariant_embedding.west) node[midway,above] {Theorems~\ref{thm:per_strategy_scale_transform}\&\ref{thm:fundamental_set}};
        
    \end{tikzpicture}
    \caption{Summary of the main contributions of this work showing how all 2×2 games can be transformed to an equilibrium-invariant embedding, an equilibrium-symmetric embedding, and a best-response-invariant embedding.}
    \label{fig:summary}
\end{figure}
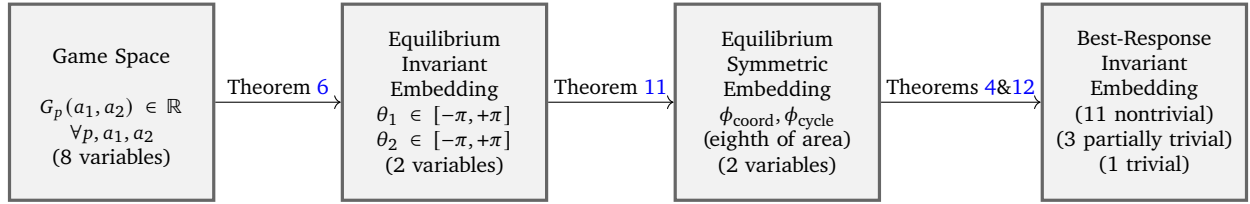

\section{Preliminaries and Related Work}

\paragraph{Game Theory}
Game theory is the study of the interactive behaviour of rational payoff maximizing agents in the presence of other agents. The environment that the agents operate in is called a game. We focus on a particular type of single-shot, simultaneous move game called a \emph{normal-form} game. A normal-form game consists of $N$ players, a set of strategies available to each player $p$, parameterized by $a_p \in \mathcal{A}_p = \{a_p^A, a_p^B, ...\}$, and a payoff for each player under a particular joint action, $G_p(a)$, where $a = (a_1, ..., a_N) = (a_p, a_{-p}) \in \mathcal{A} = \otimes_{p} \mathcal{A}_p$. The subscript notation $-p$ means ``all players apart from player $p$''. The plural, payoffs, is used to describe payoff of all players, $G(a) = (G_1(a), ..., G_N(a))$. When the payoff is not parameterized by a joint strategy, $G_p$, it refers to a full tensor of shape $[|\mathcal{A}_1|,...,|\mathcal{A}_N|]$. Elements of the payoff may be directly described as variables, $G_p(a_p^A, a_p^B) = g_p^{AB}$. Players need not act deterministically (\emph{pure strategy}), they can randomize over their strategies (\emph{mixed strategy}), and their behaviour can be correlated with other players (\emph{joint mixed strategy}). Players' behaviour is described by a joint distribution  $\sigma(a)$. The goal of each player is to maximize their expected payoff, $\sum_{a \in \mathcal{A}} \sigma(a) G_p(a)$. Players could play independently if the joint factorizes, such that $\sigma(a) = \otimes_p \sigma_p(a_p)$, where $\sigma_p(a_p)$ is a marginal mixed strategy. However this is limiting because it does not allow players to coordinate. A mediator called a \emph{correlation device} could be employed to allow players to execute arbitrary joint strategies $\sigma(a)$ that do not necessarily factorize into their marginals. Such a mediator would sample from a publicly known joint $\sigma(a)$ and secretly communicate to each player their recommended strategy. Game theory is most developed in a subset of games: those with two players and a restriction on the payoffs, $G_1(a_1, a_2) = -G_2(a_1, a_2)$, known as zero-sum. Particularly in n-player, general-sum games, there is no definitive solution concept. One approach is to consider joints that are in equilibrium: distributions that no player has incentive to unilaterally deviate from. Games are sometimes referred to by their shape, for example: $|\mathcal{A}_1|\times...\times|\mathcal{A}_N|$. In this work we primarily focus on general-sum 2×2 normal-form games.

\paragraph{Equilibrium Solution Concepts}
Broadly, \emph{solutions} of games are predictions of how self-interested rational players with agency may behave in a game. There are many solutions concepts, however the most popular class are equilibrium solutions that describe stable distributions of player strategies: distributions that no player has incentive to unilaterally deviate from. The most famous is Nash equilibrium (NE) \citep{nash1951_neq} which defines conditions for the set of stable marginal, $\sigma_p(a_p)$, strategies for each player that can be executed independently. Because of the independence property (also known as factorizable joint property), it is commonly deployed in purely competitive (zero-sum) settings. NE is straightforward to compute\footnote{Solvable in polynomial time.} in two-player zero-sum games where it is interchangeable\footnote{If there are any two equilibria $(\sigma_1(a_1), \sigma_2(a_2))$ and $(\sigma'_1(a_1), \sigma'_2(a_2))$, $(\sigma_1(a_1), \sigma'_2(a_2))$ and $(\sigma'_1(a_1), \sigma_2(a_2))$ are also equilibria.} and unexploitable\footnote{When playing the equilibrium there is no response from the opponent that will reduce the payoff.}. NE is considered fundamental in two-player zero-sum games where it is prescriptive. The set of NEs is always nonempty in finite games. However, in general, the set of NEs may be isolated from one another. In general, how to select the best NE is unclear (the \emph{equilibrium selection problem} \citep{harsanyi1988_eq_selection}), and finding a specific NE is NP-complete \citep{austrin2011_inapproximability}. Finding any NE is PPAD complete \citep{chen2006_nash_complexity,daskalakis2009_ne_complexity}. For n-player mixed-motive normal-form games, other mediated equilibria concepts, such as correlated equilibrium (CE) \citep{aumann1974_ce} and coarse correlated equilibrium (CCE) \citep{hannan1957_cce,moulin1978_cce} may be more suitable. These describe conditions on distributions over stable joint strategies, $\sigma(a)$. (C)CEs permit a larger feasible set of distributions and therefore can result in higher payoffs for players. However, they cannot be executed independently: they require a trusted coordination device to recommend an action to each player. Coordination is a sensible property if the setting allows for it. The key advantage of (C)CEs is that they have linear constraints and form a convex polytope of equilibria. This means the set can be computed in polynomial time for n-player general-sum games.

Well-Supported Correlated Equilibria (WSCEs) \citep{papadimitriouandgoldberg2006_well_supported_epsilon_eq,czumaj2014_well_supported_ne} can be defined in terms of linear inequality constraints. The deviation gain, $A^\text{WSCE}_p(a'_p, a''_p, a_{-p})$, of a player is the change in payoff the player achieves when deviating to strategy $a'_p$ after being recommended strategy $a''_p$, when the other players play $a_{-p}$. A joint, $\sigma(a)$, is in $\epsilon$-WSCE if the deviation gain is no more than some constant, $\epsilon_p \leq \epsilon$, for every recommendation, $a''_p$, deviation, $a'_p$, and player, $p$. The scalar $\epsilon_p$ describes how approximate the equilibrium is, and also commonly serves a role as a hyperparameter in algorithms where it is usually chosen to be zero. The deviation gain inequality can be expressed as:
\begin{align}
    \smashoperator{\sum_{a_{-p} \in \mathcal{A}_{-p}}} \sigma(a_{-p}|a''_p)  \left[A^\text{WSCE}_p(a'_p, a''_p, a_{-p}) := G_p(a'_p, a_{-p}) - G_p(a''_p, a_{-p}) \right] &\leq \epsilon_p \qquad \forall p \in [1,N], a''_p \neq a'_p  \in \mathcal{A}_p. \label{eq:wsce_def}
\end{align}
The conditional distribution is cumbersome, so the standard correlated equilibrium \citep{aumann1974_ce} is defined instead with a joint distribution. When $\epsilon_p = 0$, WSCE is equivalent to CE.
\begin{align}
    \smashoperator{\sum_{a \in \mathcal{A}}} \sigma(a) \left[ A^\text{CE}_p(a'_p, a''_p, a) := \begin{cases}
    A^\text{WSCE}_p(a'_p, a''_p, a_{-p}) & a_p = a''_p \\
    0 & \text{otherwise}
    \end{cases} \right] &\leq \epsilon_p \qquad \forall p \in [1,N], a''_p \neq a'_p  \in \mathcal{A}_p  \label{eq:ce_def}
\end{align}
Coarse correlated equilibria \citep{hannan1957_cce,moulin1978_cce,young2004_strategic_cce} are similar to CEs but a player may only consider deviating before receiving a recommended strategy. Therefore the deviation gain for CCEs can be derived from the CE definition by summing over all possible recommended strategies $a''_p$: $A^\text{CCE}_p(a'_p, a) = \sum_{a''_p} A^\text{CE}_p(a'_p, a''_p, a)$.
\begin{align}
    \smashoperator{\sum_{a \in \mathcal{A}}} \sigma(a) \left[A^\text{CCE}_p(a'_p, a) := G_p(a'_p, a_{-p}) - G_p(a) \right] &\leq \epsilon_p  \qquad \forall p \in [1,N], a'_p \in \mathcal{A}_p \label{eq:cce_def}
\end{align}
NEs \citep{nash1951_neq} have similar definitions to (C)CEs but have an extra constraint: the joint distribution factorizes $\otimes_p \sigma(a_p) = \sigma(a)$. Either enforcing a joint factorizes, or optimizing the marginals directly results in nonlinear constraints\footnote{This is why NEs are harder to compute than (C)CEs.}. The definition of the NE can be optionally defined in terms of the CCE or CE deviation gains. Using the CE deviation gains we define the $\epsilon$-NE and $\epsilon$-WSNE.
\begin{align}
    \smashoperator{\sum_{a \in \mathcal{A}}} \otimes_{q} \sigma(a_q) A^\text{CE}_p(a'_p, a''_p, a) &\leq \epsilon_p \qquad \forall p \in [1,N], a''_p \neq a'_p \in \mathcal{A}_p \label{eq:ne_def} \\
    \smashoperator{\sum_{a_{-p} \in \mathcal{A}_{-p}}} \otimes_{q \neq p} \sigma(a_q) A^\text{WSCE}_p(a'_p, a''_p, a_{-p}) &\leq \epsilon_p \qquad \forall p \in [1,N], a''_p \neq a'_p \in \mathcal{A}_p \label{eq:wsne_def}
\end{align}

\paragraph{Equilibrium-Invariant Transforms}
Certain transforms to payoffs, $G_p(a) \to \hat{G}_p(a)$, do not change the set of equilibria, $\sigma^*(a) \to \hat{\sigma}^*(a) = \sigma^*(a)$. These transforms can be used to reduce the degrees of freedom in the game. The most common invariant transform \citep{ostrovski2013_fictitious_play_dynamics,morris2004_best_response_equivalence,moulin1978_cce} is the affine transform (Theorem~\ref{thoerem:affine_transform}) which consists of an offset over the other players' strategies and a positive scale. \cite{harsanyi1988_eq_selection} argues that invariant transforms and symmetric equivalences are important properties in equilibrium selection. \cite{marris2022_turbo_arxiv} leverages both to simplify computing (C)CEs using neural networks. Invariances are also leveraged more generally in machine learning to efficiently optimize over rich topological spaces, e.g., in semidefinite programming and Riemannian optimization~\citep{boumal2014_manopt}. The work in this paper utilizes existing equilibrium-invariant transforms to derive a new equilibrium-invariant embedding for general n-player and 2×2 games, and proposes a novel isomorphism (Theorem~\ref{thm:per_strategy_scale_transform}), which is used to rediscover a fundamental set of 2×2 game classes.

\begin{table}[t]
    \centering
    \begin{subtable}[t]{0.32\textwidth}
        \centering
        \begin{tabular}{l|cc}
            ~ & A & B \\ \hline
            A & $g_1^{AA}$, $g_2^{AA}$ & $g_1^{AB}$, $g_2^{AB}$ \\
            B & $g_1^{BA}$, $g_2^{BA}$ & $g_1^{BB}$, $g_2^{BB}$
        \end{tabular}
        \caption{\centering Full Normal-Form\\(8 variables)}
        \label{tab:full_variables}
    \end{subtable}
    \hfill
    \begin{subtable}[t]{0.32\textwidth}
        \centering
        \begin{tabular}{l|cc}
            ~ & A & B \\ \hline
            A & $g^{AA}$ & $g^{AB}$ \\
            B & $g^{BA}$ & $g^{BB}$
        \end{tabular}
        \caption{\centering Symmetric, Common or Zero-Sum\\(4 variables)}
        \label{tab:symmetric_common_zero_sum_variables}
    \end{subtable}
    \hfill
    \begin{subtable}[t]{0.32\textwidth}
        \centering
        \begin{tabular}{l|cc}
            ~ & A & B \\ \hline
            A & $0$ & $g^{AB}$ \\
            B & $-g^{AB}$ & $0$
        \end{tabular}
        \caption{\centering Symmetric Zero-Sum\\(1 variable)}
        \label{tab:symmetric_zero_sum_variables}
    \end{subtable}
    \caption{Number of variables needed to describe 2×2 normal-form games and sparser simplifications.}
    \label{tab:num_variables}
\end{table}

\paragraph{Two-Player Two-Strategy Games}
2×2 games are the smallest possible, but have remarkable strategic depth and explain many real-world interactions. As a result, a rich literature has accumulated on explaining, visualizing, categorizing, parameterizing and naming 2×2 games. Naively a 2×2 game (Table~\ref{tab:full_variables}) requires 8 variables to define which is too many variables to intuit. Therefore many approaches attempt to reduce this complexity through either using invariances, symmetries and equivalences, or considering a subset of games (for example symmetric, common payoff, or zero-sum games; Tables~\ref{tab:symmetric_common_zero_sum_variables} and \ref{tab:symmetric_zero_sum_variables}). After this reduction in complexity, either a finite set of games or a reduced space of games remains. This set or space may have a structure that describes how close it is to other games. Approaches that simply bin games into a set are \emph{categorical} approaches. Those that also impose a notion of similarity or closeness between games are \emph{topological}. Those that impose hierarchical structure to the categorization are called \emph{taxonomies}. Those that parameterize the game are \emph{parametric}. Approaches may have multiple properties.

Games can be characterized based on their payoffs. \cite{rapoportandguyer1966_taxonomy_of_2x2_games}'s ``taxonomy of games'' (and book \citep{rapoport1976_the_2x2_game_book}) exploits a particular equivalence class where only the order of each player's payoffs matter. Changes in the magnitude of payoffs that do not change the order result in predictable scaling of the equilibrium of the game. Additionally, if we only consider the subset of games that have strict ordering (\emph{strict ordinal games}) we can represent the payoffs with permutations of the set of ordinal numbers $\{1,2,3,4\}$. This results in $4! = 24$ ways to strictly order each player's payoff, which results in $24 \times 24 = 576$ games. Utilizing strategy permutations\footnote{The literature phrases this as ``order graph'' equivalence.} for each player reduces the number of strict ordinal games to $144$. When including player permutations\footnote{Sometimes referred to as ``reflection'' in the literature. This symmetry does not quite halve the space because all symmetric games are retained. Of the 78 strict ordinal games, 66 are non-symmetric and 12 are symmetric. The literature is shier to utilize this symmetry because it reverses the roles of players.} half the non-symmetric games can be removed which reduces the cardinality to 78 strict ordinal games. A drawback of this approach is that it only classifies games with strict payoff orderings. \cite{fraser1986_non_strict_2x2_games} introduced the categorization of partially ordered ordinal games (\emph{partial ordinal games}) where payoffs can take on equal values. This results in a total cardinality of 1413 partially ordered ordinal games with strategy symmetries, or 726 if player symmetries are utilized. Partial ordinal games can be visualized using a graphical representation which shows the ordering of joint preferences (Figure~\ref{fig:canonical_2x2_games}).  \cite{goforth2005_periodic_table_of_games} improved the categorization of strict ordinal games to produce a ``periodic table''. The $144$ games were distributed in a $12 \times 12$ grid such that adjacent games had similar properties. \cite{robinson2007_toward_a_topological_treatment_of_the_nonstrictly_ordered_2x2_games} extended this topology to include partially ordered games. \cite{bruns2015_names_for_games} suggested a formulaic naming scheme for all ordinal games as previously only a small subset (often symmetric games) had established common names.

\cite{harris1969_geometric_classification_of_symmetric_2x2} gives a parameterized classification of symmetric 2×2 games using only two variables which are functions of the payoffs (Table~\ref{tab:symmetric_common_zero_sum_variables}): $r_3 = \frac{g^{BB} - g^{AB}}{g^{BA} - g^{AB}}$ and $r_4 = \frac{g^{BA} - g^{AA}}{g^{BA} - g^{AB}}$, with the constraint that $g^{BA} > g^{AB}$. This defines a plane, with regions that correspond to classes of games with similar properties. \cite{huertasrosero2003_cartography_of_symmetric_2x2} also classifies symmetric 2×2 games into 8 base classes, and $12$ total classes based on their NE. 
\cite{boors2022_2x2_game_classification_by_decomposition} classifies symmetric 2×2 games into $24$ classes based on their decomposition into common-payoff and zero-sum parts. \cite{germano2006_geometry_of_nf_games} classifies games into various equivalent classes based on their Nash equilibrium geometry. \cite{borm1987_classification_of_2x2_games} classifies all 2×2 games into a set of 15 games based on their best-response characteristics. The games can also be parameterized with 4 discrete variables. Borm did not describe any notions of similarity or distance between games in the set. Our work rediscovers \cite{borm1987_classification_of_2x2_games}'s classification through a scale-based payoff transform. In addition, we extend this classification to a metric space, provide a much more efficient two variable embedding of these games, and name these games.

\section{Metric Spaces for General Normal-Form Games}
\label{sec:metric_spaces}

Let $G \in \mathcal{G}=(\mathcal{G}_1, ..., \mathcal{G}_N)$ be the space of games with a particular number of players and strategies. $G_p \in \mathcal{G}_p = \mathbb{R}^{|\mathcal{A}_1| \times ... \times |\mathcal{A}_N|}$ is the space of payoffs for player $p$. We can include the approximation parameter (see Equation~\eqref{eq:wsce_def}) for each player $\epsilon = (\epsilon_1, ..., \epsilon_N) \in \mathbb{R}^N$ in the space of games, such that we have a new tuple space $(\mathcal{G}, \epsilon)$, that we call the approximate game space. We can define a distance metric between two approximate games $d((G^A, \epsilon^A), (G^B, \epsilon^B))$ such that the space of approximate normal-form games is a metric space. Let $\sigma \in \Delta^{|\mathcal{A}| - 1}$ be the probability simplex. Let $\sigma^* \in \text{(WS)(C)CE}(G, \epsilon) \subseteq \Delta^{|\mathcal{A}| - 1}$ be the subset of joints that are in equilibrium according to either Equation~\eqref{eq:wsce_def}, \eqref{eq:ce_def}, \eqref{eq:cce_def} or \eqref{eq:ne_def}. Most commonly we choose $\epsilon=0$ and therefore the additional notation around approximate games can be dropped.

\subsection{Equilibrium-Invariant Embedding}

We are concerned with studying game transforms, $(G, \epsilon) \to (\hat{G}, \hat{\epsilon})$, that do not change the set of approximate equilibria, $\text{(WS)(C)CE}(G, \epsilon) = \text{(WS)(C)CE}(\hat{G}, \hat{\epsilon})$. We call such a transformation an \emph{equilibrium-invariant} transform. The most common such transform \citep{ostrovski2013_fictitious_play_dynamics,morris2004_best_response_equivalence,moulin1978_cce} is the \emph{affine transform} which consists of an offset over the other players' strategies and a positive scale.
\begin{theorem}[Affine Transform] \label{thoerem:affine_transform}
$\epsilon$-NE, $\epsilon$-WSNE, $\epsilon$-CE, $\epsilon$-WSCE, and $\epsilon$-CCE are equilibrium-invariant under affine transformations of each player's payoff. Concretely, when
\begin{align}
    G_p(a) \to \hat{G}_p(a) = s_p G_p(a_p, a_{-p}) + b_p(a_{-p}), \quad \text{and} \quad \epsilon_p \to \hat{\epsilon}_p = s_p \epsilon_p,
\end{align}
an $\epsilon_p$-equilibrium in the original game is an $s_p\epsilon_p$-equilibrium in the transformed game: $\sigma(a) \to \hat{\sigma}(a) = \sigma(a)$, where $b_p(a_{-p})$ is any offset, and $s_p$ is any positive scalar.
\end{theorem}
\begin{proof}
Consider the effect of the transforms on the deviation gains.
\begin{subequations}
\begin{align}
    A^\text{WSCE}_p(a'_p, a''_p, a_{-p}) &\to s_p G_p(a'_p, a_{-p}) + b_p(a_{-p}) - s_p G_p(a''_p, a_{-p}) - b_p(a_{-p}) & &= s_p A^\text{WSCE}_p(a'_p, a''_p, a_{-p}) \\
    A^\text{CE}_p(a'_p, a''_p, a) &\to \begin{cases}
        s_p A^\text{WSCE}_p(a'_p, a''_p, a_{-p}) & a_p = a''_p \\
        0 & \text{otherwise}
    \end{cases} & &= s_p A^\text{CE}_p(a'_p, a''_p, a) \\
    A^\text{CCE}_p(a'_p, a) &\to s_p G_p(a'_p, a_{-p}) + b_p(a_{-p}) - s_p G_p(a) - b_p(a_{-p}) & &= s_p A^\text{CCE}_p(a'_p, a)
\end{align}
\end{subequations}
Equilibria are entirely defined by their inequality constraints (Equations~\eqref{eq:wsce_def}, \eqref{eq:ce_def}, \eqref{eq:cce_def} and \eqref{eq:ne_def}). The affine transform only results in a $s_p$ scale to the LHS of the inequality. If we apply the same positive scale to the RHS of the definition the inequality will still hold. Therefore an $\epsilon_p$-equilibrium in the untransformed game will be an $s_p\epsilon_p$-equilibrium in the transformed game. If $\epsilon_p=0$ the equilibria will not change.
\end{proof}

The affine transform can be used to reduce the degrees of freedom in each player's payoff by $|\mathcal{A}_{-p}| + 1$ without changing the equilibria if $b_p$ and $s_p$ are defined in terms of the payoffs themselves. Offsetting by the mean, $b_p(a_{-p}) = -\frac{1}{|\mathcal{A}_p|} \sum_{a_p} G_p(a_p, a_{-p})$, and scaling by the Frobenius norm of the payoff tensor, $s_p = 1 / ||G_p||_F$, are sensible choices. We can use this transform to design a distance metric, $d^\text{equil}$, and a distance-preserving equilibrium-invariant embedding. Let $\mathcal{G}^\text{equil}$ be an embedding in $\mathcal{G}$ given by a structure preserving mapping, such that $\mathcal{G}^\text{equil} \subset \mathcal{G}$. $\mathcal{G}^\text{equil}$ is a manifold and a metric space.
\begin{definition}[Equilibrium-Invariant Embedding]
    \begin{align} \label{eq:inv_embedding_function}
        G^\text{equil}_p(a) = \frac{1}{Z} \left(G_p(a) - \frac{1}{|\mathcal{A}_p|} \sum_{a_p} G_p(a_p, a_{-p}) \right) %
        \qquad %
        \epsilon^\text{equil}_p =  \frac{1}{Z}\epsilon_p %
        \qquad %
        Z = \left\| G_p - \frac{1}{|\mathcal{A}_p|} \sum_{a_p} G_p(a_p, a_{-p}) \right\|_F %
    \end{align}
\end{definition}
\begin{definition}[Equilibrium-Invariant Distance Metric]
    \begin{align} \label{eq:distance_metric}
        d^\text{equil} \left((G^A, \epsilon^A), (G^B, \epsilon^B) \right) %
        = \sqrt{ \sum_p \arccos \left(\sum_a G^\text{A,equil}_p(a) G^\text{B,equil}_p(a) \right)^2} +  \sqrt{\sum_p (\epsilon^\text{A,equil}_p - \epsilon^\text{B,equil}_p)^2}
    \end{align}
\end{definition}

This definition poses a problem: it is not well defined when a payoff is \emph{trivial} because the normalization constant $Z$ is zero for a trivial game. All joint distributions are in equilibrium in games where all players have trivial payoffs.
\begin{definition}[Trivial Payoff] \label{def:trivial_payoff}
    \begin{align}
        G_p(a_p, a_{-p}) = b_p(a_{-p}) \qquad \forall a_{-p} \in \mathcal{A}_{-p}
    \end{align}
\end{definition}
\begin{definition}[Nontrivial Payoff] \label{def:nontrivial_payoff}
    \begin{align}
        G_p(a_p, a_{-p}) \neq b_p(a_{-p}) \qquad \exists a_{-p} \in \mathcal{A}_{-p}
    \end{align}
\end{definition}
This can be simply remedied by defining the norm of a zero vector to be unity, $\|\boldsymbol{0}\| = 1$, which occurs when the payoff is trivial. For the embedding this is natural. For distances, the inner product, $\sum_a G^\text{A,equil}_p(a) G^\text{B,equil}_p(a)$, will be zero which is equivalent to the the payoffs being perpendicular. A trivial payoff would be defined to be perpendicular to all other payoffs. This is not an unreasonable definition, but in our work it is most common to only deal with nontrivial games, or only compare games with identical triviality structure, which we will assume going forward.

The equilibrium-invariant embedding is an oblique manifold~\citep{trendafilov2002multimode} (a product manifold of $N$ unit spheres). Each player's payoff embedding, $G^\text{equil}_p$, is a point on the surface of one of these spheres. The distance between two games is the norm of the arc lengths between the two points on each of these spheres. Therefore the maximum distance between two $\epsilon=0$ games is $\sqrt{N}\pi$. The equilibrium-invariant embedding reduces the degrees of freedom in each player's payoff by $|\mathcal{A}_{-p}| + 1$. The linear offset component contributes a $|\mathcal{A}_{-p}|$ portion, while the nonlinear scaling contributes the remaining unit. It turns out that $|\mathcal{A}_{-p}|$ is the largest reduction that can the achieved by a linear function.
\begin{theorem}[Linear Offset Rank Reduction]
    The offset component of the equilibrium-invariant embedding reduces a payoff's degrees of freedom by $|\mathcal{A}_{-p}|$. This is the most degrees of freedom that can be reduced with a linear transform without changing the equilibrium.
\end{theorem}
\begin{proof}
    The computation of deviation gains, $A$, (Equations~\eqref{eq:wsce_def}, \eqref{eq:ce_def}, and \eqref{eq:cce_def}) is a linear operation and therefore can be expressed as a matrix multiplication of an operator matrix, $T_p$, with a player's payoff, $G_p$. We use flattened forms of the payoffs and gains.
    \begin{subequations}
    \begin{align}
        A^\text{WSCE}_p(a'_p \otimes a''_p, a_{-p}) &= \sum_{a'''} T^\text{WSCE}_p(a'_p \otimes a''_p \otimes a_{-p}, a''') G_p(a''') \\
        A^\text{CE}_p(a'_p \otimes a''_p \otimes a) &= \sum_{a'''} T^\text{CE}_p(a'_p \otimes a''_p \otimes a, a''') G_p(a''') \\
        A^\text{CCE}_p(a'_p \otimes a) &= \sum_{a'''} T^\text{CCE}_p(a'_p \otimes a, a''') G_p(a''')
    \end{align}
    \end{subequations}
    By inspecting the structure of the operator matrices (Section~\ref{sec:gains_as_linear_operators}, Lemma~\ref{lemma:deviation_gain_rank}), we can calculate their rank.
    \begin{align}
        \text{rank}\left(T^\text{WSCE}_p\right) = \text{rank}\left(T^\text{CE}_p\right) = \text{rank}\left(T^\text{CCE}_p\right) = |\mathcal{A}| - |\mathcal{A}_{-p}| \qquad \forall p \in [1, N]
    \end{align}
    In general, the payoff, $G_p$, can be full rank, $\mathbb{R}^{|\mathcal{A}|}$. But after matrix multiplying with the operator matrix, which is only rank $|\mathcal{A}| - |\mathcal{A}_{-p}|$, the resulting deviation gains can be at most rank $|\mathcal{A}| - |\mathcal{A}_{-p}|$. Any linear equilibrium-invariant transform that reduces the space of games by more than $|\mathcal{A}_{-p}|$ would imply an operator matrix $T_p$ with rank less than $|\mathcal{A}| - |\mathcal{A}_{-p}|$, thereby reducing the number of inequality constraints that define the equilibrium. Therefore, any linear equilibrium-invariant transform can only reduce the space of games by at most $|\mathcal{A}_{-p}|$ without changing the set of equilibria for all games. The offset component of the affine transform reduces the degrees of freedom by $|\mathcal{A}_{-p}|$ and, by Theorem~\ref{thoerem:affine_transform}, does not change the deviation gains.
\end{proof}

\paragraph{Reversible Deviation Gains}
In general, the mapping from payoffs to deviation gains is irreversible because it is not a full-rank linear operation: there are many possible games that result in the same deviation gains. However, in the equilibrium-invariant embedding, there is a one-to-one mapping, and therefore it is possible to reverse the procedure and find the invariant embedding from the deviation gains.
\begin{theorem}[Reversible Deviation Gains] \label{theorem:reversible_gains}
    The equilibrium-invariant embedding can be recovered from the deviation gains.
    \begin{align}
        G^\text{equil}_p(a) %
        &= -\frac{1}{|\mathcal{A}_p|} \sum_{a'_p} A^\text{WSCE}_p(a'_p, a_p, a_{-p}) %
        = -\frac{1}{|\mathcal{A}_p|} \sum_{a'_p, a''_p} A^\text{CE}_p(a'_p, a''_p, a) %
        = -\frac{1}{|\mathcal{A}_p|} \sum_{a'_p} A^\text{CCE}_p(a'_p, a)
    \end{align}
\end{theorem}
\begin{proof}
    Recall the definition of the CCE deviation gains (Equation~\eqref{eq:cce_def}). Take the mean over the player's own deviation strategies, $a'_p \in \mathcal{A}_p$, and rearrange.
    \begin{align}
        A^\text{CCE}_p(a'_p, a) &= G_p(a'_p, a_{-p}) - G_p(a) \implies &
        G_p(a) &=  \underbrace{\frac{1}{|\mathcal{A}_p|}\sum_{a'_p} G_p(a'_p, a_{-p})}_{\substack{\text{Zero when zero-mean}}} - \frac{1}{|\mathcal{A}_p|}\sum_{a'_p} A^\text{CCE}_p(a'_p, a)
    \end{align}
    When the payoffs are invariant embeddings, $G^\text{equil}_p(a)$, the mean term is zero, and the proof is concluded for CCEs. Noting that $A^\text{CCE}_p(a'_p, a) = \sum_{a''_p \in \mathcal{A}_p} A^\text{CE}_p(a'_p, a''_p, a)$, we have a similar solution for CEs. Noting that $A^\text{WSCE}_p(a'_p, a''_p, a_{-p}) = \sum_{a_p \in \mathcal{A}_p} A^\text{CE}_p(a'_p, a''_p, a)$, we have a similar solution for WSCEs.
\end{proof}
This result is not directly utilized in this paper but we include it for several reasons. Firstly, it improves intuition about the space of games and how payoffs relate to the definition of equilibria and their constraints. Secondly, it provides further evidence of the fundamental nature of the equilibrium-invariant embedding. Finally, it opens up a path of research in inverse game theory, which studies the recovery of payoffs from observed equilibrium behaviour.

\paragraph{Sampling Equilibrium-Invariant Embedding}
It is easy to uniformly sample over the invariant embedding (Algorithm~\ref{alg:invariant_embedding_sample}) and trivial games (Algorithm~\ref{alg:trivial_sample}), where $\mathcal{N}(0, 1)$ is a zero-mean unit-variance normal distribution. This sampling approach is principled because it samples games that cover all interesting strategic interactions. Less principled ways of sampling (for example a uniform distribution over entries in the payoff, $G_p(a) = \mathcal{U}(0, 1)$), are common in the literature, but may not cover the strategic space of games evenly. We propose that the equilibrium-invariant embedding should be used for evaluating equilibrium solvers and producing training and testing datasets.

\begin{figure}[t]
\noindent\begin{minipage}[t]{\textwidth}
    \vspace{0pt}  
    \centering
    \begin{minipage}[t]{.49\textwidth}
        \vspace{0pt}  
        \begin{algorithm}[H]
            \caption{Equilibrium-Invariant Embedding Sampling}\label{alg:invariant_embedding_sample}
            \KwData{$|\mathcal{A}_1|, ..., |\mathcal{A}_N|$}
            \KwResult{$G^\text{equil}$}
            $N \gets \text{len}(|\mathcal{A}_1|, ..., |\mathcal{A}_N|)$\;
            \For{$p \gets [1,N]$}{
                $\text{assert}~~|\mathcal{A}_p| \geq 2$\;
                $G_p(a) \gets \mathcal{N}(0, 1)$\;
                $G_p(a) \gets G_p(a) - \frac{1}{|\mathcal{A}_p|} \sum_{a_p \in \mathcal{A}_p} G_p(a_p, a_{-p})$\;
                $G_p(a) \gets \frac{G_p(a)}{||G_p(a)||_2}$\;
            }
            $G^\text{equil} \gets \text{concat}(G_1, ..., G_N)$\;
        \end{algorithm}
    \end{minipage} \hfill
    \begin{minipage}[t]{.49\textwidth}
        \vspace{0pt}  
        \begin{algorithm}[H]
            \caption{Trivial Embedding Sampling}\label{alg:trivial_sample}
            \KwData{$|\mathcal{A}_1|, ..., |\mathcal{A}_N|$}
            \KwResult{$G^\text{trivial}$}
            $N \gets \text{len}(|\mathcal{A}_1|, ..., |\mathcal{A}_N|)$\;
            \For{$p \gets [1,N]$}{
                $\text{assert}~~|\mathcal{A}_p| \geq 2$\;
                $b_p(a_{-p}) \gets \mathcal{N}(0, 1)$\;
                $G_p(a) \gets b_p(a_{-p})$\;
                $G_p(a) \gets \frac{G_p(a)}{||G_p(a)||_2}$\;
            }
            $G^\text{tri} \gets \text{concat}(G_1, ..., G_N)$\;
        \end{algorithm}
    \end{minipage}
\end{minipage}
\end{figure}

\subsection{Better-Response Embedding}

Previously we showed that payoff scaling (within the affine game transform) is equilibrium-invariant, here we show that a novel per-strategy-scale transform is better-response-invariant\footnote{Better-response invariance implies best-response invariance.}. This transform results in reciprocal scaled corresponding equilibrium in the transformed game for $\epsilon$-NEs, $\epsilon$-WSNEs, $\epsilon$-CEs, and $\epsilon$-WSCEs.
\begin{definition}[Best-Response-Invariant]
    \begin{align}
        \arg \max_{a_p} G_p(a_p, a_{-p}) = \arg \max_{a_p} \hat{G}_p(a_p, a_{-p}) \quad \forall a_{-p} \in \mathcal{A}_{-p}
    \end{align}
\end{definition}
\begin{definition}[Better-Response-Invariant\footnote{Technically, we require $\arg \sort$ to split ties between payoffs by action index.}]
    \begin{align}
        \arg \sort_{a_p} G_p(a_p, a_{-p}) = \arg \sort_{a_p} \hat{G}_p(a_p, a_{-p}) \quad \forall a_{-p} \in \mathcal{A}_{-p}
    \end{align}
\end{definition}
\begin{theorem}[Per-Strategy-Scale Transform] \label{thm:per_strategy_scale_transform}
$\epsilon$-NE, $\epsilon$-WSNE, $\epsilon$-CE, $\epsilon$-WSCE are better-response-invariant under positive per-strategy-scale of each player's payoff which results in reciprocal per-strategy-scale ($s_p \rightarrow s_p(a_p)$) of the equilibria. Concretely, when
\begin{align}
    G_p(a) \to \hat{G}_p(a) = \left(\otimes_{q \in -p} s_q(a_q) \right) G_p(a), \text{ and } \epsilon^{WSCE}_p \to \hat{\epsilon}^{WSCE}_p(\sigma(a)) = \frac{\epsilon^{WSCE}_p}{Z_{-p}(a''_p)} \text{ or } \epsilon^\text{CE}_p \to \hat{\epsilon}^\text{CE}_p(\sigma(a)) = \frac{\epsilon^\text{CE}_p}{Z s_p(a''_p)},
\end{align}
an equilibrium in the original game has a corresponding equilibrium,
\begin{align}
    \sigma(a) \to \hat{\sigma}(a) = \frac{1}{Z \left(\otimes_{p} s_p(a_p) \right)} \sigma(a),
\end{align}
in the transformed game, where $Z = \sum_{a \in \mathcal{A}}\frac{\sigma(a)}{\otimes_p s_p(a_p)}$ and $Z_{-p}(a''_p) = \sum_{a_{-p}} \frac{\sigma(a''_p,a_{-p})}{\otimes_{-p} s_p(a_p) }$.
\end{theorem}
\begin{proof}
Consider the effect of the transforms on the deviation gains.
\begin{subequations}
\begin{align}
    A^\text{WSCE}_p(a'_p, a''_p, a_{-p}) &\to & \hat{A}^\text{WSCE}_p(a'_p, a''_p, a_{-p}) &= \left( \otimes_{q \in -p} s_q(a_q) \right) A^\text{WSCE}_p(a'_p, a''_p, a_{-p}) \\
    A^\text{CE}_p(a'_p, a''_p, a) &\to & \hat{A}^\text{CE}_p(a'_p, a''_p, a) &= \left( \otimes_{q \in -p} s_q(a_q) \right) A^\text{CE}_p(a'_p, a''_p, a)
\end{align}
\end{subequations}

Substitute the transformed deviation gains and approximations into the definition of CE (Equation~\ref{eq:ce_def}). This can be shown to be equivalent to a definition consisting of the untransformed game.
\begin{subequations}
\begin{align}
    \smashoperator{\sum_{a \in \mathcal{A}}} \hat{\sigma}(a) \hat{A}^\text{CE}_p(a'_p, a''_p, a) &\leq \hat{\epsilon}_p & &\forall p \in [1,N], a''_p \neq a'_p  \in \mathcal{A}_p \\
    \smashoperator{\sum_{a \in \mathcal{A}}} \left[\frac{1}{Z \left(\otimes_{p} s_p(a_p) \right)} \sigma(a) \right] \left[ \left( \otimes_{q \in -p} s_q(a_q) \right) A^\text{CE}_p(a'_p, a''_p, a) \right] &\leq \frac{1}{Z s_p(a''_p)}\epsilon_p & &\forall p \in [1,N], a''_p \neq a'_p  \in \mathcal{A}_p \displaybreak[1] \\
    \smashoperator{\sum_{a \in \mathcal{A}}} \frac{1}{Z s_p(a_p)} \sigma(a) A^\text{CE}_p(a'_p, a''_p, a) &\leq \frac{1}{Z s_p(a''_p)}\epsilon_p & &\forall p \in [1,N], a''_p \neq a'_p  \in \mathcal{A}_p \\
    \smashoperator{\sum_{a \in \mathcal{A}}} \sigma(a) A^\text{CE}_p(a'_p, a''_p, a) &\leq \epsilon_p & &\forall p \in [1,N], a''_p \neq a'_p  \in \mathcal{A}_p
\end{align}
\end{subequations}
The $s_p(a_p)$ and $s_p(a''_p)$ terms cancel by substituting the $a_p$ variable for $a''_p$ in the LHS, which is permitted by checking the definition of $A^\text{CE}(a'_p,a''_p,a)$ (Equation~\ref{eq:ce_def}). Therefore we have proved that if $\sigma(a)$ is an equilibrium in the untransformed game, $\hat{\sigma}(a)$ is an equilibrium in the transformed game, and $\hat{\sigma}(a)$ can be calculated directly from $\sigma(a)$ and $s_p(a_p)$. Now consider the WSCE definition (Equation~\ref{eq:wsce_def}).
\begin{subequations}
\begin{align}
\smashoperator{\sum_{a_{-p} \in \mathcal{A}_{-p}}} \hat{\sigma}(a_{-p}|a''_p) \hat{A}^\text{WSCE}_p(a'_p, a''_p, a_{-p}) &\leq \hat{\epsilon}_p & &\forall p \in [1,N], a''_p \neq a'_p  \in \mathcal{A}_p \\
\smashoperator{\sum_{a_{-p} \in \mathcal{A}_{-p}}} \hat{\sigma}(a''_p,a_{-p}) \hat{A}^\text{WSCE}_p(a'_p, a''_p, a_{-p}) &\leq \hat{\sigma}(a''_p) \hat{\epsilon}_p & &\forall p \in [1,N], a''_p \neq a'_p  \in \mathcal{A}_p \displaybreak[1] \\
\smashoperator{\sum_{a_{-p} \in \mathcal{A}_{-p}}} \frac{1}{Z s_p(a''_p)} \sigma(a''_p, a_{-p}) A^\text{WSCE}_p(a'_p, a''_p, a_{-p}) &\leq \left( \sum_{a_{-p}} \frac{\sigma(a''_p,a_{-p})}{Z s_p(a''_p) \left(\otimes_{-p} s_p(a_p) \right)} \right) \hat{\epsilon}_p & &\forall p \in [1,N], a''_p \neq a'_p  \in \mathcal{A}_p \\
\smashoperator{\sum_{a_{-p} \in \mathcal{A}_{-p}}} \sigma(a''_p, a_{-p}) A^\text{WSCE}_p(a'_p, a''_p, a_{-p}) &\leq \epsilon_p & &\forall p \in [1,N], a''_p \neq a'_p  \in \mathcal{A}_p
\end{align}
\end{subequations}
The definition of NE and WSNE is the same as CE and WSCE, but with the additional constraint that the joint factorizes: $\sigma(a) = \otimes_p \sigma_p(a_p)$. This additional constraint does not affect the proof above, so the result also holds for NE and WSNE. The payoffs are only scaled over other players' strategies, and these scales are positive, which implies better-response-invariance.
\end{proof}

Notably, Theorem~\ref{thm:per_strategy_scale_transform} does not hold for CCEs.
\begin{theorem}[Per-strategy-scaling CCE counterexample] \label{thm:per_strategy_scale_transform_cce_counterexample}
    Per-strategy-scaling of each player's payoff does not result in reciprocal per-strategy-scaling of CCEs.
\end{theorem}
\begin{proof}
    Consider the two-player three-strategy game, scaled-rock-paper-scissors, $(G_1, G_2)$, which has a CCE (amongst others) at $\sigma$ defined below.
    \begin{align}
        G_1 = \begin{bmatrix}
            0 & 4 & -1 \\
            -2 & 0 & 1 \\
            2 & -4 & 0
        \end{bmatrix} \quad %
        G_2 = \begin{bmatrix}
            0 & -1 & 1 \\
            1 & 0 & -1 \\
            -1 & 1 & 0
        \end{bmatrix} \quad %
        \sigma = \begin{bmatrix}
            0 & \frac{1}{4} & \frac{1}{12} \\
            \frac{1}{4} & 0 & \frac{1}{12} \\
            0 & 0 & \frac{1}{3}
        \end{bmatrix}
    \end{align}
    Now lets scale this game by $s_1 = [1, 1, 1]$ and $s_2 = [\frac{1}{2}, \frac{1}{4}, 1]$ to arrive at the familiar rock-paper-scissors game, $(\hat{G}_1, G_2)$.
    \begin{align}
        \hat{G}_1 = \begin{bmatrix}
            0 & 1 & -1 \\
            -1 & 0 & 1 \\
            1 & -1 & 0
        \end{bmatrix} \quad %
        G_2 = \begin{bmatrix}
            0 & -1 & 1 \\
            1 & 0 & -1 \\
            -1 & 1 & 0
        \end{bmatrix} \quad %
        \hat{\sigma} = \frac{1}{Z} \begin{bmatrix}
            0 & 1 & \frac{1}{12} \\
            \frac{1}{2} & 0 & \frac{1}{12} \\
            0 & 0 & \frac{1}{3}
        \end{bmatrix}
    \end{align}
    Note that the scaled version of $\sigma$, $\hat{\sigma}$, which is calculated according to Theorem~\ref{thm:per_strategy_scale_transform}, is not an equilibrium in the scaled game. Therefore, by counterexample, Theorem~\ref{thm:per_strategy_scale_transform} does not hold for CCEs in general.
\end{proof}

Fortunately, for two-strategy games, CEs and CCEs are equivalent \citep{monnot2017limits}. We make use of this property when deriving embeddings for 2×2 games.
\begin{remark}[Two-Strategy (C)CE Equivalence] \label{remark:two_strategy_cce}
    For n-player games with two strategies per player, $\epsilon$-CCEs and $\epsilon$-CEs are equivalent.
\end{remark}

Embedding and distance metrics for n-player games can be readily defined using the per-strategy-scale transform. However for two-player games, using zero approximation $\epsilon=0$, there is a natural embedding definition.
\begin{definition}[Two-Player Better-Response-Invariant Embedding] \label{def:two_player_better_response_invariant_embedding}
    \begin{subequations}
    \begin{align} \label{eq:two_player_better_response_invariant_embedding}
        G^\text{res}_1(a_1, a_2) = \frac{1}{Z(a_2)} \left( G_1(a_1, a_2) - \frac{1}{|\mathcal{A}_1|} \sum_{a_1} G_p(a_1, a_2) \right) \qquad Z(a_2) = \left\| G_1(:, a_2) - \frac{1}{|\mathcal{A}_1|} \sum_{a_1} G_1(a_1, a_2) \right\|_F \\
        G^\text{res}_2(a_1, a_2) = \frac{1}{Z(a_1)} \left( G_1(a_1, a_2) - \frac{1}{|\mathcal{A}_2|} \sum_{a_2} G_p(a_1, a_2) \right) \qquad Z(a_1) = \left\| G_1(a_1, :) - \frac{1}{|\mathcal{A}_2|} \sum_{a_2} G_1(a_1, a_2) \right\|_F
    \end{align}
    \end{subequations}
\end{definition}
\begin{definition}[Two-Player Better-Response-Invariant Distance Metric] \label{def:two_player_better_response_invariant_distance_metric}
    \begin{align} \label{eq:two_player_better_response_invariant_distance_metric}
        d^\text{res} \left(G^A, G^B \right) = \sqrt{\sum_{a_2} \arccos \left(\sum_{a_1} G^\text{A,res}_1(a_1, a_2) G^\text{B,res}_p(a_1, a_2)) \right)^2 %
        + %
        \sum_{a_1} \arccos \left(\sum_{a_2} G^\text{A,res}_1(a_1, a_2) G^\text{B,res}_p(a_1, a_2)) \right)^2}
    \end{align}
\end{definition}
Therefore, the better-response-invariant game embedding is an oblique manifold (a product manifold of $|\mathcal{A}_1| + |\mathcal{A}_2|$ unit spheres). For each other player's strategy for each player's payoff, the slice of the embedding is a point on the surface of one of these spheres. The distance between two games is the norm of the arc lengths between the two points on each of these spheres. Therefore the maximum distance between two games is $\sqrt{|\mathcal{A}_1| + |\mathcal{A}_2|}\pi$.

\subsection{Symmetric Embedding}

The order of the strategies in a normal-form game is arbitrary, therefore games identical up to strategy permutations could be considered equal. Furthermore, if the role of the player is not important\footnote{Role may be important if each player has access to different strategies or different numbers of strategies.}, identity up to player permutation can also be considered. If we define a canonical ordering of the strategies, then we reduce the area of the equilibrium-invariant embedding that we need to consider. One such ordering could be defined as follows. Firstly, for each player, $p$, independently sort the elements over all other players strategies, $a_{-p}$, and then lexicographically sort over the player's own strategies, $a_p$, to obtain an order permutation, $\tau^*_p(a_p)$. Secondly, to get the order of players, sort each player's whole payoff and then lexicographically sort over players to get a player order permutation, $\omega^*(p)$.
\begin{subequations}
\begin{align}
    G'_p(a_p, a_{-p}) &= \sort_{a_{-p}} G_p(a_p, a_{-p}) \quad \forall p & \tau^*_p(a_p) &= \arg \lexsort_{a_p} G'_p(a_p, a_{-p}) \quad \forall p \\
    G''_p(a) &= \sort_{a} G_p(a) \quad \forall p & \omega^*(p) &= \arg \lexsort_p G''_p(a)
\end{align}
\end{subequations}
Partial orderings occur when strategies have equal payoff, therefore even if a permutation is only partially ordered, the resulting payoffs will be unique. These permutations can be used to define another embedding: the \emph{symmetric game embedding}.
\begin{definition}[Equilibrium-Symmetric Game Embedding]
    \begin{align} \label{eq:sym_embedding_function}
        G^\text{sym}_p(a) = G_{\omega^*(p)}(\tau^*_{\omega^*(p)}(a_{\omega^*(p)}), ...)%
        \qquad %
        \epsilon^\text{sym}_p = \hat{\epsilon}_{\omega^*(p)}
    \end{align}
\end{definition}

Symmetries do not reduce the number of degrees of freedom, but do reduce the volume of games by exploiting symmetry in their definitions. There are $\prod_p \left(|\mathcal{A}_p|! \right)$ such strategy permutation symmetries and $N!$ player permutation symmetries in a normal form game. These symmetries should be used in conjunction with either equilibrium-invariant or better-response-invariant embeddings.
\begin{definition}[Equilibrium-Symmetric Distance Metric]
    \begin{align} \label{eq:symmetric_distance_metric}
        d^\text{sym} \left((G^A, \epsilon^A), (G^B, \epsilon^B) \right) = \min_{\tau_p(a_p), \omega(p)} \sqrt{ \sum_p \arccos \left(\sum_a G^\text{A,equil}_{\omega(p)}(\tau_{\omega(p)}(a_{\omega(p)}), ...) G^\text{B,equil}_p(a)\right)^2} +  \sqrt{\sum_p \left(\epsilon^\text{A,equil}_{\omega(p)} - \epsilon^\text{B,equil}_p \right)^2}
    \end{align}
\end{definition}

\section{2×2 Equilibrium-Invariant Embedding}

We now focus on the metric spaces of 2×2 games.

\subsection{Deriving the 2×2 Equilibrium-Invariant Embedding}

2×2 games have a particularly efficient equilibrium-invariant embedding parameterized by only two variables. This embedding has a distance metric equivalent to the one defined above, so is also a metric space over 2×2 games.

\begin{theorem}[2×2 Equilibrium-Invariant Embedding] \label{thm:equilibrium_invariant_embedding}
All nontrivial 2×2 game payoff matrices, $G_p$, can be mapped to payoff matrices, $G^\text{equil}_p$, parameterized by only two variables, without altering the equilibria of the game. The mapping is given by Equations~\eqref{eq:inv_embedding_p1} and \eqref{eq:inv_embedding_p2}, where $\theta_1 + \frac{\pi}{4} = \text{arctan2}(g_1^{AA}-g_1^{BA}, g_1^{AB} - g_1^{BB})$ and $\theta_2 + \frac{\pi}{4} = \text{arctan2}(g_2^{AA}-g_2^{AB}, g_2^{BA} - g_2^{BB})$.
\begin{subequations}
\begin{align}
    G_1 =& \begin{bmatrix} \label{eq:inv_embedding_p1}
        g_1^{AA} & g_1^{AB} \\
        g_1^{BA} & g_1^{BB} \\
    \end{bmatrix}
    \to G^\text{equil}_1(\theta_1) = \begin{bmatrix}
        \phantom{+} \frac{1}{\sqrt{2}} \sin(\theta_1 + \frac{\pi}{4}) & \phantom{+} \frac{1}{\sqrt{2}} \cos(\theta_1 + \frac{\pi}{4}) \\
        - \frac{1}{\sqrt{2}} \sin(\theta_1 + \frac{\pi}{4}) & - \frac{1}{\sqrt{2}} \cos(\theta_1 + \frac{\pi}{4}) \\
    \end{bmatrix} \\
    G_2 =& \begin{bmatrix} \label{eq:inv_embedding_p2}
        g_2^{AA} & g_2^{AB} \\
        g_2^{BA} & g_2^{BB} \\
    \end{bmatrix}
    \to G^\text{equil}_2(\theta_2) = \begin{bmatrix}
        \phantom{+} \frac{1}{\sqrt{2}} \sin(\theta_2 + \frac{\pi}{4}) & - \frac{1}{\sqrt{2}} \sin(\theta_2 + \frac{\pi}{4}) \\
        \phantom{+} \frac{1}{\sqrt{2}} \cos(\theta_2 + \frac{\pi}{4}) & - \frac{1}{\sqrt{2}} \cos(\theta_2 + \frac{\pi}{4}) \\
    \end{bmatrix}
\end{align}
\end{subequations}
\end{theorem}

\begin{proof}
Consider the payoff for player 1, $G_1$. First apply the offset of the affine invariant transformation, $G_1(a_1, a_2) \to \bar{G}_1(a_1, a_2) = G_1(a_1, a_2) + b_1(a_2)$, where the columns (other player strategies) of player 1's payoff matrix are normalized to zero-mean offset, $b_1(a_2) = - \frac{1}{|\mathcal{A}_1|} \sum_{a_1 \in \mathcal{A}_1} G_1(a_1, a_2)$, such that $b_1 = [-\frac{1}{2} g_1^{AA} - \frac{1}{2} g_1^{BA}, -\frac{1}{2} g_1^{AB} -\frac{1}{2} g_1^{BB}]$. Then, make a variable substitution with $g_1^{A-B,A} = g_1^{AA} - g_1^{BA}$ and $g_1^{A-B,B} = g_1^{AB} - g_1^{BB}$, which reduces the number of parameters needed to describe the payoff from $4 \to 2$.
\begin{align}
    \begin{bmatrix}
        g_1^{AA} & g_1^{AB} \\
        g_1^{BA} & g_1^{BB} \\
    \end{bmatrix}
    \to \bar{G}_1 = %
    \begin{bmatrix}
        \frac{1}{2} (g_1^{AA} - g_1^{BA}) & \frac{1}{2} (g_1^{AB} - g_1^{BB}) \\
        \frac{1}{2} (g_1^{BA} - g_1^{AA})  & \frac{1}{2} (g_1^{BB} - g_1^{BA}) \\
    \end{bmatrix} %
    = \begin{bmatrix}
        \phantom{+}\frac{1}{2} g_1^{A-B,A} & \phantom{+}\frac{1}{2} g_1^{A-B,B} \\
        -\frac{1}{2} g_1^{A-B,A} & -\frac{1}{2} g_1^{A-B,B} \\
    \end{bmatrix}
\end{align}

Now apply the scale of the affine invariant transform, a unit $L_2$ normalization, $\bar{G}_1(a_1, a_2) \to G^\text{equil}_1(a_1, a_2) = s_1 \bar{G}_1(a_1, a_2)$ where $s_1 = \frac{1}{||\bar{G}_1||_2}$, which ensures that the norm over all the elements in the payoff are equal to one. This is a valid transform as long as the payoff is nontrivial (nonzero after zero-mean offset).
\begin{align}
    \begin{bmatrix}
        \phantom{+}\frac{1}{2} g_1^{A-B,A} & \phantom{+}\frac{1}{2} g_1^{A-B,B} \\
        -\frac{1}{2} g_1^{A-B,A} & -\frac{1}{2} g_1^{A-B,B} \\
    \end{bmatrix} \to \begin{bmatrix}
        \phantom{+} \frac{1}{\sqrt{2}} \frac{g_1^{A-B,A}}{\sqrt{{g_1^{A-B,A}}^2 + {g_1^{A-B,B}}^2}} & \phantom{+} \frac{1}{\sqrt{2}} \frac{g_1^{A-B,B}}{\sqrt{{g_1^{A-B,A}}^2 + {g_1^{A-B,B}}^2}} \\
        - \frac{1}{\sqrt{2}} \frac{g_1^{A-B,A}}{\sqrt{{g_1^{A-B,A}}^2 + {g_1^{A-B,B}}^2}} & - \frac{1}{\sqrt{2}} \frac{g_1^{A-B,B}}{\sqrt{{g_1^{A-B,A}}^2 + {g_1^{A-B,B}}^2}} \\
    \end{bmatrix}
\end{align}

Note that the $L_2$ norm of the elements of this payoff now equal unity, so we now have the property that:
\begin{align}
    \left( \frac{g_1^{A-B,A}}{\sqrt{{g_1^{A-B,A}}^2 + g_1^{A-B,B}}} \right)^2
    + \left( \frac{g_1^{A-B,B}}{\sqrt{{g_1^{A-B,A}}^2 + {g_1^{A-B,B}}^2}} \right)^2 = 1
\end{align}

This is the equation of a unit circle, and the opposite and adjacent parameters can be represented as a single angle parameter $\theta_1$, where $\frac{\pi}{4}$ is an arbitrary offset chosen for visualization purposes. This further reduces the number of parameters needed to describe player 1's payoff from $2 \to 1$.
\begin{align}
    \tan\left(\theta_1 + \frac{\pi}{4}\right) = \frac{g_1^{A-B,A}}{g_1^{A-B,B}} \implies \theta_1 + \frac{\pi}{4} = \text{arctan2}(g_1^{A-B,A}, g_1^{A-B,B})
\end{align}
The function $\text{arctan2}$ is defined by \cite{organick1988_atan2}. Since the radius of the unit circle is 1, we can recover the elements of the payoff directly from $\theta_1$.
\begin{align}
    G^\text{equil}_1(\theta_1) = \begin{bmatrix}
        \phantom{+} \frac{1}{\sqrt{2}} \sin(\theta_1 + \frac{\pi}{4}) & \phantom{+} \frac{1}{\sqrt{2}} \cos(\theta_1 + \frac{\pi}{4}) \\
        - \frac{1}{\sqrt{2}} \sin(\theta_1 + \frac{\pi}{4}) & - \frac{1}{\sqrt{2}} \cos(\theta_1 + \frac{\pi}{4}) \\
    \end{bmatrix}
\end{align}

An equivalent set of reductions can be used for player 2, where $\theta_2 + \frac{\pi}{4} = \text{arctan2}(g_2^{A,A-B}, g_2^{B,A-B})$.
\begin{align}
    G_2 =& \begin{bmatrix}
        g_2^{AA} & g_2^{AB} \\
        g_2^{BA} & g_2^{BB} \\
    \end{bmatrix}
    \to G^\text{equil}_2(\theta_2) = \begin{bmatrix}
        \phantom{+} \frac{1}{\sqrt{2}} \sin(\theta_2 + \frac{\pi}{4}) & - \frac{1}{\sqrt{2}} \sin(\theta_2 + \frac{\pi}{4}) \\
        \phantom{+} \frac{1}{\sqrt{2}} \cos(\theta_2 + \frac{\pi}{4}) & - \frac{1}{\sqrt{2}} \cos(\theta_2 + \frac{\pi}{4}) \\
    \end{bmatrix}
\end{align}
\end{proof}

Therefore all nontrivial 2×2 games can be parameterized using two variables $(\theta_1, \theta_2)$ which geometrically describe points via angles on two circles \footnote{Defined as the product manifold of two circles. Topologically this is a torus. Other classifications \citep{robinsonandgoforth2005_topology_of_2x2_games_book} also have a similar torus topology.}. Several interesting structural properties emerge (Figure~\ref{fig:representation_summary}), which will be explained in the following sections. Similarly, the partially trivial invariant embedding can be visualized in a single dimension. Partially trivial games, where one of the players does not have any influence in the game, could be described by a single parameter, either $\theta_1$ or $\theta_2$, depending on which player contributes to the strategic dynamics. Partially trivial games can therefore be mapped to the \emph{partially trivial equilibrium-invariant embedding manifold}. Trivial games where no player participates are mapped to a singleton set consisting of the \nullgame~Null game.

\begin{figure}[p!]
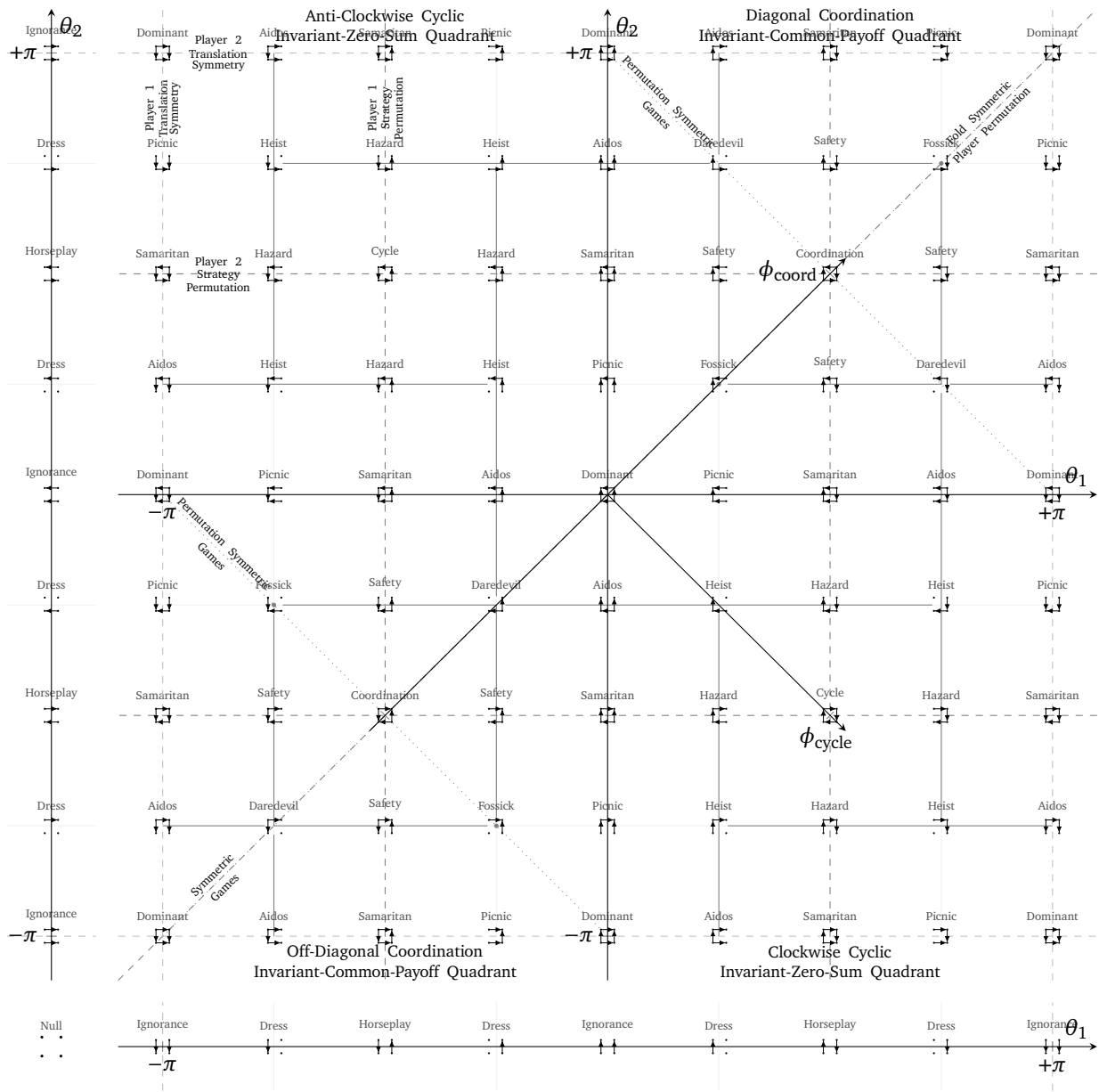

\centering
    \invarplot{}
    \caption{The 2×2 equilibrium-invariant embedding. Games can either be parameterized per-player, $\theta_1$ and $\theta_2$, or over game types, $\phi_{cycle}$ and $\phi_{coord}$. Symmetries (Table~\ref{tab:2x2_symmetries}) are indicated by dashed lines and allow reduction in area by a factor of up to $8$, depending on which symmetries are utilized. The strategy permutation equilibrium-symmetric embedding is north west patterned (\northwestpattern). The strategy and player permutation equilibrium-symmetric embedding is north east patterned (\northeastpattern). Permutation symmetry with zero-sum assumption is vertical patterned (\verticalpattern). The quadrants separate the space into cyclic and coordination regions. The top-left and bottom-right are cyclic (\cyclegame~Cycle), invariant-zero-sum regions. The top-right and bottom-left are coordination (\coordinationgame~Coordination), invariant-common-payoff regions. The dotted lines indicate symmetric games or symmetric games with permuted strategies. The light solid lines indicate equivalence class boundaries. The solid gray lines, points, and space between them indicate boundaries in the support of the equilibrium (see also Figure~\ref{fig:equilibrium_support}). The \dominantgame~Dominant, \picnicgame~Picnic, and \samaritangame~Samaritan have pure equilibria, the square regions are full-support (either \cyclegame~cycle or \coordinationgame~Coordination), the gray line boundaries have either two (\aidosgame~Aidos, \heistgame~Heist, \hazardgame~Hazard) or three (\safetygame~Safety, \daredevilgame~Daredevil) joint strategy support, and the gray points are diagonal or off-diagonal support (\fossickgame~Fossick). The 11 nontrivial best-response-invariant set of games is annotated with names suffixed with the equilibrium support. \dominantgame~Dominant, the only game that is both invariant-zero-sum and invariant-common-payoff, is at the origin. [Note: some PDF viewers incompletely render this figure.]}
    \label{fig:representation_summary}
\end{figure}

\subsection{Invariant-Zero-Sum and Invariant-Common-Payoff Quadrants}
\label{sec:pseudo_quadrants}

Two-player zero-sum games are well studied in the literature because they are easier to solve than their mixed-motive cousins. In particular for the NE, the problem can be expressed as a min-max optimization. If multiple equilibria exist, all equilibria have the same payoff, and they are known to be interchangeable (it does not matter which equilibrium the opponent chooses). But such games represent a small subset of general-sum games. If we could map a larger set of games onto the space of zero-sum games using invariant transformations, we could prove that they would share these interesting properties.

Two such approaches have already been explored in the literature. Firstly, \emph{constant-sum} games are scalar-offset transformations of zero-sum games, $G_p(a) \to \hat{G}_p(a) = G_p(a) + b_p$, which are well known to share zero-sum game's properties. Secondly, \emph{strictly-competitive} games \citep{adler2009_strict_comp_zero_sum} are scalar-offset and player-scale transformations of zero-sum games, $G_p(a) \to \hat{G}_p(a) = s_p G_p(a) + b_p$. However, there exists a larger space of games that can be mapped to zero-sum games.
\begin{definition}[Invariant-Zero-Sum Games]
    Invariant-zero-sum games are those that can be mapped to zero-sum games using equilibrium-invariant transformations, $G_p(a) \to \hat{G}_p(a) = s_p G_p(a) + b_p(a_{-p})$ such that $\sum_p \hat{G}_p(a) = 0 \,\, \forall a \in \mathcal{A}$.
\end{definition}
By definition, the equilibria of invariant-zero-sum games are identical to their zero-sum counterpart. Furthermore, the equilibria are still interchangeable in the transformed game, however not all equilibria will necessarily have the same payoffs. Other generalizations of zero-sum games including \emph{strategically zero-sum} \citep{moulin1978_cce}, \emph{best-response zero-sum} \citep{rosenthal1974_best_response_zero_sum}, and \emph{order zero-sum} \citep{shapley1963_order_zero_sum}, are all supersets of invariant zero-sum games. The diagonal quadrants (top-left and bottom-right) in the visualization (Figure~\ref{fig:representation_summary}) are \emph{invariant-zero-sum}.
\begin{theorem}[Invariant-Zero-Sum Quadrants]
    When $\sin(\theta_1) \sin(\theta_2) < 0$ in Equations~\eqref{eq:inv_embedding_p1} and \eqref{eq:inv_embedding_p2}, the respective game $(\hat{G}_1(\theta_1), \hat{G}_2(\theta_2))$ is invariant-zero-sum.
\end{theorem}
\begin{proof}
    For a game to be invariant-zero-sum there have to exist invariant transforms $s_2 > 0$, $b_1(a_2)$, and $b_2(a_1)$ such that $G_1(a_1, a_2) + b_1(a_2) = -s_2 G_2(a_1, a_2) - b_2(a_1) \,\, \forall a_1, a_2$. Consider the transforms $s_2 = -\frac{\sin(\theta_1)}{\sin(\theta_2)}$, $b_1 = [-\frac{1}{2} \frac{\sin(\theta_1)}{\tan(\theta_2)},\frac{1}{2} \frac{\sin(\theta_1)}{\tan(\theta_2)}]$, and $b_2 = [-\frac{1}{2}\cos(\theta_1),\frac{1}{2}\cos(\theta_1)]$, which result in payoffs:
    \begin{align}
        \hat{G}_1(\theta_1,\theta_2) = -\hat{G}_2(\theta_1,\theta_2) = \begin{bmatrix}
            \phantom{-} \frac{1}{2} \frac{\sin(\theta_1)}{\tan(\theta_2)} + \frac{1}{\sqrt{2}} \sin(\theta_1 + \frac{\pi}{4}) & -\frac{1}{2} \frac{\sin(\theta_1)}{\tan(\theta_2)} + \frac{1}{\sqrt{2}} \sin(\theta_1 + \frac{\pi}{4}) \\
            \phantom{-} \frac{1}{2} \frac{\sin(\theta_1)}{\tan(\theta_2)} - \frac{1}{\sqrt{2}} \sin(\theta_1 + \frac{\pi}{4}) & - \frac{1}{2} \frac{\sin(\theta_1)}{\tan(\theta_2)} - \frac{1}{\sqrt{2}} \sin(\theta_1 + \frac{\pi}{4}) \\
        \end{bmatrix}
    \end{align}
    This is only a valid transformation when $s_2 > 0 \implies \frac{\sin(\theta_1)}{\sin(\theta_2)} < 0 \implies \sin(\theta_1) \sin(\theta_2) < 0$.
\end{proof}

A similar larger set of cooperative common-payoff games, $G_p(a) = G_q(a)~\forall p, q \in [1, N]$, can also be derived using invariant transforms.
\begin{definition}[Invariant-Common-Payoff Games]
    Invariant-common-payoff games are those that can be mapped to common-payoff games using equilibrium-invariant transformations, $G_p(a) \to \hat{G}_p(a) = s_p G_p(a) + b_p(a_{-p})$.
\end{definition}
These games correspond exactly to the set of 2×2 weighted potential games~\citep{monderer1996potential}, where the common payoff acts as the potential function. The off-diagonal quadrants (top-right and bottom-left) in the visualization are \emph{invariant-common-payoff}).

\begin{theorem}[Invariant-Common-Payoff Quadrants]
    When $\sin(\theta_1) \sin(\theta_2) > 0$ in Equations~\eqref{eq:inv_embedding_p1} and \eqref{eq:inv_embedding_p2}, the respective game $(\hat{G}_1(\theta_1), \hat{G}_2(\theta_2))$ is invariant-common-payoff.
\end{theorem}
\begin{proof}
    For a game to be invariant-common-payoff there has to exist invariant transforms $s_2 > 0$, $b_1(a_2)$, and $b_2(a_1)$ such that $G_1(a_1, a_2) + b_1(a_2) = s_2 G_2(a_1, a_2) + b_2(a_1)$. Consider the transforms $s_2 = \frac{\sin(\theta_1)}{\sin(\theta_2)}$, $b_1 = [\frac{1}{2} \frac{\sin(\theta_1)}{\tan(\theta_2)},-\frac{1}{2} \frac{\sin(\theta_1)}{\tan(\theta_2)}]$, and $b_2 = [\frac{1}{2}\cos(\theta_1),-\frac{1}{2}\cos(\theta_1)]$, which result in payoffs:
    \begin{align}
        \hat{G}_1(\theta_1,\theta_2) = \hat{G}_2(\theta_1,\theta_2) = \begin{bmatrix}
            \phantom{-} \frac{1}{2} \frac{\sin(\theta_1)}{\tan(\theta_2)} + \frac{1}{\sqrt{2}} \sin(\theta_1 + \frac{\pi}{4}) & -\frac{1}{2} \frac{\sin(\theta_1)}{\tan(\theta_2)} + \frac{1}{\sqrt{2}} \sin(\theta_1 + \frac{\pi}{4}) \\
            \phantom{-} \frac{1}{2} \frac{\sin(\theta_1)}{\tan(\theta_2)} - \frac{1}{\sqrt{2}} \sin(\theta_1 + \frac{\pi}{4}) & - \frac{1}{2} \frac{\sin(\theta_1)}{\tan(\theta_2)} - \frac{1}{\sqrt{2}} \sin(\theta_1 + \frac{\pi}{4}) \\
        \end{bmatrix}
    \end{align}
    This is only a valid transformation when $s_2 > 0 \implies \frac{\sin(\theta_1)}{\sin(\theta_2)} > 0 \implies \sin(\theta_1) \sin(\theta_2) > 0$.
\end{proof}

The game at the origin (that we name \dominantgame~Dominant) is both invariant-zero-sum and invariant-common-payoff. Along with the permuted versions of this game, this is the only game with this property.
\begin{theorem}[Dominant Special Case]
    When $\sin(\theta_1) = 0$ and $\sin(\theta_2) = 0$ in Equations~\eqref{eq:inv_embedding_p1} and \eqref{eq:inv_embedding_p2}, the respective game $(\hat{G}_1(\theta_1), \hat{G}_2(\theta_2))$ is both invariant-zero-sum and invariant-common-payoff.
\end{theorem}
\begin{proof}
    Consider when $\theta_1 = \theta_2 = 0$. When using transforms $b_1 = [-\frac{1}{2}, \frac{1}{2}]$ and $b_2 = [-\frac{1}{2}, \frac{1}{2}]$ the game becomes zero-sum (Equation~\eqref{eq:pzs}). Additionally, when using transforms $b_1 = [\frac{1}{2}, -\frac{1}{2}]$ and $b_2 = [\frac{1}{2}, -\frac{1}{2}]$ the game becomes common-payoff (Equation~\eqref{eq:pcp}).
    
    \begin{subequations}
        \noindent\begin{minipage}[b]{0.48\textwidth}
            \begin{align} \label{eq:pzs}
                \hat{G}_1 = -\hat{G}_2 = \begin{bmatrix}
                    \phantom{+} 0 & \phantom{+} \frac{1}{2} \\
                    - \frac{1}{2} & \phantom{+} 1 \\
                \end{bmatrix}
            \end{align}
        \end{minipage}
        \hfill
        \begin{minipage}[b]{0.48\textwidth}
            \begin{align} \label{eq:pcp}
                \hat{G}_1 = \hat{G}_2 = \begin{bmatrix}
                    \phantom{+} 1 & \phantom{+}0 \\
                    \phantom{+} 0 & - 1 \\
                \end{bmatrix}
            \end{align}
        \end{minipage}
    \end{subequations}
\end{proof}

Therefore all symmetric 2×2 games are either invariant-zero-sum, invariant-common-payoff, or both. The borders between the quadrants are neither invariant-zero-sum or invariant-common-payoff. \dominantgame~Dominant is the only game which is both invariant-zero-sum and invariant-common-payoff.

\subsection{2×2 Equilibrium-Invariant Distance Metric}

\begin{theorem}[2×2 Equilibrium-Invariant Distance Metric]
    The distance metric between two 2×2 games is given by:
    \begin{align} \label{eq:2x2_distance_metric}
        d(\Theta^A, \Theta^B) = ||[\min(|\theta_1^A - \theta_1^B|, 2\pi - |\theta_1^A - \theta_1^B|), \min(|\theta_2^A - \theta_2^B|, 2\pi - |\theta_2^A - \theta_2^B|)]||_p
    \end{align}
\end{theorem}
\begin{proof}
Consider player 1's pre-norm component of the distance metric (Equation~\eqref{eq:distance_metric}) between two parameterized games $\hat{G}_1^A(\theta_1^A)$ and $\hat{G}_1^B(\theta_1^B)$ (Equation~\eqref{eq:inv_embedding_p1}).
\begin{subequations}
\begin{align}
    v_1 &= \text{arccos} \left( \sin(\theta^A_1 + \frac{\pi}{4}) \sin(\theta^B_1 + \frac{\pi}{4}) + \cos(\theta^A_1 + \frac{\pi}{4}) \cos(\theta^B_1 + \frac{\pi}{4}) \right) \\
    &= \text{arccos} \left( \cos(\theta^A_1 - \theta^B_1) \right) \\
    &= \min(|\theta_1^A - \theta_1^B|, 2\pi - |\theta_1^A - \theta_1^B|)
\end{align}
\end{subequations}
A similar calculation can be made for player 2. Together, this results in Equation~\eqref{eq:2x2_distance_metric}.
\end{proof}

This definition of a distance metric on the 2×2 equilibrium-invariant embedding of games is natural. As established in Theorem~\ref{thm:equilibrium_invariant_embedding}, games with the same equilibria can be embeddings concisely represented by points on two independent unit-circles or equivalently by their angles, i.e., $\Theta = (\theta_1, \theta_2)$. Distance along a unit-circle is measured by arc length. Therefore, the natural way to measure distance between two games is to sum the arc lengths between their embeddings on both circles. More generally, we can consider measuring distances between games' representations using any $p$-norm. In summary, the distance, $d$, between two games can be found by a) converting the game to the equilibrium-invariant embedding, b) calculating the arc length between each component of the two embeddings, and c) finding a distance using any $p$-norm, where $\Theta^A = (\theta^A_1, \theta^A_2)$ and $\Theta^B = (\theta^B_1, \theta^B_2)$.

For all norms with $p<\infty$, there is a unique game which maximizes the distance to any other game. We call such a game the \emph{opposite game}, which can be calculated by translating the representation by a constant of $\pi$ around the circle.
\begin{align}
    \theta_1 \to \text{mod}(\theta_1 + 2\pi, 2\pi) - \pi \qquad \theta_2 \to \text{mod}(\theta_2 + 2\pi, 2\pi) - \pi
\end{align}
For example, the opposite of a  \cyclegame~clockwise Cycle game is an \equigame{-1}{1}{1}{-1}~anti-clockwise Cycle game, and the opposite of a \coordinationgame~diagonal Coordination game is an \equigame{-1}{1}{-1}{1}~off-diagonal (``anti'') Coordination game.

\section{2×2 Equilibrium-Symmetric Embedding}

Two common symmetries are strategy permutation and player permutation. The order of strategies and players in normal-form games is arbitrary, and permuting along these dimensions leads to strategically identical games. The area of the equilibrium-invariant embedding can be reduced by considering these symmetries (Table~\ref{tab:2x2_symmetries}). This reduced space, called the \emph{equilibirum-symmetric embedding}, is shaded in the visualization (Figure~\ref{fig:representation_summary}). It is easy to convert from the equilibrium-invariant embedding to the equilibrium-symmetric embedding (Algorithm~\ref{alg:equilibrium_symmetric_embedding}). If the roles of the players are important, only strategy permutation may be used, and the equilibrium-symmetric embedding will only reduce by a factor of four. The symmetries being utilized are usually clear from the context.

\subsection{Deriving the 2×2 Equilibrium-Symmetric Embedding}

\begin{theorem}[Equilibrium-Symmetric Embedding] \label{thm:equilibrium_symmetric_embedding}
    Strategy permutation and player permutation can reduce the area of the equilibrium-invariant embedding by a factor of eight when only considering a canonical ordering.
\end{theorem}
\begin{proof}[Proof Sketch]
    Symmetries (Table~\ref{tab:2x2_symmetries}), including player permutation and strategy permutation for each player, can be leveraged to reduce the representation space. Each of these reductions reduces the representation area by a factor of two, and when composed, results in a factor of eight reduction.
\end{proof}

\begin{table}[t]
    \centering
    \begin{tabular}{lll}
        Symmetry & Transformation Description & Transformation \\ \hline
        Player 1 Strategy Permutation & Translate on $\theta_1$, fold over $\theta_1 = 0$ & $(\theta_1, \theta_2) \to (\theta_1 + \pi, -\theta_2)$ \\
        Player 2 Strategy Permutation & Fold over $\theta_2 = 0$, translate on $\theta_2$ & $(\theta_1, \theta_2) \to (-\theta_1, \theta_2 + \pi)$ \\
        Player Permutation & Fold over $\theta_1 = \theta_2$ & $(\theta_1, \theta_2) \to (\theta_2, \theta_1)$ \\
    \end{tabular}
    \caption{Symmetries in 2×2 normal-form games used to derive the equilibrium-symmetric embedding.}
    \label{tab:2x2_symmetries}
\end{table}

\begin{algorithm}[t]
    \caption{Equilibrium-Symmetric Embedding}\label{alg:equilibrium_symmetric_embedding}
    \KwData{$\theta_1$, $\theta_2$}
    \KwResult{$\theta_1^*$, $\theta_2^*$}
    $\theta_1^*, \theta_2^* \gets \theta_1, \theta_2$\;
    \If{$\theta_1^* \leq -\frac{\pi}{2}$ or $\frac{\pi}{2} < \theta_1^*$}{
        \tcp{Permute Player 1's Strategies}
        $(\theta_1^*, \theta_2^*) \gets (\theta_1^* + \pi, -\theta_2^*)$\;
    }
    \If{$\theta_2^* \leq -\frac{\pi}{2}$ or $\frac{\pi}{2} < \theta_2^*$}{
        \tcp{Permute Player 2's Strategies}
        $(\theta_1^*, \theta_2^*) \gets (-\theta_1^*, \theta_2^* + \pi)$\;
    }
    \If{$\theta_2^* > \theta_1^*$}{
        \tcp{Permute Players}
        $(\theta_1^*, \theta_2^*) \gets (\theta_2^*, \theta_1^*)$\;
    }
\end{algorithm}

\subsection{2×2 Equilibrium-Symmetric Player-Agnostic Embedding}
\label{sec:game_embedding}

We can define a new parameterization of the equilibrium-invariant embedding, $(\phi_\text{cycle}, \phi_\text{coord})$, which is simply a change in basis from $(\theta_1, \theta_2)$. In the equilibrium-symmetric embedding, $0 \leq \phi_\text{cycle} \leq 1$ and $-1 < \phi_\text{coord} \leq 1$, with $\phi_\text{cycle} \pm \phi_\text{coord} \leq 1$. This basis describes properties of the game, rather than the payoff of each player. The conversion from the per-player basis and the game-property basis is simple.

\noindent\begin{minipage}[b]{0.4\linewidth}
\begin{subequations}
\begin{align}
    \phi_\text{coord} &= \frac{1}{\pi} \left( \theta_1 + \theta_2 \right) \\
    \phi_\text{cycle} &= \frac{1}{\pi} \left( \theta_1 - \theta_2 \right)
\end{align}
\end{subequations}
\end{minipage} ~
\begin{minipage}[b]{0.4\linewidth}
\begin{subequations}
\begin{align}
    \theta_1 &= \frac{\pi}{2}\left( \phi_\text{coord} + \phi_\text{cycle} \right) \\
    \theta_2 &= \frac{\pi}{2}\left( \phi_\text{coord} - \phi_\text{cycle} \right)
\end{align}
\end{subequations}
\end{minipage}

\noindent Therefore games can be described by their coordination (``common-payoff-ness'') and ``cyclicness'' (``zero-sum-ness'') respectively. The game at the origin, with zero coordination and zero cyclicness, is transitive (the \dominantgame~Dominant game).
\begin{subequations}
\begin{align}
    \hat{G}_1(\phi_\text{coord}, \phi_\text{cycle}) =& \begin{bmatrix}
        \phantom{+} \frac{1}{\sqrt{2}} \sin \left(\frac{\pi}{2}\left( \phi_\text{coord} + \phi_\text{cycle} \right) + \frac{\pi}{4}\right) & \phantom{+} \frac{1}{\sqrt{2}} \cos \left(\frac{\pi}{2}\left( \phi_\text{coord} + \phi_\text{cycle} \right) + \frac{\pi}{4}\right) \\
        - \frac{1}{\sqrt{2}} \sin \left(\frac{\pi}{2}\left( \phi_\text{coord} + \phi_\text{cycle} \right) + \frac{\pi}{4}\right) & - \frac{1}{\sqrt{2}} \cos \left(\frac{\pi}{2}\left( \phi_\text{coord} + \phi_\text{cycle} \right) + \frac{\pi}{4}\right) \\
    \end{bmatrix} \\\
    \hat{G}_2(\phi_\text{coord}, \phi_\text{cycle}) =& \begin{bmatrix}
        \phantom{+} \frac{1}{\sqrt{2}} \sin \left(\frac{\pi}{2}\left( \phi_\text{coord} - \phi_\text{cycle} \right) + \frac{\pi}{4}\right) & - \frac{1}{\sqrt{2}} \sin \left(\frac{\pi}{2}\left( \phi_\text{coord} - \phi_\text{cycle} \right) + \frac{\pi}{4}\right) \\
        \phantom{+} \frac{1}{\sqrt{2}} \cos \left(\frac{\pi}{2}\left( \phi_\text{coord} - \phi_\text{cycle} \right) + \frac{\pi}{4}\right) & - \frac{1}{\sqrt{2}} \cos \left(\frac{\pi}{2}\left( \phi_\text{coord} - \phi_\text{cycle} \right) + \frac{\pi}{4}\right) \\
    \end{bmatrix}
\end{align}
\end{subequations}

\subsection{2×2 Equilibrium-Symmetric Distance Metric}

A distance metric between games in the symmetric embedding can also be defined by enumerating all (up to eight) equivalent isomorphic games and finding the minimum distance between the closest pair. For example, some of the distances between points in the symmetric embedding are shown in Table~\ref{tab:equi_distance}.

\section{2×2 Best-Response-Invariant Embedding}

Using the better-response-invariant embedding (Definition~\ref{def:two_player_better_response_invariant_embedding}) we can derive a set of 15 fundamental games. For two-strategy games, the concepts of better-response and best-response are synonymous. Therefore we call this set of games the 2×2 best-response-invariant embedding. Accompanying this embedding are 15 equivalence classes: many equilibrium-invariant embeddings map to each best-response-invariant embedding. This distinction is subtle but important. Sometimes we refer to best-response-invariant embeddings, and sometimes we refer to the equivalence classes they belong to. The best-response-invariant embeddings can be thought of as canonical (fundamental) examples of games within each equivalence class.  The equilibria (WSNE, NE, WSCE, CE, and CCE) of all games within each equivalence class can be calculated through simple scaling of the equilibria of these fundamental games. Of the 15 equivalent game classes, 11 are nontrivial, 3 are partially trivial, and 1 is trivial. These 15 classes are the same as those proposed by \cite{borm1987_classification_of_2x2_games} and correspond to games with the same best-response dynamics. Our work improves upon Borm's classification because in our derivation, the best-response-invariant embedding inherits the distance metric from the equilibrium-invariant embedding. This allows us to situate these game classes within a metric space, which greatly improves the clarity of the embeddings.

\subsection{Deriving the Best-Response Embedding}

\begin{theorem}[2×2 Best-Response Embedding] \label{thm:fundamental_set}
All 2×2 game payoffs can be mapped to a fundamental set consisting of 81 games. After symmetry this reduces to 11 nontrivial games, 3 partially trivial games, and 1 trivial game, resulting in 15 fundamental games. One such mapping is:
\begin{subequations}
\begin{align}
    G_1 &= \begin{bmatrix}
        g_1^{AA} & g_1^{AB} \\
        g_1^{BA} & g_1^{BB} \\
    \end{bmatrix} \to \scriptstyle \left \{ %
        \begin{bmatrix}
            \scriptstyle +1 & \scriptstyle +1 \\
            \scriptstyle -1 & \scriptstyle -1 \\
        \end{bmatrix}, %
        \begin{bmatrix}
            \scriptstyle +1 & \scriptstyle \phantom{+}0 \\
            \scriptstyle -1 & \scriptstyle \phantom{+}0 \\
        \end{bmatrix}, %
        \begin{bmatrix}
            \scriptstyle +1 & \scriptstyle -1 \\
            \scriptstyle -1 & \scriptstyle +1 \\
        \end{bmatrix}, %
        \begin{bmatrix}
            \scriptstyle \phantom{+}0 & \scriptstyle +1 \\
            \scriptstyle \phantom{-}0 & \scriptstyle -1 \\
        \end{bmatrix}, %
        \begin{bmatrix}
            \scriptstyle \phantom{+}0 & \scriptstyle \phantom{+}0 \\
            \scriptstyle \phantom{-}0 & \scriptstyle \phantom{-}0 \\
        \end{bmatrix}, %
        \begin{bmatrix}
            \scriptstyle \phantom{+}0 & \scriptstyle -1 \\
            \scriptstyle \phantom{-}0 & \scriptstyle +1 \\
        \end{bmatrix}, %
        \begin{bmatrix}
            \scriptstyle -1 & \scriptstyle +1 \\
            \scriptstyle +1 & \scriptstyle -1 \\
        \end{bmatrix}, %
        \begin{bmatrix}
            \scriptstyle -1 & \scriptstyle \phantom{+}0 \\
            \scriptstyle +1 & \scriptstyle \phantom{-}0 \\
        \end{bmatrix}, %
        \begin{bmatrix}
            \scriptstyle -1 & \scriptstyle -1 \\
            \scriptstyle +1 & \scriptstyle +1 \\
        \end{bmatrix} %
    \right \} \\
    G_2 &= \begin{bmatrix}
        g_2^{AA} & g_2^{AB} \\
        g_2^{BA} & g_2^{BB} \\
    \end{bmatrix} \to \scriptstyle \left \{ %
        \begin{bmatrix}
            \scriptstyle +1 & \scriptstyle -1 \\
            \scriptstyle +1 & \scriptstyle -1 \\
        \end{bmatrix}, %
        \begin{bmatrix}
            \scriptstyle +1 & \scriptstyle -1 \\
            \scriptstyle \phantom{+}0 & \scriptstyle \phantom{+}0 \\
        \end{bmatrix}, %
        \begin{bmatrix}
            \scriptstyle +1 & \scriptstyle -1 \\
            \scriptstyle -1 & \scriptstyle +1 \\
        \end{bmatrix}, %
        \begin{bmatrix}
            \scriptstyle \phantom{+}0 & \scriptstyle \phantom{-}0 \\
            \scriptstyle +1 & \scriptstyle -1 \\
        \end{bmatrix}, %
        \begin{bmatrix}
            \scriptstyle \phantom{+}0 & \scriptstyle \phantom{+}0 \\
            \scriptstyle \phantom{-}0 & \scriptstyle \phantom{-}0 \\
        \end{bmatrix}, %
        \begin{bmatrix}
            \scriptstyle \phantom{+}0 & \scriptstyle \phantom{-}0 \\
            \scriptstyle -1 & \scriptstyle +1 \\
        \end{bmatrix}, %
        \begin{bmatrix}
            \scriptstyle -1 & \scriptstyle +1 \\
            \scriptstyle +1 & \scriptstyle -1 \\
        \end{bmatrix}, %
        \begin{bmatrix}
            \scriptstyle -1 & \scriptstyle +1 \\
            \scriptstyle \phantom{+}0 & \scriptstyle \phantom{-}0 \\
        \end{bmatrix}, %
        \begin{bmatrix}
            \scriptstyle -1 & \scriptstyle +1 \\
            \scriptstyle -1 & \scriptstyle +1 \\
        \end{bmatrix} %
    \right \}.
\end{align}
\end{subequations}
\end{theorem}
\begin{proof}
Using the per-strategy scale transform, with $b_1$ defined as the usual zero mean offset and $s_2=[|g_1^{A-B,A}|, |g_1^{A-B,B}|]$, we can derive the player embedding for player 1.
\begin{subequations}
\begin{align}
    G_1 &= \begin{bmatrix}
        g_2^{AA} & g_2^{AB} \\
        g_2^{BA} & g_2^{BB} \\
    \end{bmatrix} \to \begin{bmatrix}
        \phantom{+}\frac{1}{2} g_1^{A-B,A} & \phantom{+}\frac{1}{2} g_1^{A-B,B} \\
        -\frac{1}{2} g_1^{A-B,A} & -\frac{1}{2} g_1^{A-B,B} \\
    \end{bmatrix} \to \begin{bmatrix}
        \phantom{-}\frac{g_1^{A-B,A}}{|g_1^{A-B,A}|} & \phantom{-} \frac{g_1^{A-B,B}}{|g_1^{A-B,B}|} \\
        - \frac{g_1^{A-B,A}}{|g_1^{A-B,A}|} & - \frac{g_1^{A-B,B}}{|g_1^{A-B,B}|} \\
    \end{bmatrix} \\
    &= \begin{bmatrix}
        \{ +1, +1, +1, 0, 0, 0, -1, -1, -1 \} & \{ +1, 0, -1, +1, 0, -1, +1, 0, -1 \} \\
        \{ -1, -1, -1, 0, 0, 0, +1, +1, +1 \} & \{ -1, 0, +1, -1, 0, +1, -1, 0, +1 \} \\
    \end{bmatrix}
\end{align}
\end{subequations}
Dividing a number by its absolute value results in its sign, and adopting the convention (without consequence) that $\frac{0}{|0|} = 0$, we arrive at 9 payoff possibilities, one of which is the trivial zero payoff. We can follow a similar derivation for player 2. As a result, $9 \times 9 = 81$ games are possible in total. Of these 64 are nontrivial, 16 are partially trivial and 1 is trivial. The permutation symmetries discussed earlier facilitate the observation that only 15 of the games are unique up to symmetric equivalence.
\end{proof}

The best-response-invariant embedding set could have been defined using binary payoffs \citep{fishburn1990_binary_2x2_games} because we only need to establish preferences over the strategies. However we opt for ternary payoffs ($-1, 0, +1$) because they more clearly differentiate the three possibilities: preferring strategy $A$, preferring strategy $B$, or being indifferent. This preference ordering is equivalent to the best-response dynamics of each player. For two-strategy games, better-response and best-response dynamics are identical. Better-response invariance is closely related to the affine transform used to derive the equilibrium-invariant embedding \citep{morris2004_best_response_equivalence,ostrovski2013_fictitious_play_dynamics}. Equilibrium-invariance implies better-response-invariance, which in turn implies best-response-invariance.

\subsection{Equivalence Classes, Graph Representation, and Measure}

The equivalence classes of the best-response-invariant embedding are equivalent to considering orderings of each column in player 1's payoff, and each row in player 2's payoff (the best-response dynamics). By convention let the top-left joint strategy be $AA$ and the bottom right be $BB$, let player 1 be the row player and player 2 be the column player. This ordering can be simply visualized as a directed graph: player 1's column ordering can be indicated with vertical arrows (for example \equigame{-+}{+-}{00}{00} indicates player 1 prefers strategy $B$ when player 2 plays $A$, and $A$ when player 2 players $B$), and player 2's row ordering can be indicated with horizontal arrows (for example \equigame{00}{00}{-+}{-+}). When there is no preference, no edge need be drawn, \equigame{00}{00}{00}{00}. Taken together, we can visualize any 2×2 game. For example, Prisoner's Dilemma (Figure~\ref{fig:prisoners_dilemma}) would be denoted \equigame{-+}{-+}{-+}{-+}, Chicken (Figure~\ref{fig:chicken}), would be denoted \equigame{-+}{+-}{-+}{+-}, and Matching Pennies (Figure~\ref{fig:game_matching_pennies}), would be denoted \equigame{+-}{-+}{-+}{+-}. The set of directed graphs have 15 equivalence classes with respect to graph isomorphism.

Most of the games in the best-response-invariant embedding have measure-zero equivalence classes within the equilibrium-invariant embedding, meaning the probability they are randomly sampled is zero. There is a simple rule of thumb to determine the measure of an equivalence class: a) if no players are indifferent when best-responding, the game is positive measure, b) if one player is indifferent when best responding to one action, the class is zero-measure and is a ``boundary class'' appearing between two positive measure classes, and c) if both players are indifferent to one of the opponent's actions, the class is a ``point class'' appearing between four boundary classes. It is easy to visually inspect (Figure~\ref{fig:equilibrium_support}) that only 4 of the games have positive probability of being sampled: \samaritangame~Samaritan, \dominantgame~Dominant, \coordinationgame~Coordination, and \cyclegame~Cycle (Table~\ref{tab:naming_games}). This provides intuition to why the measure-zero games rarely appear in the literature. Our observation matches similar analysis made by \cite{borm1987_classification_of_2x2_games}.

\subsection{Naming the Best-Response-Invariant Embedding}

\cite{borm1987_classification_of_2x2_games} provided numerical classification for the 15 fundamental classes of games. We have already specified parameterized embeddings, $(\theta_1,\theta_2)$ or $(\phi_\text{coord}, \phi_\text{cycle})$, and a graphical representation. In addition, the set is small enough to benefit from a standardised naming scheme. We attempt to, as far as possible, follow established naming conventions and draw inspiration from previous work \citep{bruns2015_names_for_games}. Symmetric games (those that lie on $\theta_1 = \theta_2$ or $\phi_\text{cycle}=0$) are the most well-studied and have established common names. Only some non-symmetric games such as \penniesgame~Matching Pennies and Samaritan's Dilemma have been studied. Care was taken not to name classes after games that are similar but subtly different (e.g. \huntgame~Stag Hunt and \chickengame~Chicken were deliberately avoided). The majority of the games have strategies which players are indifferent between. These classes of games are little studied and therefore do not have common names. Table~\ref{tab:naming_games} shows the names chosen for the 15 fundamental games proposed in this work.

\begin{table}[t]
    \centering
    {\setlength{\tabcolsep}{4pt}
    \begin{tabular}{llllllrcclllll}
        G & Name & N & PN & MN & C & B & $\theta$ & $\phi$ & S & M & T & I &  Other Names \\ \hline
        \dominantgame & Dominant & Do & \eqsupport{1}{0}{0}{0} & & \dominantsupport & 5 & $\phantom{+}0,\phantom{+}0$ & $\phantom{+}0,\phantom{+}0$ & \checkmark & $\frac{1}{4}$ & N & 0 & Prisoner's Dilemma \\
        \coordinationgame & Coordination & Co &  \eqsupport{1}{0}{0}{0}~\eqsupport{0}{0}{0}{1} & \eqsupport{1}{1}{1}{1} & \coordinationsupport & 14 & $\phantom{+}\frac{\pi}{2},\phantom{+}\frac{\pi}{2}$  & $\phantom{+}1,\phantom{+}0$ & \checkmark & $\frac{1}{8}$ & N & 0 & Battle, Hunt, Chicken \\
        \cyclegame & Cycle & Cy & & \eqsupport{1}{1}{1}{1} & \cyclesupport & 15 & $\phantom{+}\frac{\pi}{2},-\frac{\pi}{2}$ & $\phantom{+}0,\phantom{+}1$ & & $\frac{1}{8}$ & N & 0 & Matching Pennies \\
        \samaritangame & Samaritan & Sm & \eqsupport{1}{0}{0}{0} & & \samaritansupport & 8 & $\phantom{+}\frac{\pi}{2},\phantom{+}0$ & $\phantom{+}\frac{1}{2},\phantom{+}\frac{1}{2}$ & & $\frac{1}{2}$ & N & 0 & \\
        
        \hazardgame & Hazard & Hz & \eqsupport{1}{0}{0}{0} & \eqsupport{1}{1}{0}{0} & \hazardsupport & 13 & $\phantom{+}\frac{\pi}{2},-\frac{\pi}{4}$ & $\phantom{+}\frac{1}{4},\phantom{+}\frac{3}{4}$ & & $-$ & N & 1 & \\
        \safetygame & Safety & Sf & \eqsupport{1}{0}{0}{0}~\eqsupport{0}{0}{0}{1} & \eqsupport{1}{0}{1}{0} & \safetysupport & 12 & $\phantom{+}\frac{\pi}{2},\phantom{+}\frac{\pi}{4}$ & $\phantom{+}\frac{3}{4},\phantom{+}\frac{1}{4}$ & & $-$ & N & 1 & \\
        \aidosgame & Aidos & Ad & \eqsupport{1}{0}{0}{0}~\eqsupport{0}{1}{0}{0} & \eqsupport{1}{1}{0}{0} & \aidossupport & 7 & $\phantom{+}0,-\frac{\pi}{4}$ & $-\frac{1}{4},-\frac{1}{4}$ & & $-$ & N & 1 & \\
        \picnicgame & Picnic & Pn & \eqsupport{1}{0}{0}{0} & & \picnicsupport & 6 & $\phantom{+}\frac{\pi}{2},\phantom{+}0$ & $\phantom{+}\frac{1}{4},\phantom{+}\frac{1}{4}$ & & $-$ & N & 1 & \\
        
        \daredevilgame & Daredevil & Dd & \eqsupport{1}{0}{0}{0}~\eqsupport{0}{1}{0}{0}~\eqsupport{0}{0}{1}{0} & \eqsupport{1}{0}{1}{0}~\eqsupport{1}{1}{0}{0} & \daredevilsupport & 11 & $-\frac{\pi}{4},-\frac{\pi}{4}$ & $-\frac{1}{2},\phantom{+}0$ & \checkmark & $\cdot$ & N & 2 & \\
        \fossickgame & Fossick & Fo & \eqsupport{1}{0}{0}{1} & & \fossicksupport & 9 & $\phantom{+}\frac{\pi}{4},\phantom{+}\frac{\pi}{4}$ & $\phantom{+}\frac{1}{2},\phantom{+}0$ & \checkmark & $\cdot$ & N & 2 & \\
        \heistgame & Heist & Hs & \eqsupport{1}{0}{0}{0}~\eqsupport{0}{1}{0}{0} & \eqsupport{1}{1}{0}{0} & \heistsupport & 10 & $\phantom{+}\frac{\pi}{4},-\frac{\pi}{4}$ & $\phantom{+}0,\phantom{+}\frac{1}{2}$ & & $\cdot$ & N & 2 & \\
        
        \hline
        
        \ignorancegame & Ignorance & Ig & \eqsupport{1}{0}{0}{0}~\eqsupport{0}{1}{0}{0} & \eqsupport{1}{1}{0}{0} & \ignorancesupport & 2 &  & & & $\frac{1}{2}$ & P & 2 &  \\
        \horseplaygame & Horseplay & Hp & \eqsupport{1}{0}{0}{0}~\eqsupport{1}{0}{0}{1} & \eqsupport{1}{1}{0}{0}~\eqsupport{0}{0}{1}{1}  & \horseplaysupport & 4 &  & & & $\frac{1}{2}$ & P & 2 &  \\
        \dressgame & Dress & Dr & \eqsupport{1}{0}{0}{0}~\eqsupport{0}{1}{0}{0}~\eqsupport{0}{0}{0}{1} & \eqsupport{1}{1}{0}{0}~\eqsupport{0}{1}{0}{1} & \dresssupport & 3 &  & & & $-$ & P & 3 & Red Dress \\
        
        \hline
        
        \nullgame & Null & Nu & \eqsupport{1}{0}{0}{0}~\eqsupport{0}{1}{0}{0}~\eqsupport{0}{0}{0}{1}~\eqsupport{0}{0}{1}{0} & \eqsupport{1}{1}{1}{1} & \nullsupport & 1 &  & & \checkmark & 1 & T & 4 & Trivial, Zero \\
    \end{tabular}
    }
    \caption{Naming scheme and properties of the set of 2×2 best-response-invariant embeddings. Only \dominantgame~Dominant, \coordinationgame~Coordination, \cyclegame~Cycle, \samaritangame~Samaritan, and \nullgame~Null have been well studied and have established names. The remaining games (which have indifferences) either do not appear in the literature at all or do so seldomly. Key: graph representation (G), short name (N), pure Nash (PN), mixed Nash (MN), correlated equilibrium support (C), \cite{borm1987_classification_of_2x2_games}'s classes (B), symmetric (S), measure (M), triviality (T), and indifferences (I).}
    \label{tab:naming_games}
\end{table}

\begin{figure}[t]
    \centering
    \footnotesize

    \begin{subfigure}[b]{0.19\textwidth}
        \centering
        \begin{tabular}{c|cc}
              & S & W \\ \hline
            S & $+1,+1$ & $+1,-1$ \\
            W & $-1,+1$ & $-1,-1$ \\
        \end{tabular}
        \caption{\centering \dominantgame~Dominant}
        \label{fig:equi_dominant}
    \end{subfigure}
    \hfill
    \begin{subfigure}[b]{0.19\textwidth}
        \centering
        \begin{tabular}{c|cc}
              & A & B \\ \hline
            A & $+1,+1$ & $-1,-1$ \\
            B & $-1,-1$ & $+1,+1$ \\
        \end{tabular}
        \caption{\centering \coordinationgame~Coordination}
        \label{fig:equi_coordination}
    \end{subfigure}
    \hfill
    \begin{subfigure}[b]{0.19\textwidth}
        \centering
        \begin{tabular}{c|cc}
              &       H &       T \\ \hline
            H & $+1,-1$ & $-1,+1$ \\
            T & $-1,+1$ & $+1,-1$ \\
        \end{tabular}
        \caption{\centering \cyclegame~Cycle}
        \label{fig:equi_cycle}
    \end{subfigure}
    \hfill
    \begin{subfigure}[b]{0.19\textwidth}
        \centering
        \begin{tabular}{c|cc}
              & H & N \\ \hline
            R & $4,3$ & $1,1$ \\
            W & $3,4$ & $2,2$ \\
        \end{tabular}
        \caption{\centering \samaritangame~Samaritan}
        \label{fig:equi_samaritan}
    \end{subfigure}
    \hfill
    \begin{subfigure}[b]{0.19\textwidth}
        \centering
        \begin{tabular}{c|cc}
              &       R &       C \\ \hline
            N & $~0,+1$ & $~0,+1$ \\
            I & $-2,+3$ & $+2,+2$ \\
        \end{tabular}
        \caption{\centering \hazardgame~Hazard}
        \label{fig:equi_hazard}
    \end{subfigure}

    \begin{subfigure}[b]{0.19\textwidth}
        \centering
        \begin{tabular}{c|cc}
              &       C &       R \\ \hline
            I & $+1,+4$ & $-2,+3$ \\
            N & $~0,+1$ & $~0,+1$ \\
        \end{tabular}
        \caption{\centering \safetygame~Safety}
        \label{fig:equi_safety}
    \end{subfigure}
    \hfill
    \begin{subfigure}[b]{0.19\textwidth}
        \centering
        \begin{tabular}{c|cc}
              &       B &       N \\ \hline
            N & $+1,~0$ & $+1,~0$ \\
            R & $-1,+1$ & $-1,-1$ \\
        \end{tabular}
        \caption{\centering \aidosgame~Aidos}
        \label{fig:equi_aidos}
    \end{subfigure}
    \hfill
    \begin{subfigure}[b]{0.19\textwidth}
        \centering
        \begin{tabular}{c|cc}
              &       P &       T \\ \hline
            A & $+1,+1$ & $~0,-1$ \\
            D & $-1,+1$ & $~0,-1$ \\
        \end{tabular}
        \caption{\centering \picnicgame~Picnic}
        \label{fig:equi_picnic}
    \end{subfigure}
    \hfill
    \begin{subfigure}[b]{0.19\textwidth}
        \centering
        \begin{tabular}{c|cc}
              &       B &       N \\ \hline
            N & $~0,~0$ & $+1,~0$ \\
            R & $~0,+1$ & $-1,-1$ \\
        \end{tabular}
        \caption{\centering \daredevilgame~Daredevil}
        \label{fig:equi_daredevil}
    \end{subfigure}
    \hfill
    \begin{subfigure}[b]{0.19\textwidth}
        \centering
        \begin{tabular}{c|cc}
              &       W &       G \\ \hline
            W & $+1,~0$ & $+1,~0$ \\
            G & $-1,+1$ & $-1,-1$ \\
        \end{tabular}
        \caption{\centering \fossickgame~Fossick}
        \label{fig:equi_fossick}
    \end{subfigure}

    \begin{subfigure}[b]{0.19\textwidth}
        \centering
        \begin{tabular}{c|cc}
              &       P &       R \\ \hline
            N & $+1,~0$ & $~0,~0$ \\
            H & $-1,+1$ & $~0,-1$ \\
        \end{tabular}
        \caption{\centering \heistgame~Heist}
        \label{fig:equi_heist}
    \end{subfigure}
    \hfill
    \begin{subfigure}[b]{0.19\textwidth}
        \centering
        \begin{tabular}{c|cc}
              &       A &       B \\ \hline
            W & $+1,~0$ & $+1,~0$ \\
            L & $-1,~0$ & $-1,~0$ \\
        \end{tabular}
        \caption{\centering \ignorancegame~Ignorance}
        \label{fig:equi_ignorance}
    \end{subfigure}
    \hfill
    \begin{subfigure}[b]{0.19\textwidth}
        \centering
        \begin{tabular}{c|cc}
              &       B &       N \\ \hline
            N & $+1,~0$ & $+1,~0$ \\
            R & $-1,+1$ & $-1,-1$ \\
        \end{tabular}
        \caption{\centering \horseplaygame~Horseplay}
        \label{fig:equi_horseplay}
    \end{subfigure}
    \hfill
    \begin{subfigure}[b]{0.19\textwidth}
        \centering
        \begin{tabular}{c|cc}
              &       C &       F \\ \hline
            C & $-2,~0$ & $+1,~0$ \\
            F & $-2,+1$ & $-1,+1$ \\
        \end{tabular}
        \caption{\centering \dressgame~Dress}
        \label{fig:equi_dress}
    \end{subfigure}
    \hfill
    \begin{subfigure}[b]{0.19\textwidth}
        \centering
        \begin{tabular}{c|cc}
              &       A &       B \\ \hline
            A & $~0,~0$ & $~0,~0$ \\
            B & $~0,~0$ & $~0,~0$ \\
        \end{tabular}
        \caption{\centering \nullgame~Null}
        \label{fig:equi_null}
    \end{subfigure}
    
    \caption{Narrative payoff tables of the 15 best-response-invariant 2×2 games.}
    \label{fig:best_response_invariant_2x2_games}
\end{figure}

\paragraph{\dominantgame~Dominant} In this symmetric game (Figure~\ref{fig:equi_dominant}) each player has a strictly dominant strategy. Therefore, this game has a single pure NE and (C)CE \dominantsupport. The Dominant embedding is at the origin of the equilibrium-invariant embedding, like other topologies \citep{rapoport1976_the_2x2_game_book}, and is the only 2×2 game that is both invariant-zero-sum and invariant-common-payoff. Dominant has many common games within its equivalence class including Prisoner's Dilemma, Peace, Deadlock, Total Conflict, Concord, and Compromise. In particular, Prisoner's Dilemma is one of the most studied games in economics. The dominant equivalence class occurs with probability $\frac{1}{4}$ when sampling uniformly over the equilibrium-invariant embedding. Furthermore, this class is the most diverse: games within the class can be anti-clockwise invariant-zero-sum, clockwise invariant-zero-sum, diagonal invariant-common-payoff, off-diagonal invariant-common-payoff, or none.

\paragraph{\coordinationgame~Coordination} Coordination (Figure~\ref{fig:equi_coordination}) is a symmetric common-payoff coordination game where there are two equally desirable outcomes. Players simply have to coordinate to ensure they achieve one of them. Coordination has two pure NEs (\eqsupport{1}{0}{0}{0} and \eqsupport{0}{0}{0}{1}) with equal payoff and a low payoff mixed NE \coordinationsupport. It is the only nontrivial game with a positive volume of (C)CE equilibria. Coordination has many games within its equivalence class including \huntgame~Stag Hunt, \chickengame~Chicken, Assurance, and \battlegame~Bach or Stravinsky. The Coordination equivalence class occurs with probability $\frac{1}{8}$ when sampling uniformly over the equilibrium-invariant embedding. All games in the coordination equivalence class are invariant-common-payoff.

\paragraph{\cyclegame~Cycle} This game is an anti-symmetric zero-sum game (Figure~\ref{fig:equi_cycle}), also commonly called \penniesgame~Matching Pennies. In this game, the first player prefers both players to play the same strategy (for example, both heads or both tails), and the second player prefers each player to play a different strategy. The best-responses of the pure strategies result in cyclic dynamics. A similar three strategy variant of this game with the same property, Rock Paper Scissors, is a popular children's game. Permuting the strategies of a player will result in changing the game from a clockwise cycle to an anti-clockwise cycle. Cycle has no pure NEs, a single completely mixed NE (with uniform probability) and (C)CE \cyclesupport. Cycle is the only game where all equilibria are strictly in the interior of the simplex. All games within \cyclegame~Cycle's best-response-invariant equivalence class are also cyclic, but with different biases on the strategies they prefer. This class occurs with probability $\frac{1}{8}$ when sampling uniformly over the equilibrium-invariant embedding.

\paragraph{\samaritangame~Samaritan} This game (Figure~\ref{fig:equi_samaritan}) has a worker player and a Samaritan player. The worker has two strategies: rest (R) or work (W). The Samaritan has two strategies: help (H) or do not help (N). The Samaritan always prefers to help but also prefers it when the worker also works (and so does not free-load). The worker prefers not to work if they receive help but prefers to work to sustain themselves if they do not receive help. Samaritan has a single pure NE and (C)CE \samaritansupport. It is named after Samaritan's Dilemma \citep{schmidtchen2002_samaritan_revisted}, a game proposed by \cite{buchanan1975_samaritans_dilemma}. Samson \citep{brams1993_theory_of_moves}, Alibi \citep{robinsonandgoforth2005_topology_of_2x2_games_book}, Anticipation, Bully, Hamlet, Asymmetric Dilemma, and Called Bluff are all in this game's better-response-invariant equivalence class. The \samaritangame~Samaritan embedding is neither zero-sum-invariant nor common-payoff-invariant. However other games within its equivalence class can be either zero-sum-invariant or common-payoff-invariant, or neither. Samaritan equivalence class is special because it occurs with the highest probability $\frac{1}{2}$ when sampling uniformly over the equilibrium-invariant embedding. As well as being the most common class, it is also the class most closely connected to all other games (Section~\ref{sec:equi_distance}).

\paragraph{\hazardgame~Hazard} In Hazard (Figure~\ref{fig:equi_hazard}), player 1 is an insurance company that can either insure (I) or decline (D) player 2's health insurance. Player 2 enjoys drinking at parties, which is risky behaviour (R), but also appreciates the health benefits of not drinking, which is careful behaviour (C). Overall, in the absence of health insurance, player 2 is indifferent to a party lifestyle or a healthy one. However, if player 1 issues medical insurance, player 2 will have a greater risk tolerance and will adopt a party lifestyle. Player 1 does not make any money if it does not issue insurance, makes money if it does and player 2 opts for a healthy lifestyle, and loses money if it issues insurance to an unhealthy lifestyle. \cite{bruns2015_names_for_games} briefly defines a game that he describes as a moral hazard, with the formulaic name ``Middle Hunt × Low Dilemma''. Hazard has a single pure NE \eqsupport{1}{0}{0}{0}, a mixed NE \eqsupport{1}{1}{0}{0}, and (C)CEs with support \hazardsupport. The hazard equivalence class has an indifference and so is zero-measure in the equilibrium-invariant embedding. However it is a boundary class which borders the Cycle equivalence class. Of the boundary classes it occurs with probability $\frac{1}{4}$ when uniformly sampling over the equilibrium-invariant embedding. All games within this equivalence class are invariant-zero-sum.

\paragraph{\safetygame~Safety} This game (Figure~\ref{fig:equi_safety}) is also hinted at by \cite{bruns2015_names_for_games} in which it was called ``Middle Hunt × Low Concord''. It is similar to Hazard, except the insurer now incentivizes player 1 into a healthy lifestyle with free gym membership. This fixes the moral hazard and now player 2 is incentivized to have a healthy lifestyle with insurance. Safety has two pure NEs (\eqsupport{1}{0}{0}{0} and \eqsupport{0}{0}{0}{1}), a mixed NE \eqsupport{1}{0}{1}{0} and (C)CEs \safetysupport. The Safety equivalence class has an indifference and so is zero-measure in the equilibrium-invariant embedding. However it is a boundary class which borders the Coordination equivalence class. Of the boundary classes it occurs with probability $\frac{1}{4}$ when uniformly sampling over the equilibrium-invariant embedding. All games within this equivalence class are invariant-common-payoff.

\paragraph{\aidosgame~Aidos} In Aidos (Figure~\ref{fig:equi_aidos}), player 1 is a deity and can either reveal themselves (R) or not (N), however they are shy and strictly prefer not to reveal themselves. Player 2 is a human and wants to believe what is true. With lack of evidence either way, the human is indifferent to the existence of the deity, however if the deity reveals themselves they prefer to believe (B) rather than not believe (N). There are two pure NEs (\eqsupport{1}{0}{0}{0} and \eqsupport{0}{1}{0}{0}), and any mixture of these is also a mixed NE \eqsupport{1}{1}{0}{0} and (C)CE. The theme of this game is inspired by Revelation \citep{brams1993_theory_of_moves}, a game with similar dynamics. The Aidos equivalence class has an indifference and so is zero-measure in the equilibrium-invariant embedding. However it is a boundary class which borders the Dominant and Samaritan equivalence classes in the region where the pure equilibrium strategy changes (in contrast to Picnic, described next). Of the boundary classes it occurs with probability $\frac{1}{4}$ when uniformly sampling over the equilibrium-invariant embedding. The games in Aidos' equivalence class can be either zero-sum-invariant, common-payoff-invariant, or neither. The only game in the set that is neither is the Aidos embedding.

\paragraph{\picnicgame~Picnic} In this game (Figure~\ref{fig:equi_picnic}) a host can either organise a picnic (P) or order takeaway (T) and they strictly prefer organising a picnic. A guest can either attend (A) or decline (D). They are indifferent about attending when takeaway is ordered but would prefer to attend if a picnic is organised. Picnic has a single pure NE \picnicsupport. The Picnic equivalence class has an indifference and so is zero-measure in the equilibrium-invariant embedding. However it is a boundary class which borders the Dominant and Samaritan equivalence classes in the region where the pure equilibrium strategy does not change (in contrast to Aidos). Of the boundary classes it occurs with probability $\frac{1}{4}$ when uniformly sampling over the equilibrium-invariant embedding. The Picnic embedding is neither zero-sum-invariant nor common-payoff-invariant, however games in its equivalence class can be either zero-sum-invariant or common-payoff-invariant.

\paragraph{\daredevilgame~Daredevil} In the literature, \chickengame~Chicken (Figure~\ref{fig:chicken}) has two pure anti-coordination NEs (\eqsupport{0}{0}{1}{0} and \eqsupport{0}{1}{0}{0}) and a single mixed NE \eqsupport{1}{1}{1}{1} which places the majority of the mass on the worst joint outcome. The presence of a mixed NE means that Chicken has a full-support equilibrium and therefore falls into the equivalence class of \equigame{-1}{1}{-1}{1}~Coordination. However, there is another interesting symmetric, invariant-common-payoff, measure-zero equivalence class game, Daredevil (Figure~\ref{fig:equi_daredevil}), with properties similar to Chicken. This game has the same two pure anti-coordination NEs (\eqsupport{0}{0}{1}{0} and \eqsupport{0}{1}{0}{0}), and an additional pure NE \eqsupport{1}{0}{0}{0}: crash. There is no completely mixed NE, but there are are also two pure-mixed NEs (\eqsupport{1}{0}{1}{0} and \eqsupport{1}{1}{0}{0}). It also has (C)CEs \daredevilsupport~ which can mix arbitrarily over any of the pure NEs. The game differs from Chicken in the fact that if the other player chooses not to swerve, the player is indifferent to whether they swerve and avoid damage or continue and crash. It has similar dynamics to, and can be thought of as, an extreme edge-case of Chicken. To our knowledge, this game has not been studied before in the literature - perhaps unsurprising because it has a measure-zero equivalence class. Daredevil is a point class because the Daredevil embedding is the only game within its class. It is invariant-common-payoff and borders the Aidos and Safety equivalence classes. Of the three point classes, it occurs with probability $\frac{1}{4}$.

\paragraph{\fossickgame~Fossick} In the literature, \huntgame~Stag Hunt (Figure~\ref{fig:stag_hunt}) is a game with two pure coordination NEs (\eqsupport{1}{0}{0}{0} and \eqsupport{0}{0}{0}{1}), where one is preferred over the other, and a completely mixed NE \eqsupport{1}{1}{1}{1}. Stag Hunt is also in the equivalence class of \coordinationgame~Coordination. Fossick (Figure~\ref{fig:equi_fossick}) is an extreme version of Stag Hunt which has two pure NEs(\eqsupport{1}{0}{0}{0} and \eqsupport{0}{0}{0}{1}), but no mixed NE. It also has (C)CEs that can mix over only the coordination strategies \fossicksupport. In this game players are in a gold rush and can either search for water to sustain themselves in the wilderness or fossick for gold. If the other player searches for water, the player is indifferent to what they search for. However if the other player fossicks for gold, the player would feel left out and would prefer to also fossick for gold. To our knowledge, this game has not been studied before in the literature. Fossick is part of a zero-measure equivalence class, and is also a point class (the Fossick embedding is the only game in the equivalence class). Of the three point classes, it occurs with probability $\frac{1}{4}$. Fossick is invariant-common-payoff.

\paragraph{\heistgame~Heist} In Heist (Figure~\ref{fig:equi_heist}), player 1 is a nervous robber and can either do nothing (N) or stage a heist (H). Player 2 is a security guard and can either go on patrol (P) or rest (R). If the robber does not stage a heist, the security guard is indifferent to whether they are on patrol or at rest. However if a heist does occur the guard prefers to be on patrol. If the security guard is at rest, the robber is unsure if the heist is worth the risk and is indifferent. However if the security guard is on patrol, the robber prefers not to stage a heist. This game is most similar to Pursuit \citep{brams1993_theory_of_moves}. Heist has two pure NEs (\eqsupport{1}{0}{0}{0} and \eqsupport{0}{1}{0}{0}), and (C)CEs that mix between them \eqsupport{1}{1}{0}{0}. Heist is part of a zero-measure equivalence class, and is also a point class. Of the three point classes, Heist is the most common and occurs with probability $\frac{1}{2}$. Heist is invariant-zero-sum.

\paragraph{\ignorancegame~Ignorance} Ignorance (Figure~\ref{fig:equi_ignorance}) is a partially trivial game where one of the players has no preferences at all. In Ignorance, an informed player wants to win (W) and avoid losing (L) and an ignorant player is indifferent between their strategies and outcomes. Ignorance has two pure NEs (\eqsupport{1}{0}{0}{0} and \eqsupport{0}{1}{0}{0}), a mixed NE \eqsupport{1}{1}{0}{0} and (C)CE \eqsupport{1}{1}{0}{0}. The equivalence class of Ignorance occurs with probability $\frac{1}{2}$ when sampling uniformly over the partially-trivial equilibrium-invariant embedding.

\paragraph{\horseplaygame~Horseplay} Horseplay (Figure~\ref{fig:equi_horseplay}) is a partially trivial game where one of the players has no preferences at all. Horseplay is a game between an adult and a child. The child can either lift up their arms (A) or curl into a ball (B), they are indifferent to which they do or what the outcome is, they are just happy to play. For convenience, the adult will prefer to throw (T) the child into the air if they have their arms up, or lift them up by their feet (L) if they are curled into a ball. Horseplay has two pure NEs (\eqsupport{1}{0}{0}{0} and \eqsupport{0}{0}{0}{1}), three mixed NEs (\eqsupport{1}{1}{0}{0}, \eqsupport{0}{0}{1}{1}, and \eqsupport{1}{1}{1}{1}), and full support (C)CEs \eqsupport{1}{1}{1}{1}. The equivalence class of Horseplay occurs with probability $\frac{1}{2}$ when sampling uniformly over the partially-trivial equilibrium-invariant embedding.

\paragraph{\dressgame~Dress} Described by \cite{simpson2010_red_dress}, Dress (Figure~\ref{fig:equi_dress}) is a game where two people are going on a date. Player 1 selfishly wishes their partner to dress formally (F), and only if they do so, also lazily prefers to dress for comfort (C). Player 2 also wants their partner to dress formally, but is indifferent to what they wear themselves. Because player 2's preferences are solely based on the other player's strategies, their payoffs are trivial, so Dress is a partially trivial game. This game has three pure NEs, and any mixture between those three NEs (\eqsupport{1}{0}{0}{0}, \eqsupport{0}{1}{0}{0}, and \eqsupport{0}{0}{0}{1}), two mixed NEs (\eqsupport{1}{1}{0}{0} and \eqsupport{0}{1}{0}{1}), and (C)CEs \dresssupport. The equivalence class of Dress is a boundary class between Ignorance and Horseplay. It is measure-zero when sampling uniformly over the partially-trivial equilibrium-invariant embedding.

\paragraph{\nullgame~Null} In the Null game (Figure~\ref{fig:equi_null}, players have no preferences over strategies which means any pure or mixed strategy is an NE and any joint distribution is a (C)CE \nullsupport. Games of the form $G_1(a_1, a_2) = b_1(a_2)$, $G_2(a_1, a_2) = b_2(a_1)$ are in the Null equivalence class. Null is a point game and is the only class of trivial games. This game is sometimes called the Zero game or Trivial game.

\subsection{Distance Metric}
\label{sec:equi_distance}

Distances between the best-response-invariant embeddings can be computed (Table~\ref{tab:equi_distance}) using the equilibrium-symmetric distance metric. \cyclegame~Cycle is the most isolated game with the greatest average distance to other games. \samaritangame~Samaritan is the least isolated game with the smallest average distance to other games. No game is more than two steps\footnote{Steps or ``hops'' are measured over the planar grid visualized in Figure~\ref{fig:representation_summary}.} away from \samaritangame~Samaritan. Out of the partially trivial games, \dressgame~Dress is a boundary class between \ignorancegame~Ignorance and \horseplaygame~Horseplay. Small perturbations in payoffs can change the game and cause step changes in the equilibrium. The distance metrics (and the topology) show the possible adjacent games that small perturbations could result in. Such an analysis could be useful in selecting equilibria that are robust to permutations or dealing with uncertainty in payoffs. We expand on this more in the discussion section.

\begin{table}[t]
\centering
{\setlength{\tabcolsep}{3pt}
\begin{tabular}{r|ccccccccccc|ccc|c}
     &
     \rotatebox{90}{\dominantgame~Dominant} &
     \rotatebox{90}{\coordinationgame~Coordination} &
     \rotatebox{90}{\cyclegame~Cycle} &
     \rotatebox{90}{\samaritangame~Samaritan} &
     \rotatebox{90}{\hazardgame~Hazard} &
     \rotatebox{90}{\safetygame~Safety} &
     \rotatebox{90}{\aidosgame~Aidos} &
     \rotatebox{90}{\picnicgame~Picnic} &
     \rotatebox{90}{\daredevilgame~Daredevil} &
     \rotatebox{90}{\fossickgame~Fossick} &
     \rotatebox{90}{\heistgame~Heist} &
     \rotatebox{90}{\ignorancegame~Ignorance} &
     \rotatebox{90}{\horseplaygame~Horseplay} &
     \rotatebox{90}{\dressgame~Dress} &
     \rotatebox{90}{\dressgame~Null} \\ \hline
    Dominant \dominantgame & 0 & 4 & 4 & 2 & 3 & 3 & 1 & 1 & 2 & 2 & 2 & & & & \\
    Coordination \coordinationgame & 4 & 0 & 4 & 2 & 3 & 1 & 3 & 3 & 2 & 2 & 4 & & & & \\
    Cycle \cyclegame & 4 & 4 & 0 & 2 & 1 & 3 & 3 & 3 &  4 & 4 & 2 & & & & \\
    Samaritan \samaritangame & 2 & 2 & 2 & 0 & 1 & 1 & 1 & 1 & 2 & 2 & 2 & & & & \\
    Hazard \hazardgame & 3 & 3 & 1 & 1 & 0 & 2 & 2 & 2 & 3 & 3 & 1 & & & & \\
    Safety \safetygame & 3 & 1 & 3 & 1 & 2 & 0 & 2 & 2 & 1 & 1 & 3 & & & & \\
    Aidos \aidosgame &  1 & 3 & 3 & 1 & 2 & 2 & 0 & 2 & 1 & 3 & 1 & & & & \\
    Picnic \picnicgame & 1 & 3 & 3 & 1 & 2 & 2 & 2 & 0 & 3 & 1 & 1 & & & &  \\
    Daredeveil \daredevilgame & 2 & 2 & 4 & 2 & 3 & 1 & 1 & 3 & 0 & 2 & 2 & & & & \\
    Fossick \fossickgame & 2 & 2 & 4 & 2 & 3 & 1 & 3 & 1 & 2 & 0 & 2  & & & & \\
    Heist \heistgame & 2 & 4 & 2 & 2 & 1 & 3 & 1 & 1 & 2 & 2 & 0  & & & & \\ \hline
    Ignorance \ignorancegame & & & & & & & & & & & & 0 & 2 & 1 & \\
    Horseplay \horseplaygame &  & & & & & & & &  & & & 2 & 0 & 1 & \\
    Dress \dressgame & & & & & & & & & & & & 1 & 1 & 0 & \\ \hline
    Zero \nullgame & & & & & & & & & & & & & & & 0 \\
\end{tabular}
}
    \caption{Table showing the $L_1$ distance between games in the best-response-invariant embedding. Distances between nontrivial and trivial games are left undefined.}
    \label{tab:equi_distance}
\end{table}




\section{Game Visualization}

2×2 embeddings can be utilized to produce visualizations of 2×2, $|\mathcal{A}_1| \times |\mathcal{A}_2|$, two-player extensive-form, n-player polymatrix, and n-player extensive-form games (Figure~\ref{fig:game_visualization_examples}). Equilibrium and payoff properties can be directly deduced from these visualizations.

\begin{figure}[t!]
    \centering
    \begin{subfigure}[t]{0.12\textwidth}
        \includegraphics[scale=0.9]{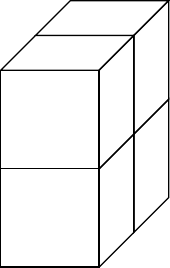}
        \caption{}
    \end{subfigure} \hfill
    \begin{subfigure}[t]{0.16\textwidth}
        \includegraphics[scale=0.9]{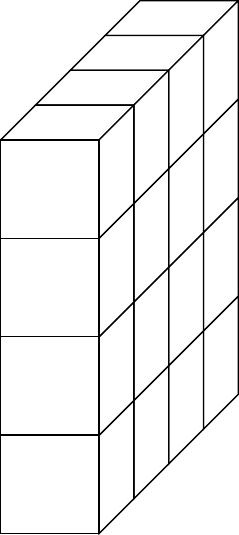}
        \caption{}
    \end{subfigure} \hfill
    \begin{subfigure}[t]{0.34\textwidth}
        \includegraphics[scale=0.9]{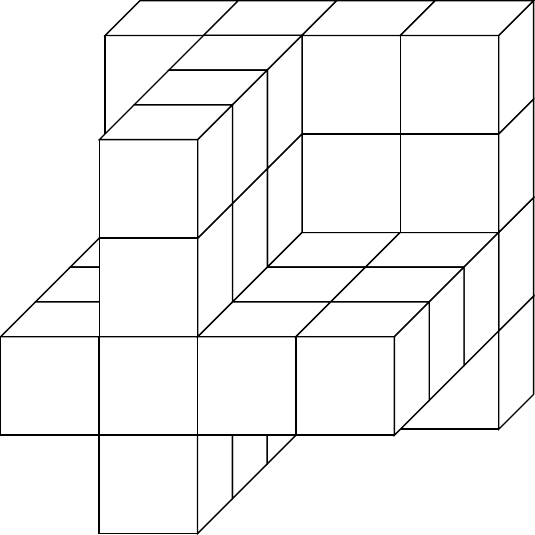}
        \caption{}
    \end{subfigure} \hfill
    \begin{subfigure}[t]{0.34\textwidth}
        \includegraphics[scale=0.9]{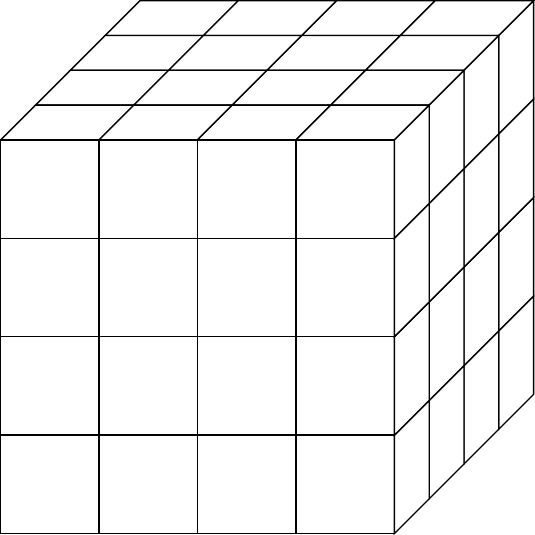}
        \caption{}
    \end{subfigure}
    \caption{Examples of different games that can be visualized using the techniques in this work: (a) 2×2 games, (b) two-player games, and (c) polymatrix games. Each cube represents a joint strategy in the game. Normal-form games with more that two-players (d) cannot be directly visualized with the techniques described and have to be approximated around a joint strategy using a local polymatrix approximation.}
    \label{fig:game_visualization_examples}
\end{figure}

\subsection{2×2 Game Visualization}
\label{sec:2x2vis}

First, consider the space of joint strategies, $\sigma(a)$, which can be denoted with a flat vector $\boldsymbol{\sigma} = [\sigma(a^{AA}), \sigma(a^{AB}), \allowbreak\sigma(a^{BA}), \sigma(a^{BB}))]$, where $a^{IJ} = (a_1^I, a_2^J)$. The standard constraints on a probability distribution apply: probabilities are nonnegative, $\sigma(a) \geq 0~ \forall a \in \mathcal{A}$, and sum to unity, $\sum_{a \in \mathcal{A}} \sigma(a) = 1$. The unity sum constraint means that a distribution over the four strategies of a 2×2 game, can be expressed with only three variables because one is redundant given the rest (e.g. $\sigma(a^{BB}) = 1 - \sigma(a^{AA}) - \sigma(a^{AB}) - \sigma(a^{BA})$). This means it is possible to visualize a joint distribution with four components in only three dimensions, by ignoring the space of distributions that do not sum to unity. Typically this is accomplished by specifying four vertices of a tetrahedron (a three dimensional object). Points in the simplex are then described in terms of mixtures of these four vertices (known as a barycentric coordinate system \cite{mobius1827_barycentric_coordinates}). Barycentric coordinates, $\boldsymbol{\sigma}$, can be converted to Cartesian coordinates, $\boldsymbol{x}$, via a linear transform, $\boldsymbol{x} = T \boldsymbol{\sigma}$. The columns of $T$ are points of a regular tetrahedron. There are many ways to choose points of a tetrahedron, one is given in Equation~\eqref{eq:tetrahedron}.
\begin{align} \label{eq:tetrahedron}
    \boldsymbol{x} = T \boldsymbol{\sigma} \qquad\qquad T = \begin{bmatrix}
        -\frac{1}{2} & \frac{1}{2} & 0 & 0 \\
        -\frac{\sqrt{3}}{6} & -\frac{\sqrt{3}}{6} & \frac{\sqrt{3}}{6} & 0 \\
        -\frac{\sqrt{6}}{12} & -\frac{\sqrt{6}}{12} & -\frac{\sqrt{6}}{12} & \frac{\sqrt{6}}{4} 
    \end{bmatrix}
\end{align}

Each nonnegative inequality constraint geometrically corresponds to a normal vector in the equation of a plane (e.g. $[1, 0, 0, 0] \boldsymbol{\sigma}^T \geq \boldsymbol{0}$). These planes split the space into two halves: those that are feasible distributions and those that are infeasible distributions. Together the four nonnegative probability inequality constraints result in a convex polytope with four faces, specifically a regular tetrahedron, when visualized in three dimensions (Figure~\ref{fig:joint_tetrahedron_simplex}). $\sigma(a_1, a_2)$ corresponds to a full joint distribution.  A subset of joints that factorize into their marginals, $\sigma(a_1, a_2) = \sigma(a_1)\sigma(a_2)$, is worth highlighting because of its relationship to NEs which by definition have to factorize. Factorizable joints result in a manifold within the tetrahedron (Figure~\ref{fig:joint_tetrahedron_manifold}).

\begin{figure}[t]
    \centering
    \begin{subfigure}[b]{0.32\linewidth}
        \centering
        \includegraphics[width=1.0\textwidth]{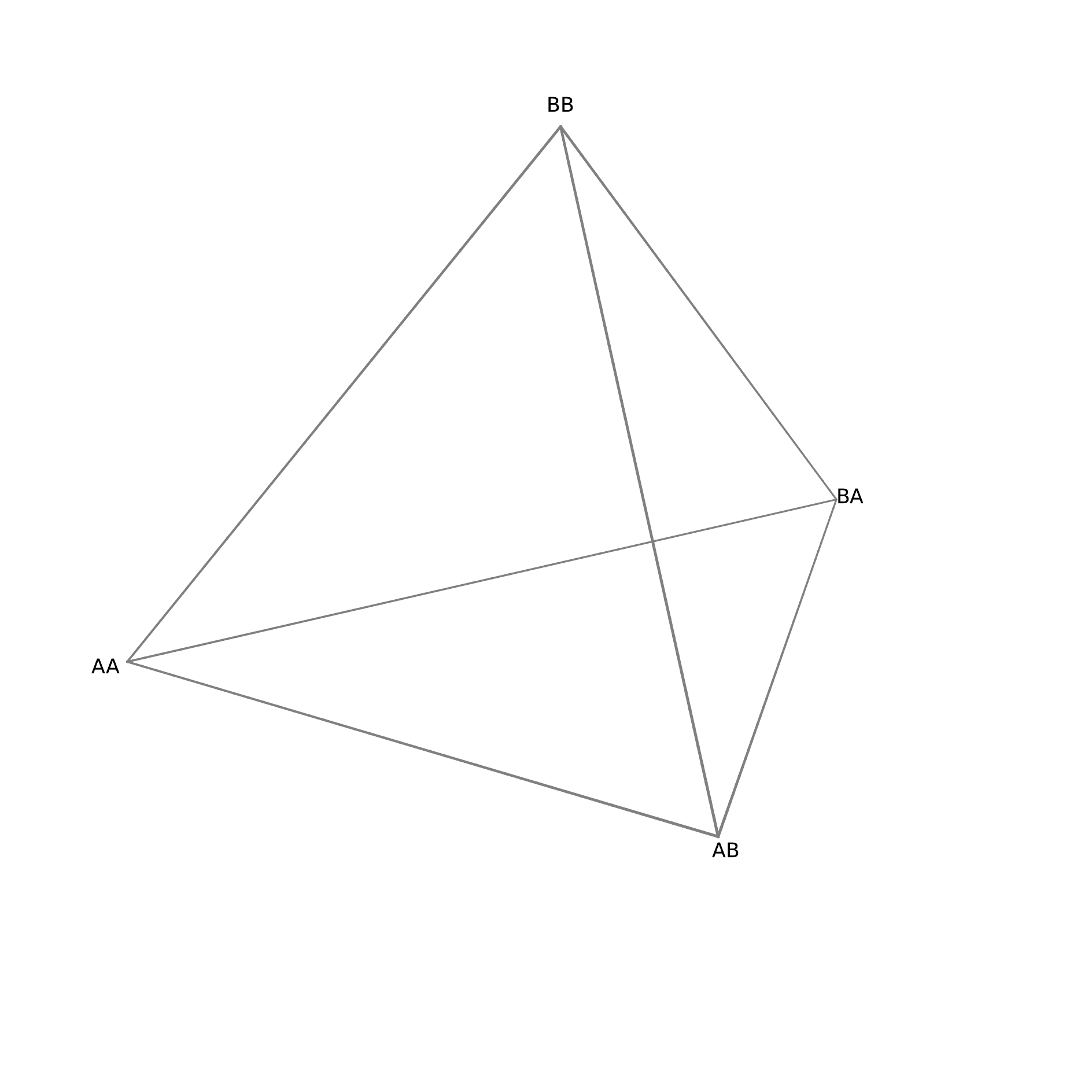}
        \vspace{-4em}
        \caption{Probability Simplex}
        \label{fig:joint_tetrahedron_simplex}
    \end{subfigure} ~~~
    \begin{subfigure}[b]{0.32\linewidth}
        \centering
        \includegraphics[width=1.0\textwidth]{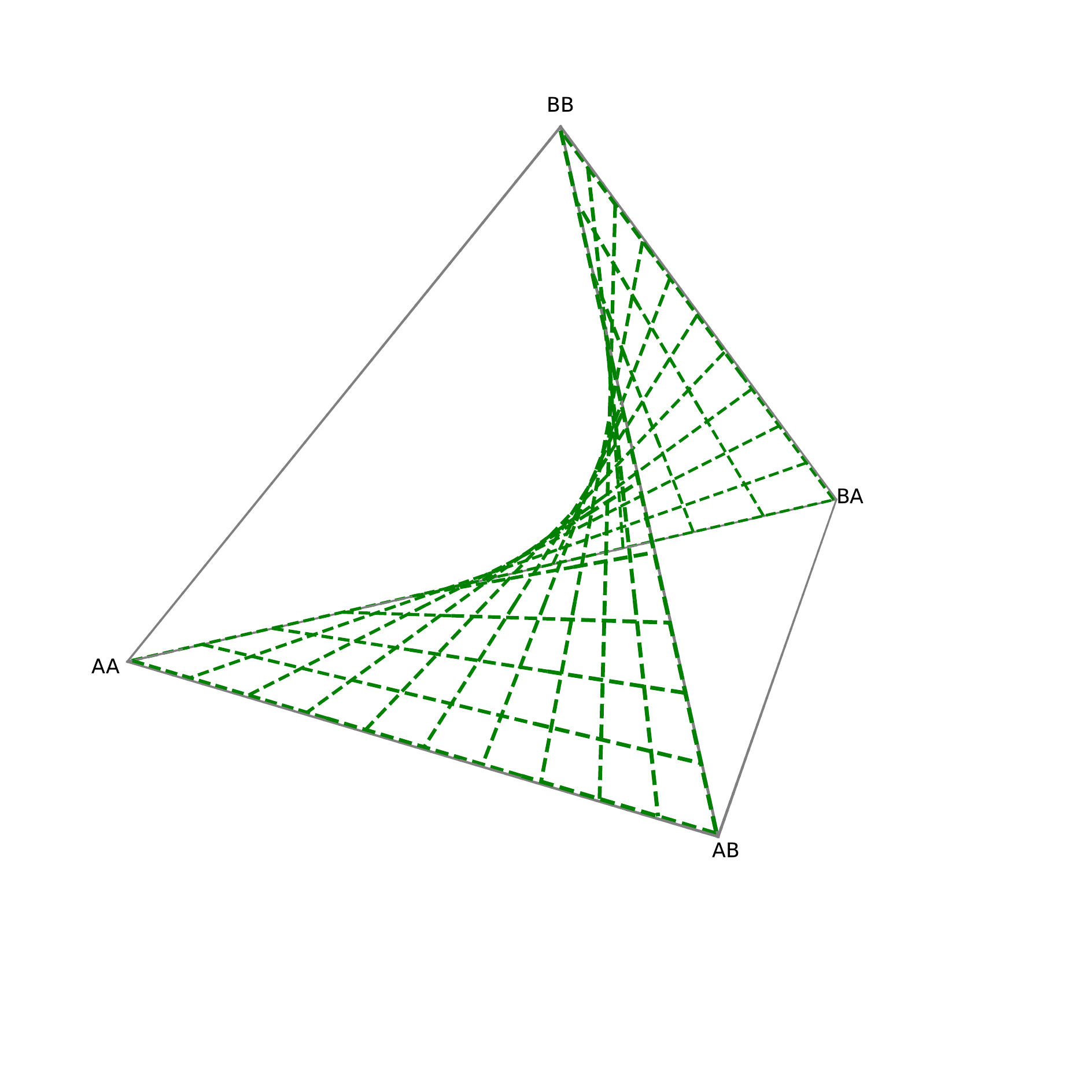}
        \vspace{-4em}
        \caption{NE Manifold}
        \label{fig:joint_tetrahedron_manifold}
    \end{subfigure}
    \caption[Probability Simplex and NE Manifold]{Visualization of the space of valid joint distribution, $\sigma(a_1, a_2)$, (tetrahedron) and valid factorizable joint distributions, $\sigma(a_1, a_2) = \sigma(a_1) \sigma(a_2)$, (manifold). The vertices of the tetrahedron correspond to pure joint strategies. The interior of the tetrahedron corresponds to mixed joint strategies.}
    \label{fig:joint_tetrahedron}
\end{figure}

Now consider a player's payoff, $G_p(a_1, a_2)$. We know that each player's 2×2 equilibrium-invariant embedding is parameterized by an angle, $G^\text{equil}_p(\theta_p)$, on a circle. This circle can be meaningfully traced in three-dimensions: $x(\theta_p) = T \boldsymbol{g}^\text{equil}_p(\theta_p)$. A particular payoff for a player, $G^\text{equil}_p$, can be represented by drawing an arrow from the origin to a point on this circle. The direction of this arrow conveys meaning: it is the direction which linearly increases the equilibrium-invariant embedding payoff received under a joint, and it points to the region where equilibria are likely to reside. The original payoffs, $T \boldsymbol{g}_p(\theta_p)$, could be any vectors perpendicular to the plane that the circle lies on.

The deviation gains (Equations~\eqref{eq:wsce_def}, \eqref{eq:ce_def}, and \eqref{eq:cce_def}) of the various equilibrium concepts can also be visualized. Each row of the deviation gain matrix corresponds to a half plane. The rows in aggregate make a convex polytope, which can also be visualized in three dimensions. NEs are where this polytope intersects with the factorizable joint manifold.

We propose a visualization of 2×2 games that includes their equilibrium-invariant embedding, best-response-invariant embedding, (C)CE polytope and NE set (Figure~\ref{fig:twoxtwo_polytope}). When the arrows are near opposite, the game is invariant-zero-sum (for example  \coordinationgame~Cycle, Figure~\ref{fig:cycle_polytope}). When the arrows are near alignment, the game is invariant-common-payoff (for example  \coordinationgame~Coordination, Figure~\ref{fig:coordination_polytope}). Dominant games have near perpendicular vectors. Sometimes a continuum of NEs is feasible, when that is the case it will be shown with a dashed line (for example  \horseplaygame~Horseplay, Figure~\ref{fig:horseplay_polytope}).

\begin{figure}[t]
    \centering
    \begin{subfigure}[t]{0.32\textwidth}
        \includegraphics[width=1.0\textwidth]{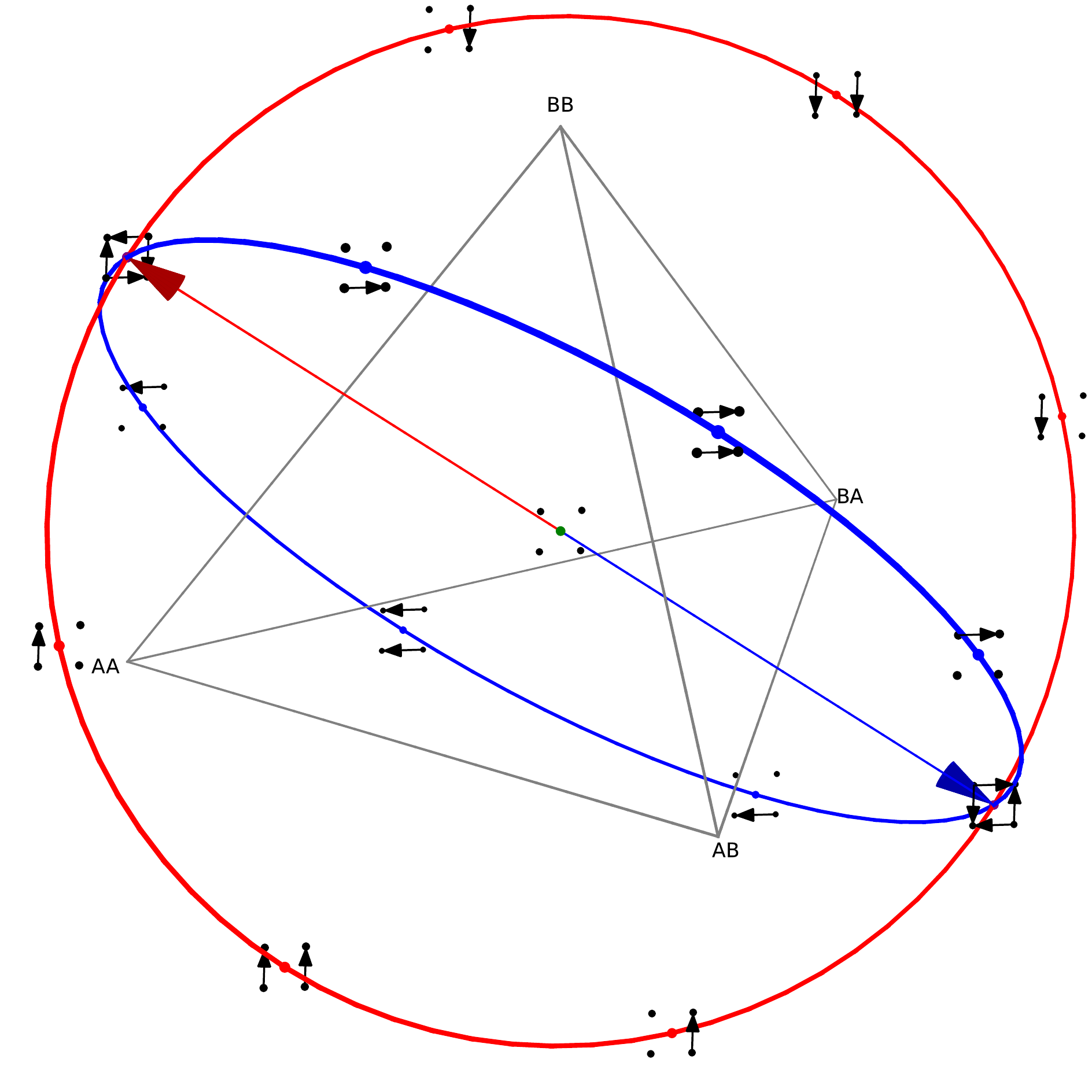}
        \caption{\cyclegame~Cycle Polytope}
        \label{fig:cycle_polytope}
    \end{subfigure} ~
    \begin{subfigure}[t]{0.32\textwidth}
        \includegraphics[width=1.0\textwidth]{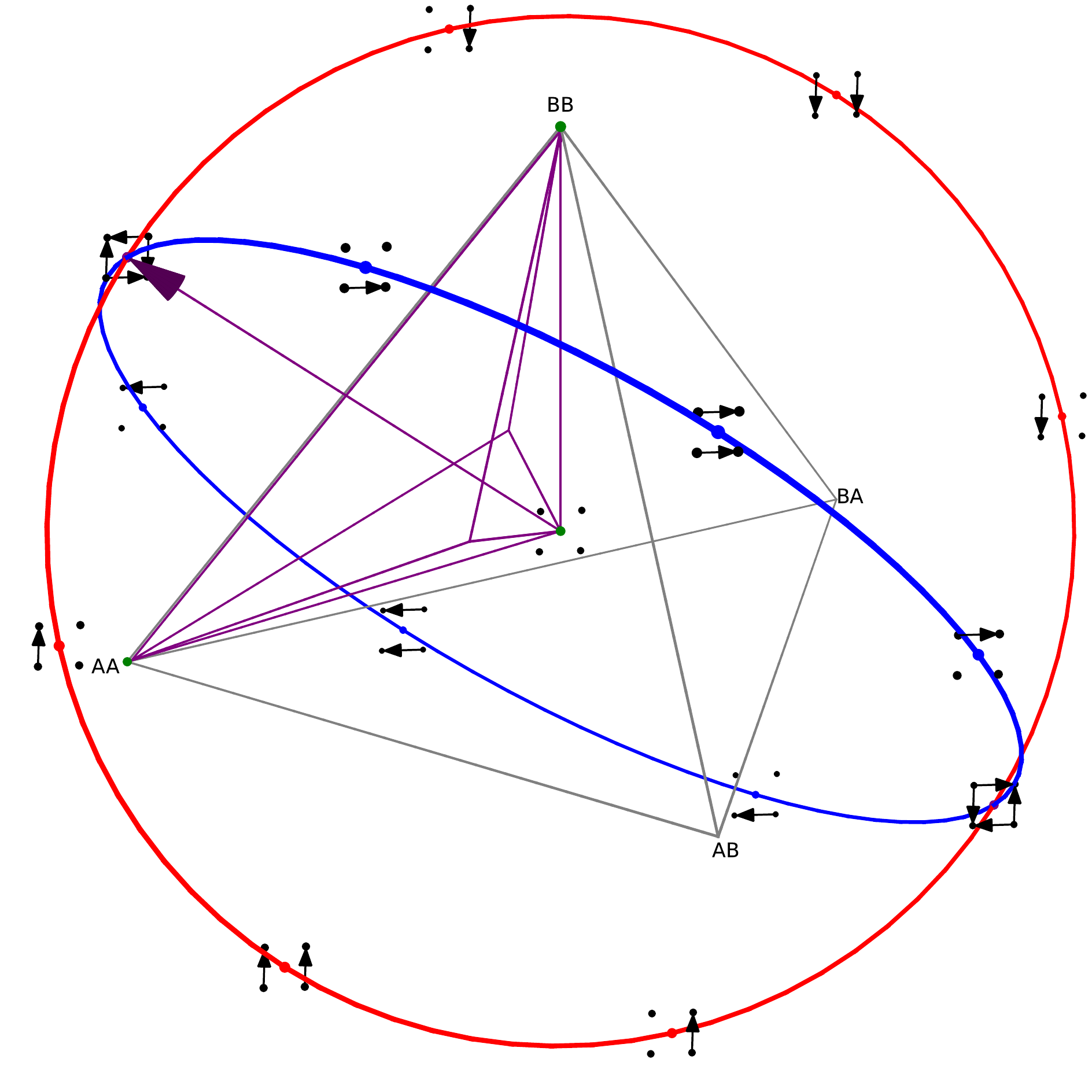}
        \caption{\coordinationgame~Coordination Polytope}
        \label{fig:coordination_polytope}
    \end{subfigure} ~
    \begin{subfigure}[t]{0.32\textwidth}
        \includegraphics[width=1.0\textwidth]{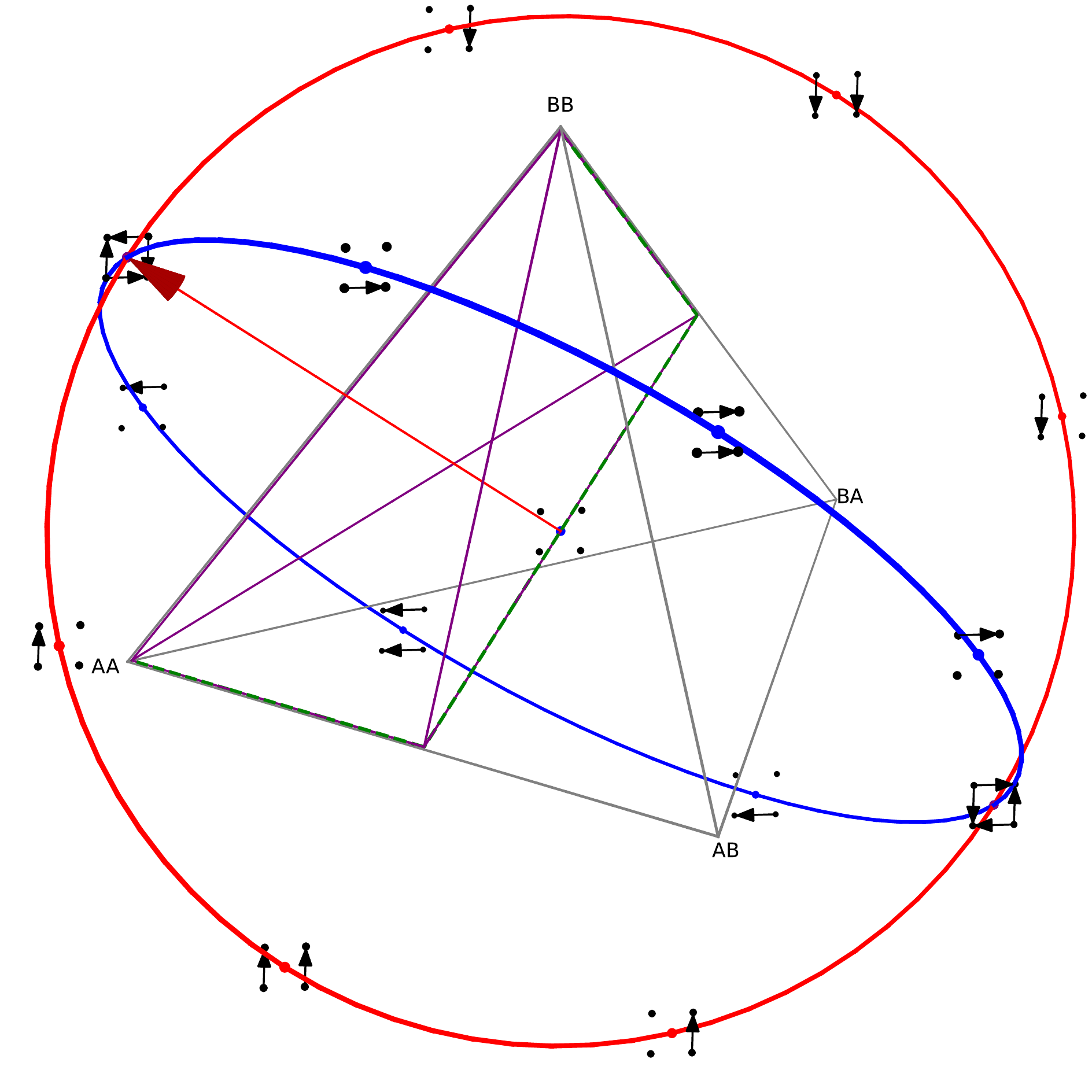}
        \caption{\horseplaygame~Horseplay Polytope}
        \label{fig:horseplay_polytope}
    \end{subfigure}
    
    \begin{subfigure}[t]{0.32\textwidth}
        \includegraphics[width=1.0\textwidth]{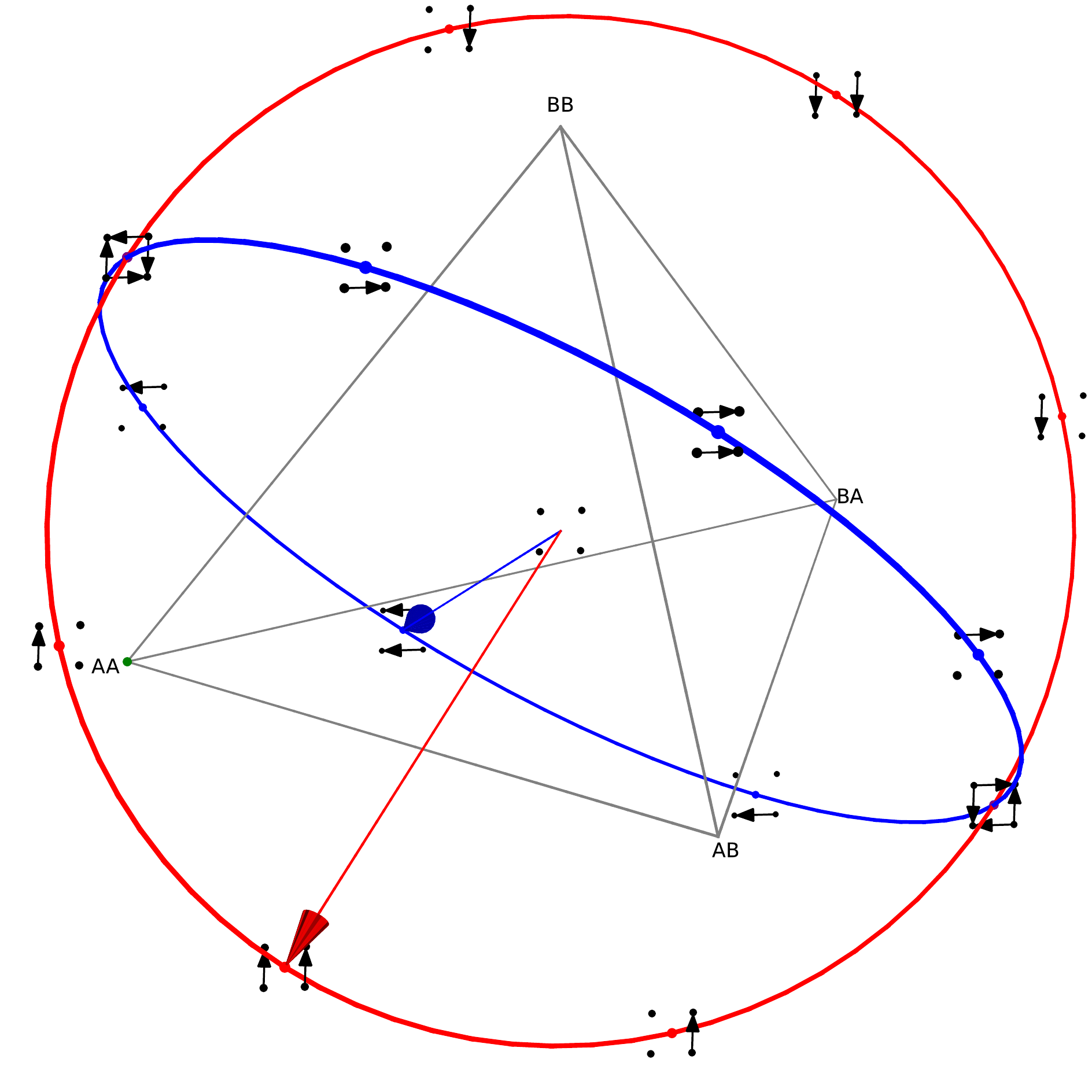}
        \caption{\cyclegame~Dominant Polytope}
        \label{fig:dominant_polytope}
    \end{subfigure} ~
    \begin{subfigure}[t]{0.32\textwidth}
        \includegraphics[width=1.0\textwidth]{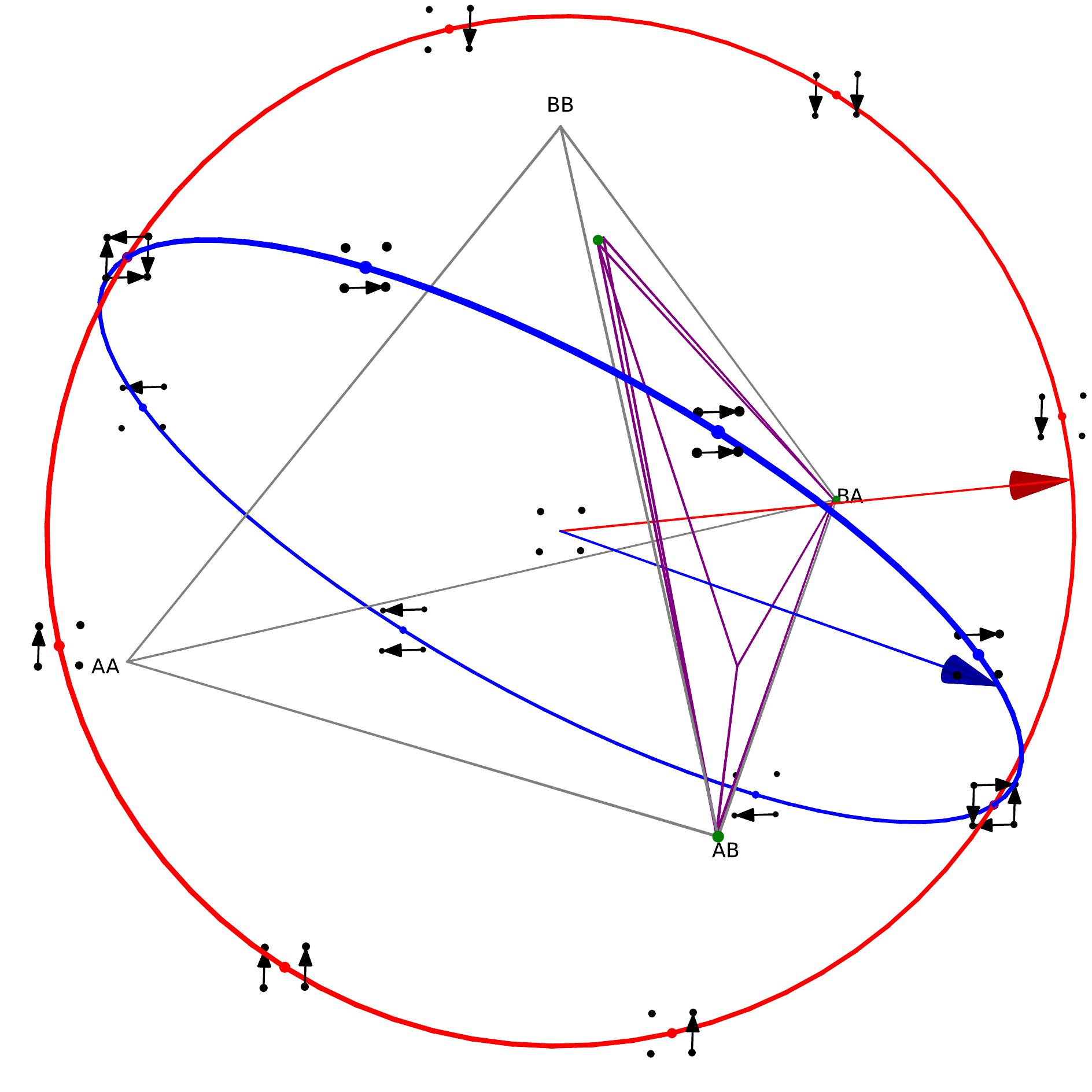}
        \caption{\chickengame~Chicken Polytope}
        \label{fig:chicken_polytope}
    \end{subfigure} ~
    \begin{subfigure}[t]{0.32\textwidth}
        \includegraphics[width=1.0\textwidth]{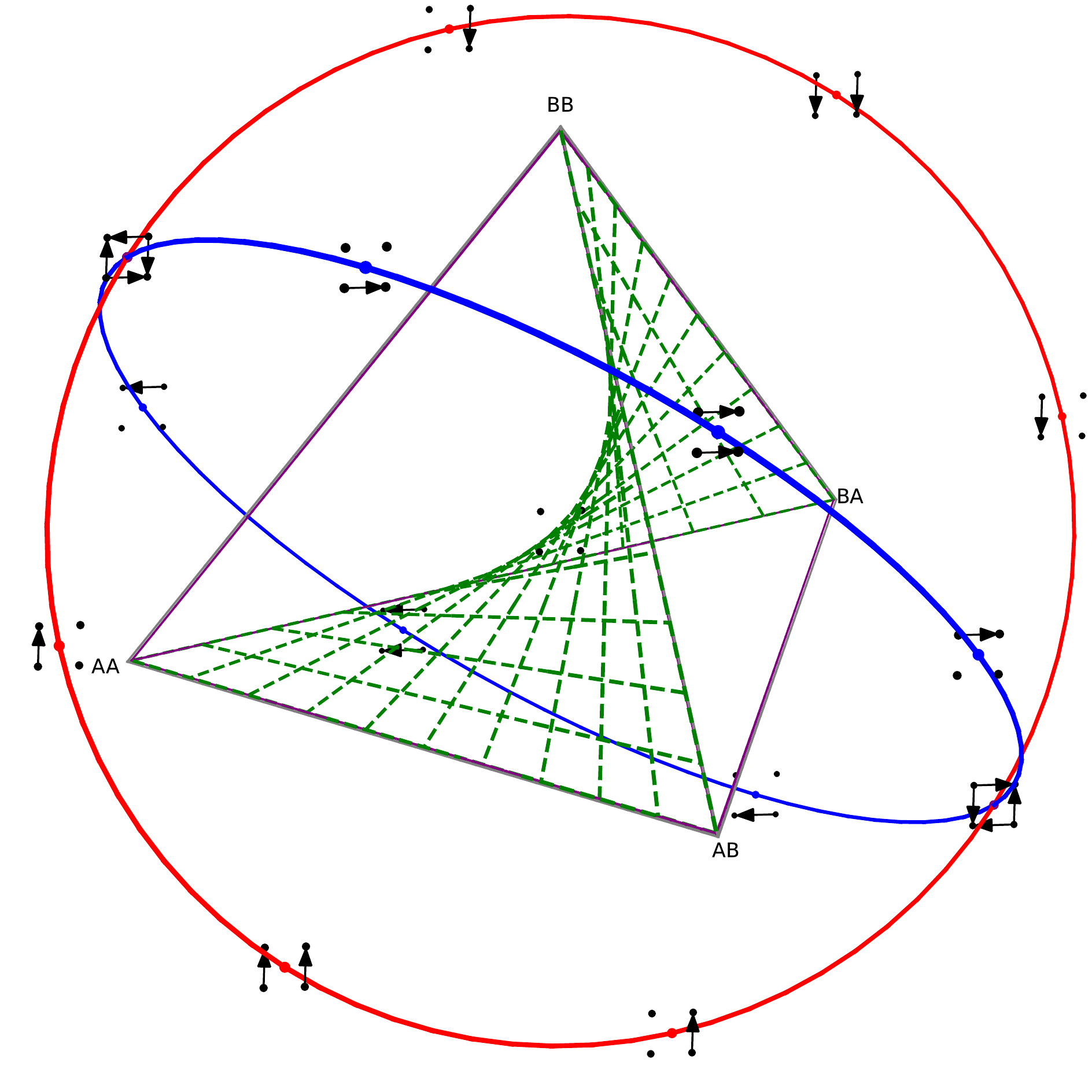}
        \caption{\nullgame~Null Polytope}
        \label{fig:null_polytope}
    \end{subfigure}
    \caption{Visualization of 2×2 games. The $\theta_1$ (red) and $\theta_2$ (blue) parameters are shown on their unit circles. The arrows are the invariant unit vectors of each player's payoff so the direction indicates joints that will linearly increase payoff. The (C)CE polytope of feasible equilibria is shown in purple. Green shows NEs.}
    \label{fig:twoxtwo_polytope}
\end{figure}

\subsection{Two-Player Game Visualization}
\label{sec:twoplayervis}

\begin{figure}[t!]
    \centering
    \stratsymplot{width=\linewidth,image=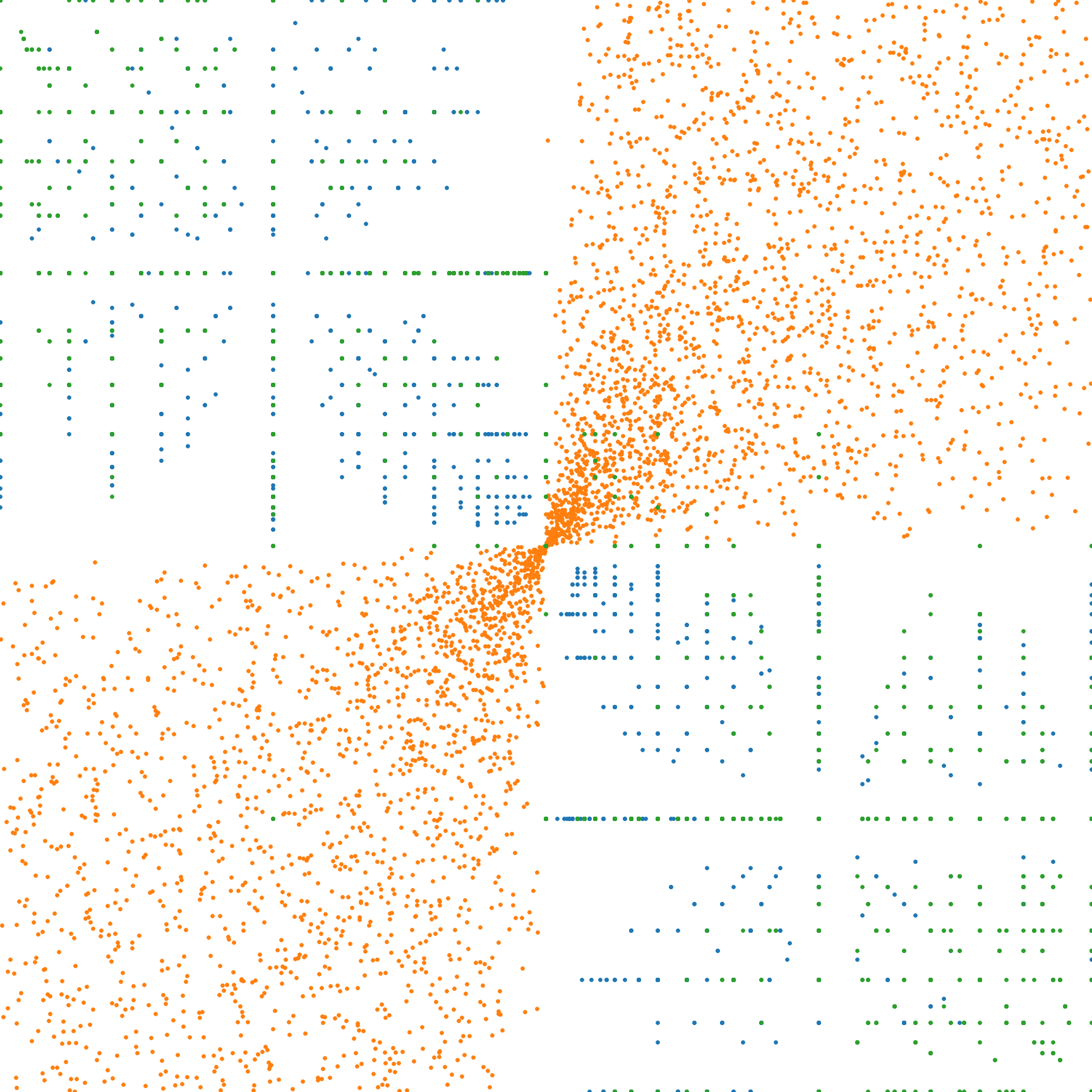,no partial trivial embedding}
    \caption{Visualization of the strategy space of asymmetric, two-player, extensive-form games: \colorsquare{tabB1color} zero-sum Kuhn Poker, \colorsquare{tabB2color} common-payoff Tiny Bridge, and \colorsquare{tabB3color} mixed-motive Sheriff. Zero-sum games only have interactions in the diagonal quadrants. Common-payoff games only have interactions in the off-diagonal quadrants. Mixed motive games can have interactions in all quadrants. The games are asymmetric so we can utilize strategy permutation but not player permutation to plot the interactions on a reduced invariant embedding.}
    \label{fig:mxm_asymmetric}
\end{figure}

A two-player game with more than two strategies ($|\mathcal{A}_1| \times |\mathcal{A}_2|$), many strategies (such as an extensive-form\footnote{All extensive-form games have an equivalent normal-form representation. However, subgame perfection is not representable in normal-form.} game with deterministic policies), or even infinitely many strategies (such as an extensive-form game with stochastic policies), can be summarized by considering either every possible 2×2 subgame or approximated by sampling subgames, $\tilde{G}_p$. To produce a single point, uniformly sample strategies $\tilde{a}_1^A$, $\tilde{a}_1^B$, $\tilde{a}_2^A$, and $\tilde{a}_2^B$, to produce a sampled 2×2 payoff. Then find the equilibrium-invariant embedding $(\tilde{\theta}_1, \tilde{\theta}_2)$, using symmetries if appropriate.
\begin{align}
    \tilde{G}^\text{global}_p(\tilde{a}_1^A, \tilde{a}_1^B, \tilde{a}_2^A, \tilde{a}_2^B) = \begin{bmatrix}
        G_p(\tilde{a}_1^A,\tilde{a}_2^A) & G_p(\tilde{a}_1^A,\tilde{a}_2^B) \\
        G_p(\tilde{a}_1^B,\tilde{a}_2^A) & G_p(\tilde{a}_1^B,\tilde{a}_2^B)
    \end{bmatrix} \to (\tilde{\theta}^\text{global}_1, \tilde{\theta}^\text{global}_2)
\end{align}
Large two-player games can be represented as a point cloud of sampled equilibrium-invariant embeddings. We propose that properties of the game can be deduced from the positions and density of the the embeddings in this plot. We call this sampling scheme \emph{global} because all strategies are uniformly sampled over the full normal-form game being analysed.

\paragraph{Extensive-Form Games}  We visualize (Figure~\ref{fig:mxm_asymmetric}) three two-player extensive-form games taken from the OpenSpiel library \citep{lanctot2019_openspiel}. All extensive-form games can be converted to normal-form games by considering the enumeration of pure-strategy policies available to each player. To produce our visualization we do not require enumeration of all policies; instead we can simply sample random pure joint policies to approximate the visualization allowing us to scale analysis to large extensive-form games. Because all the games considered in this visualization are asymmetric, we do not utilize player symmetry. We sample strategies arbitrarily from the payoffs, so the order of strategies is also arbitrary and we can utilize strategy symmetry. Therefore, the plots cover the domain: $-\frac{\pi}{2} \leq \theta_p < \frac{\pi}{2}$.

Kuhn Poker \citep{kuhn1950_poker} is an asymmetric, zero-sum simplified poker variant. First, notice that all points lie in invariant-zero-sum quadrants, which is what the theory predicts for a two-player zero-sum extensive-form game. Interestingly, many boundary classes (\aidosgame~Aidos, \picnicgame~Picnic, and \hazardgame~Hazard) appear. This shows that this game has situations where players are indifferent between strategies. Kuhn poker does have a high concentration of points in the \dominantgame~Dominant and \samaritangame~Samaritan classes indicating that in most situations there are obvious strategies that either both or at least one of the players should choose. There are also a significant number of \cyclegame~Cycle interactions which indicates that this game has some interesting strategic depth.

Two-player Cooperative Tiny Bridge \citep{lockhart2020_bridge} is a common-payoff game. As expected, all points lie in the invariant-common-payoff quadrants of the visualization. The game also has a high concentration of points near the origin indicating that most strategic interaction is \dominantgame~Dominant: both players have a strict preference between strategies. \coordinationgame~Coordination also features prominently, indicating there are situations that require players to coordinate to maximize payoffs. It seems no boundary classes appear in Tiny Bridge meaning that it is rare for players to be indifferent over their strategies.

Sheriff \citep{farina2019_sheriff} is an asymmetric, mixed-motive game. The point cloud of this game covers both invariant-common-payoff and invariant-zero-sum quadrants, however it predominately occupies the invariant-zero-sum quadrant indicating that the game is more competitive than cooperative. Again, boundary classes (\aidosgame~Aidos, \picnicgame~Picnic, and \hazardgame~Hazard) appear, indicating times when players are indifferent.

\begin{figure}[t!]
    \centering
    \begin{subfigure}[t]{0.48\textwidth}
        \playersymzerosumplot{width=\linewidth,image=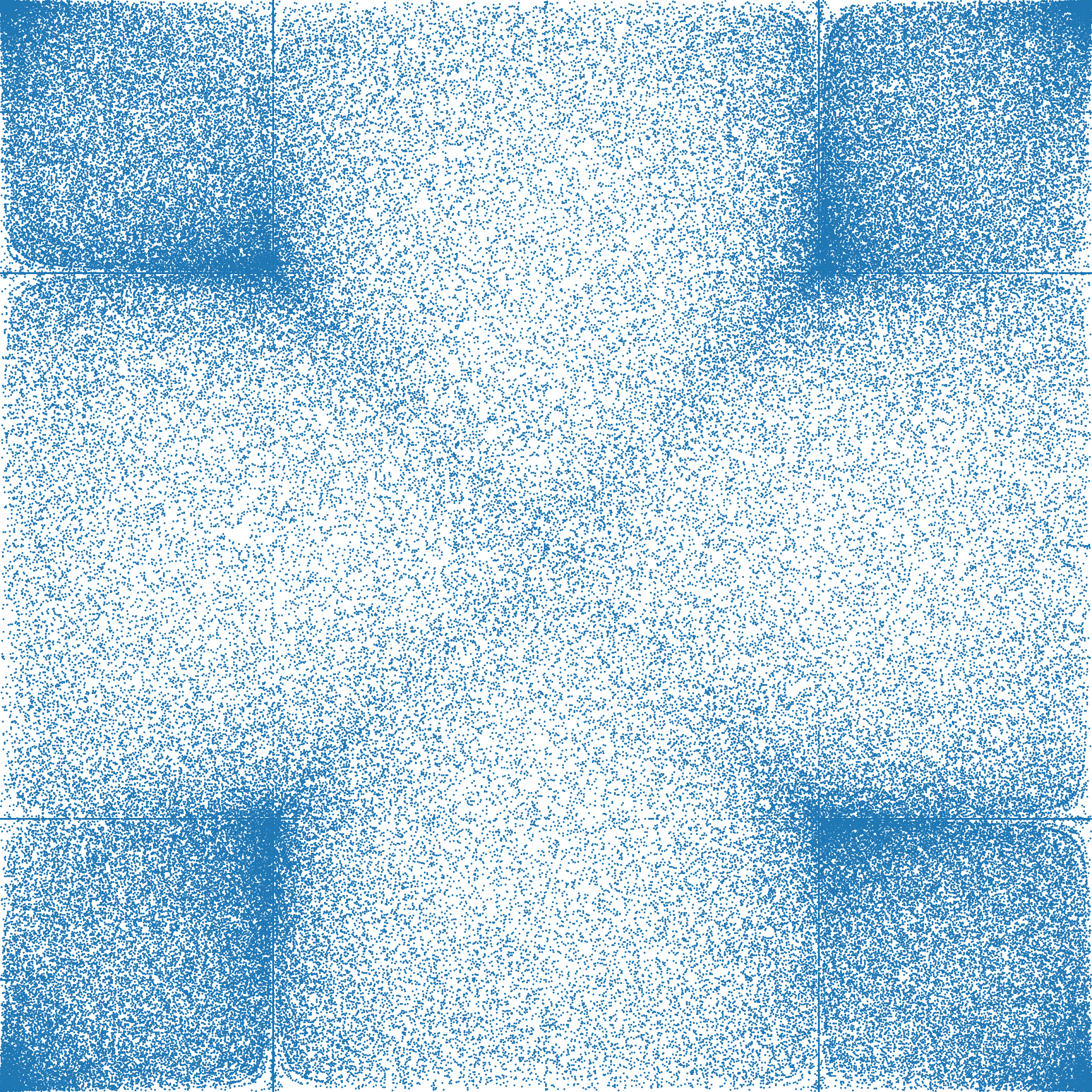,font size=\footnotesize}
        \caption{Approximate summary of the full $888 \times 888$ AlphaStar league \citep{vinyals2019_starcraft}. Each strategy corresponds to a learned policy in the extensive-form game and the payoff was calculated empirically. The plot is redundant outside of the domain $0 \leq \theta_1 \leq \frac{\pi}{2}$ and $-\frac{\pi}{2} \leq \theta_2 \leq 0$.}
        \label{fig:alphastar_global}
    \end{subfigure} \hfill
    \begin{subfigure}[t]{0.48\textwidth}
        \playersymzerosumplot{width=\linewidth,image=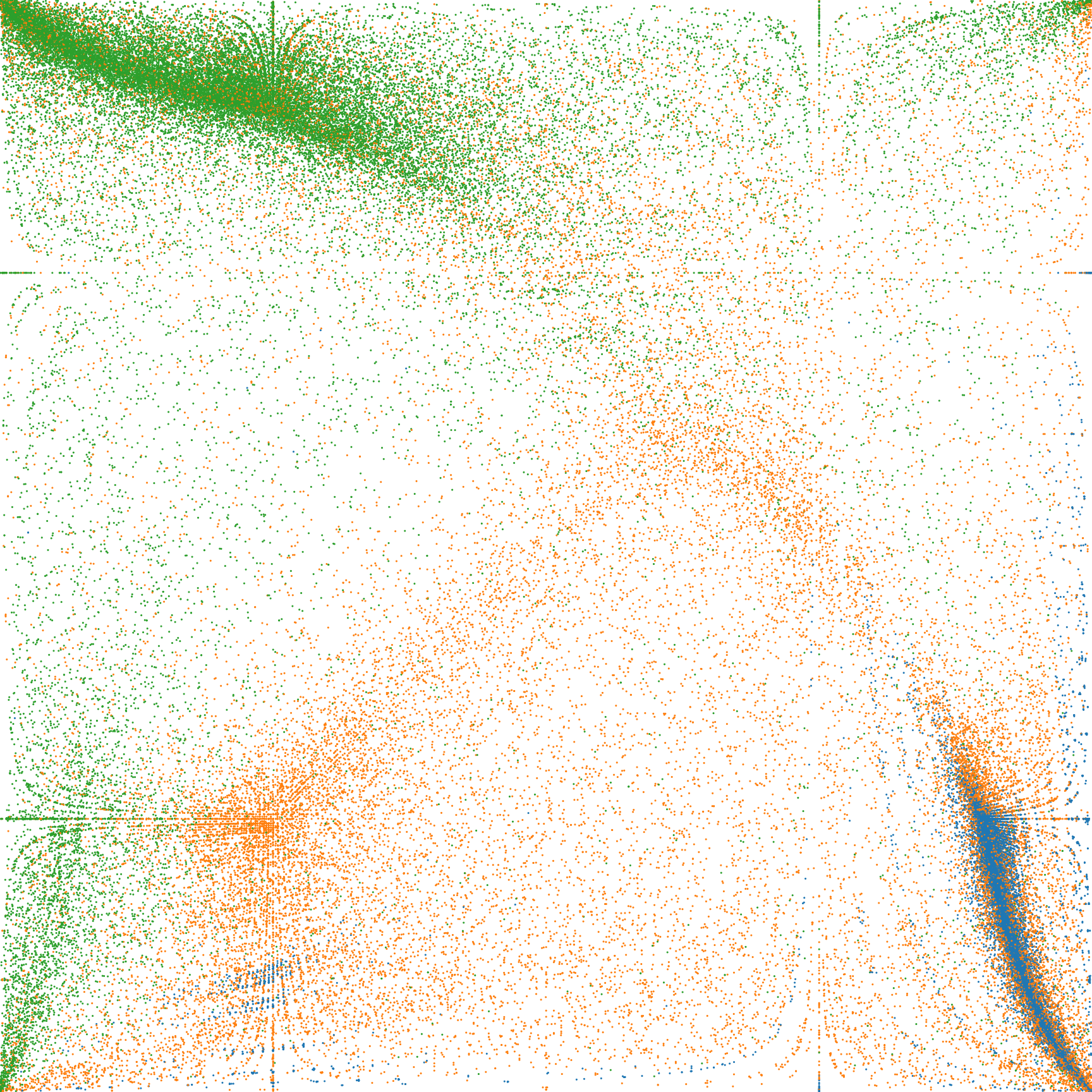,font size=\footnotesize}
        \caption{Local approximation of the AlphaStar league. The first strategy is fixed for both players, and only the second sampled. Key: \colorsquare{tabB1color} early training, \colorsquare{tabB2color} middle training, and \colorsquare{tabB3color} late training.}
        \label{fig:alphastar_local}
    \end{subfigure}
    \caption{Analysis of a large empirical normal-form game. In early and late training, the local strategic space is mostly transitive (Dominant), and in mid training the space is cyclic (Cycle). This matches \cite{czarnecki2020_spinning}'s ``spinning top'' hypothesis in games of skill.}
    \label{fig:alphastar}
\end{figure}

\paragraph{AlphaStar League} The AlphaStar league is a symmetric zero-sum, $888\times888$, empirical normal-form game. It consists of payoffs between policies learned by AlphaStar \citep{vinyals2019_starcraft}. The visualization of this game is shown in Figure~\ref{fig:alphastar_global}. The game is symmetric, so we can utilize player symmetry. The game is zero-sum, so we only need to consider zero-sum quadrants. This plot is also strategy invariant, so we could have utilized strategy symmetry. However, to have a simpler comparison to Figure~\ref{fig:alphastar_local}, we do not use strategy symmetry. Therefore the plot is over the domain: $0 \leq \theta_1 \leq \pi$ and $-\pi \leq \theta_2 \leq 0$, although outside of $0 \leq \theta_1 \leq \frac{\pi}{2}$ and $-\frac{\pi}{2} \leq \theta_2 \leq 0$ the plot is redundant. The empirical game mainly has \dominantgame~Dominant strategic interaction, although there are significant \samaritangame~Samaritan and \cyclegame~Cycle components.

\cite{czarnecki2020_spinning} hypothesised that games of skill have few cyclic interactions early in training as policies can easily transitively improve. In mid training, there are numerous reasonably competent and diverse policies, which results in interesting cyclic interactions. However, in late training, the interactions between policies become less cyclic, as the skills needed to play the game are perfected. \cite{czarnecki2020_spinning} describes this phenomenon as a ``spinning top'' shaped distribution of cycles. We test this hypothesis in our visualization by using an alternative game sampling scheme: fixing the first strategy, $a^A$, for both players, and sampling the second strategy for each player randomly, $\tilde{a}^B_p$.
\begin{align}
    \tilde{G}^\text{local}_p(\tilde{a}^B_1,\tilde{a}^B_2) = \begin{bmatrix}
        G_p(a^A,a^A) & G_p(a^A,\tilde{a}^B_2) \\
        G_p(\tilde{a}^B_1,a^A) & G_p(\tilde{a}^B_1,\tilde{a}^B_2)
    \end{bmatrix} \to (\tilde{\theta}^\text{local}_1, \tilde{\theta}^\text{local}_2)
\end{align}
The AlphaStar league game is symmetric, so it is more natural to use the same fixed strategy for both players. Note that in general, we could have different fixed strategies for each player.
\begin{align}
    \tilde{G}^\text{local}_p(\tilde{a}^B_1,\tilde{a}^B_2) = \begin{bmatrix}
        G_p(a_1^A,a_2^A) & G_p(a_1^A,\tilde{a}^B_2) \\
        G_p(\tilde{a}^B_1,a_2^A) & G_p(\tilde{a}^B_1,\tilde{a}^B_2)
    \end{bmatrix} \to (\tilde{\theta}^\text{local}_1, \tilde{\theta}^\text{local}_2)
\end{align}
We call this sampling scheme \emph{local} because it is sampling over deviation strategies around a fixed background strategy $a^{AA} = (a_1^A, a_2^A)$. Such a sampling strategy will provide an indication of the strategic dynamics locally around these background strategies.

The strategies in the AlphaStar league game correspond to policies roughly ordered according to their training time. Later strategies are policies trained against distributions over previous policies. Therefore, later policies have more training time and have trained against more diverse opponents, and will more likely be stronger. But what does the strategic landscape of the game look like at each stage of training? By analysing the game around background strategies which correspond to early, middle and late training, we can test \cite{czarnecki2020_spinning}'s hypothesis (Figure~\ref{fig:alphastar_local}).

Unsurprisingly, \colorsquare{tabB1color} early training primarily occupies a \equigame{-1}{-1}{-1}{-1}~Dominant region, where players strictly wish to deviate from the their weakly trained fixed background policy. In contrast,  \colorsquare{tabB3color} late training primarily occupies the opposite \equigame{1}{1}{1}{1}~Dominant region, where players strictly prefer to stick with their fixed background strategies. \colorsquare{tabB2color} Mid training has a varied strategic space, but has more mass in the \cyclegame~Cycle region than the other background policies. Broadly, this analysis supports the ``spinning top'' hypothesis. A deeper analysis of the training timeline could uncover more structure.

\subsection{N-Player Polymatrix Game Visualization}
\label{sec:polymatrixvis}

It is possible to extend the two-player local subgame sampling technique to n-player games by defining a fixed background strategy for all players, $a^{A...A} = (a_1^A, ..., a_N^A)$, and then visualizing all pairwise player interactions.
\begin{align}
    \tilde{G}^\text{local}_{pq}(\tilde{a}^B_p,\tilde{a}^B_q) = \begin{bmatrix}
        G_p(a_p^A,a_q^A,a_{-p-q}^{...}) & G_p(a_p^A,\tilde{a}^B_2,a_{-p-q}^{...}) \\
        G_p(\tilde{a}^B_p,a_q^A,a_{-p-q}^{...}) & G_p(\tilde{a}^B_p,\tilde{a}^B_q,a_{-p-q}^{...})
    \end{bmatrix} \to (\tilde{\theta}^\text{local}_p, \tilde{\theta}^\text{local}_q) \qquad \forall p \neq q \in [1,N]
\end{align}
To avoid ambiguity with the player slots of the original n-player game we rename the representation parameters from $\theta_1$ and $\theta_2$ to $\theta_r$ and $\theta_c$ to indicate the row and column players. Because the order of the players and strategies matter in this pairwise approximation of an n-player game, it is necessary to use the full invariant space: $-\pi \leq \theta_r \leq +\pi$ and $-\pi \leq \theta_c \leq +\pi$. It is unnecessary to plot both permutations of the player pairs: the point cloud will be the same but mirrored over $\theta_r = \theta_c$.

Considering only pairwise player interaction is an established succinct game representation called the \emph{polymatrix approximation} \citep{janovskaja1968_polymatrix}. This representation is described by the tuple $(G_{1,2}(a_1, a_2), G_{1,3}(a_1, a_3),\allowbreak ...,\allowbreak G_{N-1,N}(a_N, a_{N-1}))$. By fixing the background strategies in an n-player normal-form game we create a \emph{local polymatrix approximation} around these background strategies. We need not restrict ourselves to strategies in a normal-form game. We can instead sample entire policies, $\tilde{\pi}_p^B$, in extensive-form games around fixed background policies, $\pi^{A...A} = (\pi_1^A, ..., \pi_N^A)$. Either stochastic or deterministic policies can be sampled (we sample deterministic policies in our experiments). Therefore the visualization can be used to approximate the strategic dynamics of a game around salient policies, such as the policies players are currently executing in a game, or perhaps ones that have been found using a learning method. We expect the local polymatrix approximation landscape to be different around different background policies, as policies influence transition dynamics in the game.
\begin{align}
    \tilde{G}^\text{local}_{pq}(\tilde{\pi}^B_p,\tilde{\pi}^B_q) = \begin{bmatrix}
        G_p(\pi_p^A,\pi_q^A,\pi_{-p-q}^{...}) & G_p(\pi_p^A,\tilde{\pi}^B_2,\pi_{-p-q}^{...}) \\
        G_p(\tilde{\pi}^B_p,\pi_q^A,\pi_{-p-q}^{...}) & G_p(\tilde{\pi}^B_p,\tilde{\pi}^B_q,\pi_{-p-q}^{...})
    \end{bmatrix} \to (\tilde{\theta}^\text{local}_p, \tilde{\theta}^\text{local}_q) \qquad \forall p \neq q \in [1,N]
\end{align}

\paragraph{Three-Player Leduc Poker} Leduc poker \citep{southey2005_leduc}, implemented in OpenSpiel \citep{lanctot2019_openspiel}, is a simplified Texas Hold'em implementation with $N$ players, two suits and $2(N+1)$ cards. We study the three-player game (Figure~\ref{fig:leduc_poker}), with each player respectively having the background policies: always raise, always call, and always fold. If any of these actions are infeasible at an information state, the policy will fall back to the call action. Unsurprisingly, given players must pay a blind, always fold is a losing strategy. This is most apparent when analysing \colorsquare{tab2color} call vs fold where call is almost always the best strategy to play. The majority of the points are in the basins of games along the $\theta_c$ axis: \equigame{1}{1}{1}{1}~\equigame{1}{1}{-1}{-1}~Dominant and \equigame{1}{1}{1}{-1}~\equigame{1}{1}{-1}{1}~Samaritan. This indicates two properties. Firstly, that call is almost always a dominant strategy for the row player (in the context of the other background policies), regardless of the column player's strategy. Secondly, the strength of the fold strategy depends on what the other players play. It is sometimes good (\equigame{1}{1}{1}{1} Dominant and \equigame{1}{1}{1}{-1} Samaritan), and sometimes bad (\equigame{1}{1}{-1}{-1} Dominant and \equigame{1}{1}{-1}{1} Samaritan). Fold is also poor in the context of \colorsquare{tab3color} fold vs raise. Deviating from fold is sometimes a better strategy (\equigame{-1}{-1}{1}{1}~Dominant and \equigame{-1}{-1}{1}{-1}~Samaritan) but there are situations where it is not (\equigame{1}{1}{1}{1}~Dominant and \equigame{1}{-1}{1}{1}~Samaritan). It is very rare for fold to be the worse strategy when playing against the other background policies, but be the preferred strategy against the deviation strategy (\equigame{-1}{1}{1}{1}~Samaritan). Generally, clustering around the $\theta_c$ axis indicates a strong row player and clustering around the $\theta_r$ axis indicates a strong column player. The \colorsquare{tab1color} raise vs call dynamics are interesting: the majority of the points lie in the anti-clockwise invariant-zero-sum region. Player 3, who always folds, plays such a weak strategy that they are an irrelevant player and the local raise vs call two-player game is almost perfectly zero-sum. Call is usually a good strategy against raise and a poor strategy against deviations from raise, and while raise is a poor strategy against call it is usually a good strategy against deviations, resulting in a \equigame{-1}{1}{1}{-1} Cycle. The majority of the points for all the local games lie in the invariant-zero-sum quadrants, which is unsurprising because this is a zero-sum game. Only two-player invariant-zero-sum games lie completely in the invariant-zero-sum quadrants.

\begin{figure}[t!]
    \centering
    \invarplot{no quadrant names,image=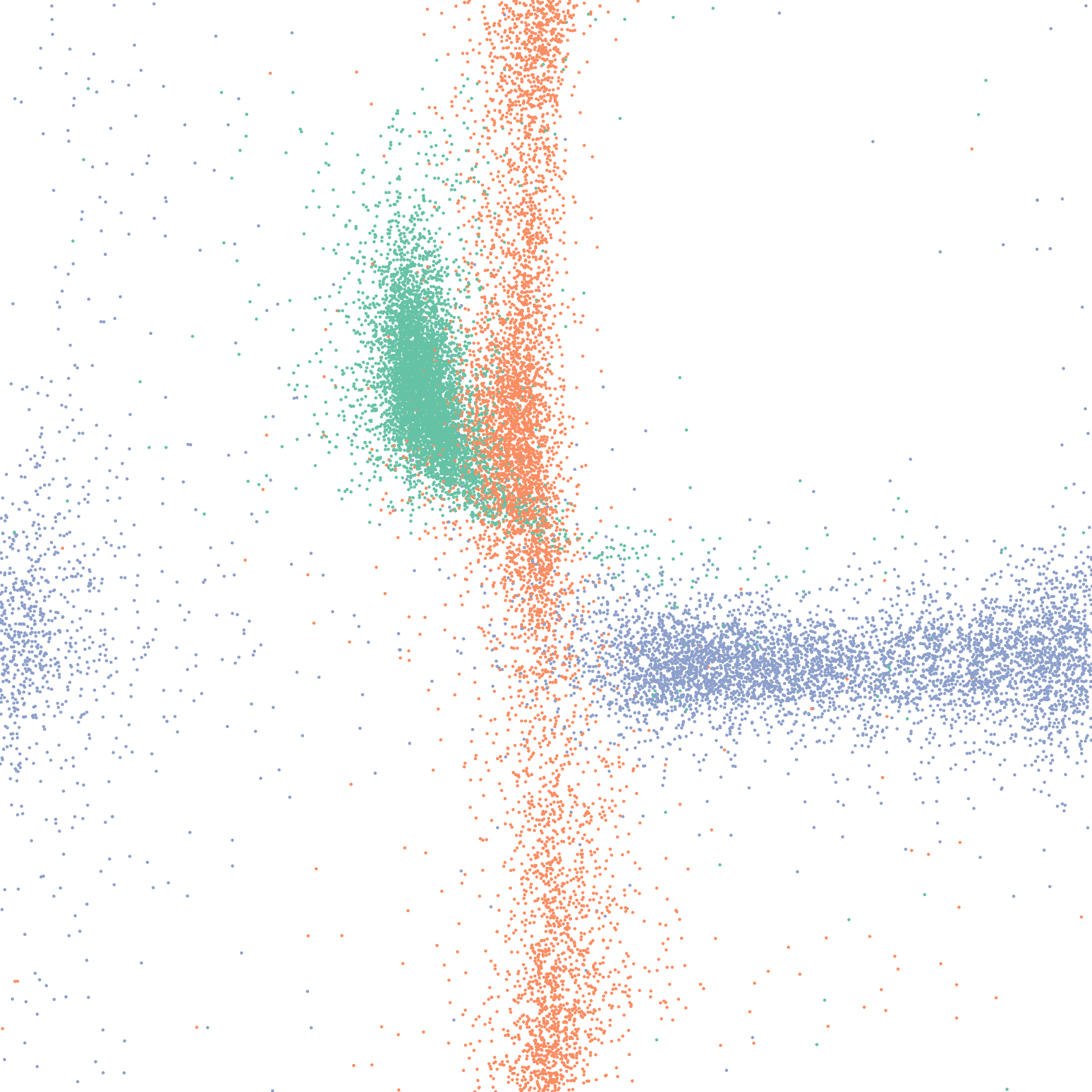,no partial trivial embedding,rc axes,no phi axes,no highlight symmetries,no symmetric games,no symmetries}
    \caption{Local polymatrix approximation visualization of three-player Leduc poker around the background strategies: always raise, always call, and always fold. Key: \colorsquare{tab1color} raise vs call, \colorsquare{tab2color} call vs fold, \colorsquare{tab3color} fold vs raise.}
    \label{fig:leduc_poker}
\end{figure}

\paragraph{Tiny Bridge 2vs2} Tiny Bridge is an extensive-form, two versus two, team game \citep{lockhart2020_bridge} implemented in OpenSpiel \citep{lanctot2019_openspiel} with zero-sum dynamics between teams and common-payoff dynamics within the team. The game is a simplified version of the popular card game Contract Bridge, but consists of only two suits each with four cards. Using the polymatrix visualization tools described above we study the dynamics of Tiny Bridge. In doing this we verify that the visualization a) captures interesting dynamics between teammates and opponents, and b) is able to differentiate between dynamics under different background policies. To produce quantifiably different policies, we generate twelve random deterministic policies for each player, compute an empirical normal-form game from their expected returns \citep{wellman2006_egta} (Figure~\ref{fig:tiny_bridge_payoffs}), and rate them using a game theoretic rating scheme \citep{marris2022_game_theoretic_rating} (Figure~\ref{fig:tiny_bridge_rating}). The ratings scheme is calculated on expected payoff under a joint equilibrium distribution (Figure~\ref{fig:tiny_bridge_cce}). We choose to use a maximum entropy CCE. We select the strongest policies as the ``best vs best'' background set (Figure~\ref{fig:tiny_bridge_polymatrix_visualization_best}) and the strongest of team 1 with the weakest of team 2 as the ``best vs worst'' background set (Figure~\ref{fig:tiny_bridge_polymatrix_visualization_mixed}). The local polymatrix approximation of each set is visualized using a different colour for each of the pairwise interactions. The mixed set contains (relatively) strong policies for team 1: they are good at countering their opponents and good at coordinating with each other. The strong set contains competent policies for all players.

The two teams, labeled 1 and 2, each have two players, labeled A and B. For both sets, as expected, the common-payoff intra-team dynamics (\colorsquare{tab1color}~1A vs 1B and \colorsquare{tab2color}~2A vs 2B) only lie in the off-diagonal quadrants and the zero-sum inter-team dynamics (\colorsquare{tab3color}~1A vs 2A, \colorsquare{tab4color}~1A vs 2B, \colorsquare{tab5color}~1B vs 2A, and \colorsquare{tab6color}~1B vs 2B) lie in the diagonal quadrants. Focusing on intra-team dynamics (\colorsquare{tab1color}~1A vs 1B and \colorsquare{tab2color}~2A vs 2B), for the best vs best background set the point cloud is concentrated around \dominantgame~Dominant and \equigame{1}{1}{1}{-1}~Samaritan. This indicates that both teams have good intra-team cooperation. However for the best vs worst background set, the worse team's policies (\colorsquare{tab2color}~2A vs 2B) have poor intra-team cooperation. Points are clustered around \equigame{-1}{-1}{-1}{-1}~Dominant and \equigame{-1}{-1}{-1}{1}~Samaritan, indicating that deviating away from the background policy is advantageous. Focusing on inter-team dynamics, the best vs worst background set shows that team 1 out-competes team 2 (\colorsquare{tab3color}~1A vs 2A, \colorsquare{tab4color}~1A vs 2B, \colorsquare{tab5color}~1B vs 2A, and \colorsquare{tab6color}~1B vs 2B). The points are clustered along the $\theta_c$ axis which indicates a stronger row player. Furthermore, the points are primarily in the basins of \equigame{1}{1}{-1}{1}~Samaritan and \equigame{1}{1}{-1}{-1}~Dominant: showing that the row player prefers its background policy over deviations and the column player prefers deviations over background policy. Player 1A's dynamics (\colorsquare{tab3color}~1A vs 2A and \colorsquare{tab4color}~1A vs 2B) are closer to the $\theta_c$ axis than player 1B's dynamics (\colorsquare{tab5color}~1B vs 2A and \colorsquare{tab6color}~1B vs 2B) indicating that out of the two players on team 1, A is more competitive. This is not surprising as 1A plays first which has a natural advantage in Bridge as they have first opportunity to bid and convey information. The inter-team dynamics in the best vs best background policies are more balanced.

\begin{figure}[t!]
    \centering
    \begin{subfigure}[t]{0.19\textwidth}
        \includegraphics[width=\textwidth]{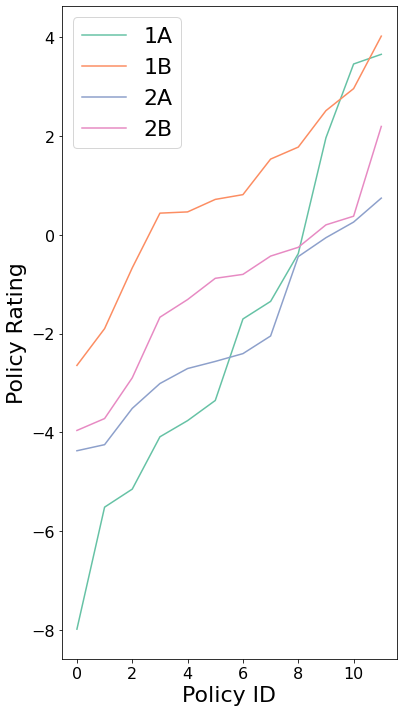}
        \caption{Ratings of Policies}
        \label{fig:tiny_bridge_rating}
    \end{subfigure} \hfill
    \begin{subfigure}[t]{0.38\textwidth}
        \includegraphics[width=\textwidth]{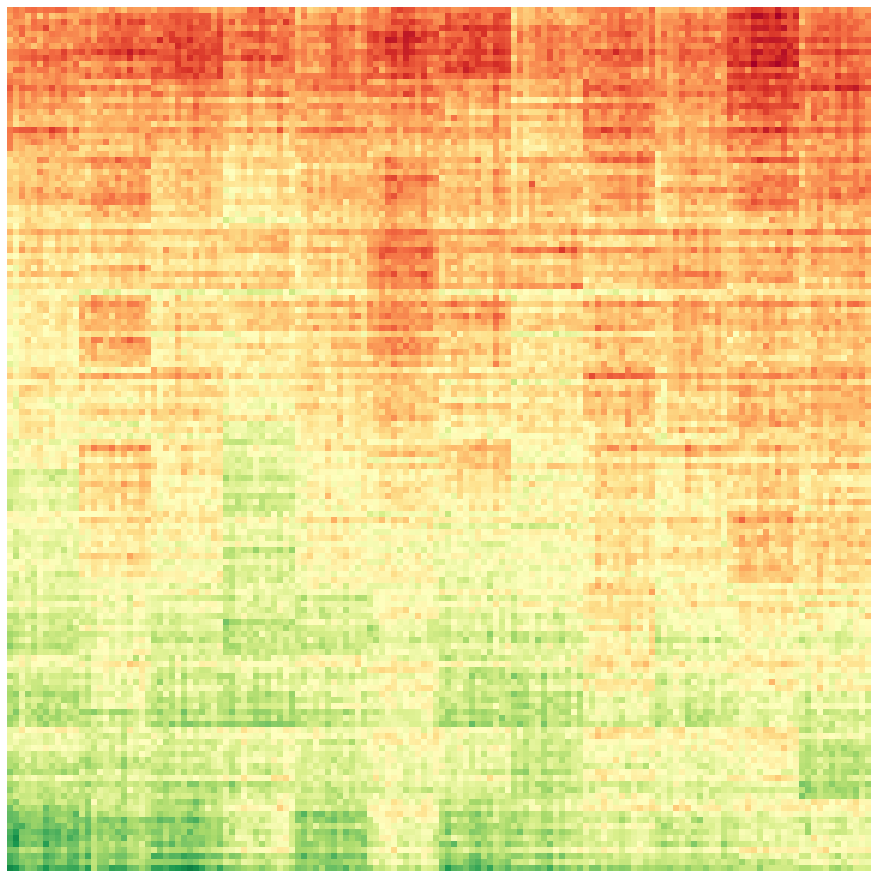}
        \caption{Team 1 vs Team 2 Payoffs}
        \label{fig:tiny_bridge_payoffs}
    \end{subfigure} \hfill
    \begin{subfigure}[t]{0.38\textwidth}
        \includegraphics[width=\textwidth]{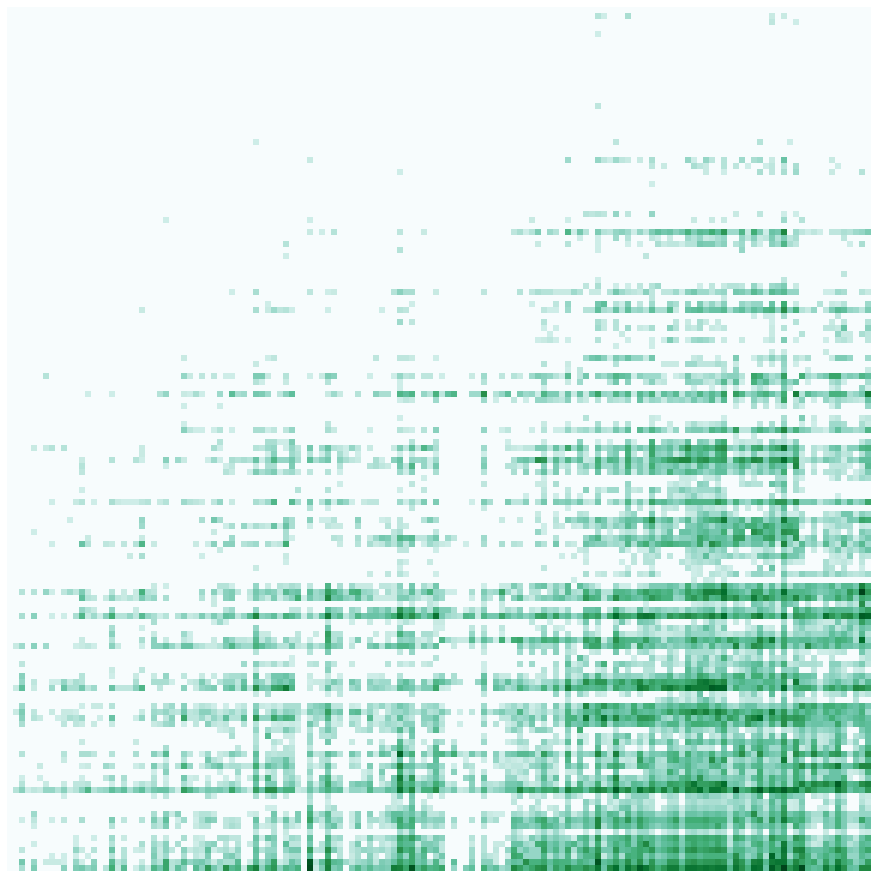}
        \caption{Team 1 vs Team 2 CCE}
        \label{fig:tiny_bridge_cce}
    \end{subfigure}
    \caption{Twelve randomly generated policies for each of the four players in Tiny Bridge were rated using a game theoretic rating technique \citep{marris2022_game_theoretic_rating} based on a CCE. The ratings, payoffs, and CCE have been reordered from lowest to highest rating. The CCE and payoffs are visualized in two dimensions as team vs team for convenience; they remain four player games. Green, yellow and red indicate high, zero, and low team 1 payoff.}
    \label{fig:tiny_bridge_policy}
\end{figure}

\begin{figure}[t!]
    \centering
    \begin{subfigure}[t]{0.49\textwidth}
        \invarplot{no quadrant names,image=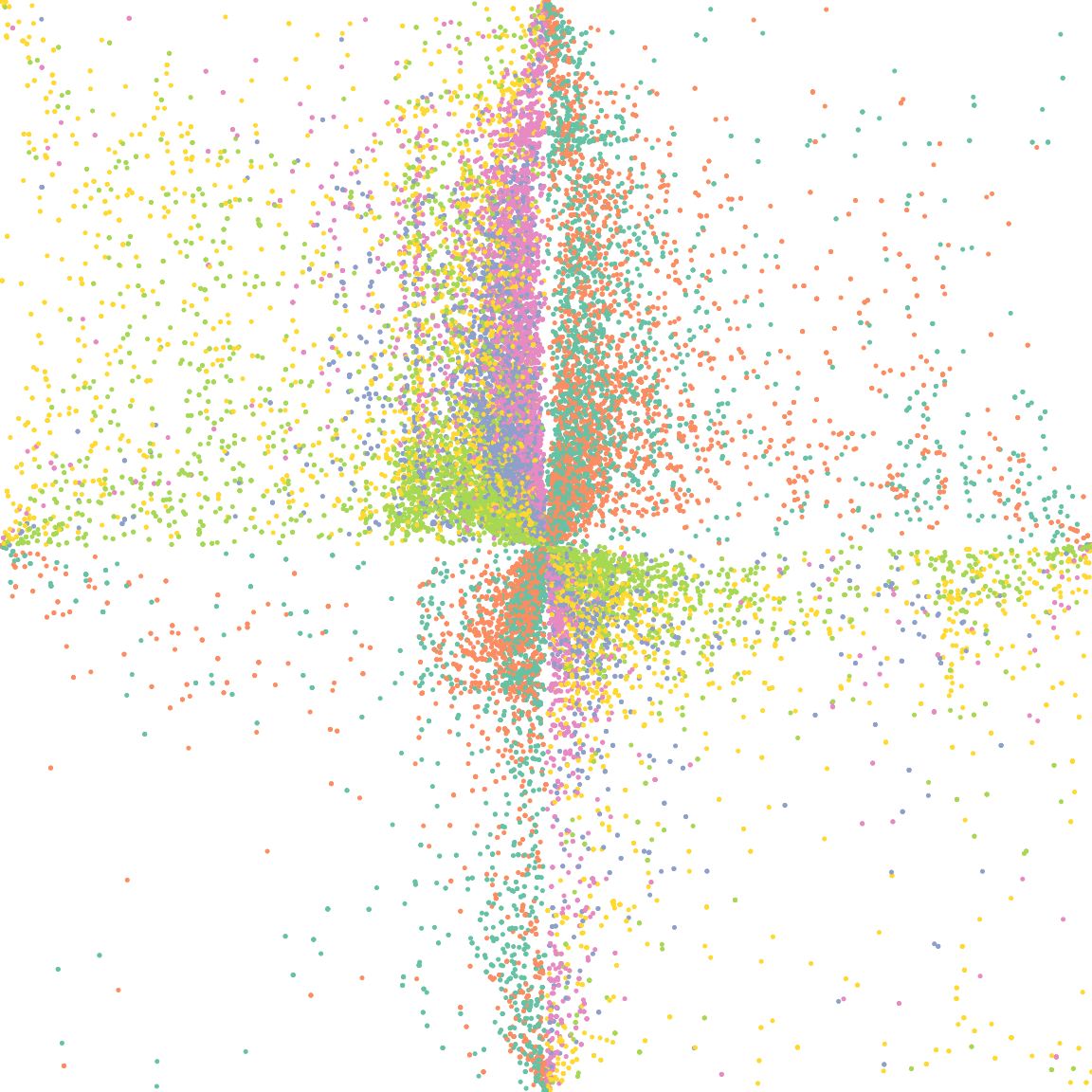,no partial trivial embedding,rc axes,no phi axes,no highlight symmetries,no symmetric games,no symmetries}
        \caption{Best team 1 vs best team 2}
        \label{fig:tiny_bridge_polymatrix_visualization_best}
    \end{subfigure} \hfill
    \begin{subfigure}[t]{0.49\textwidth}
        \invarplot{no quadrant names,image=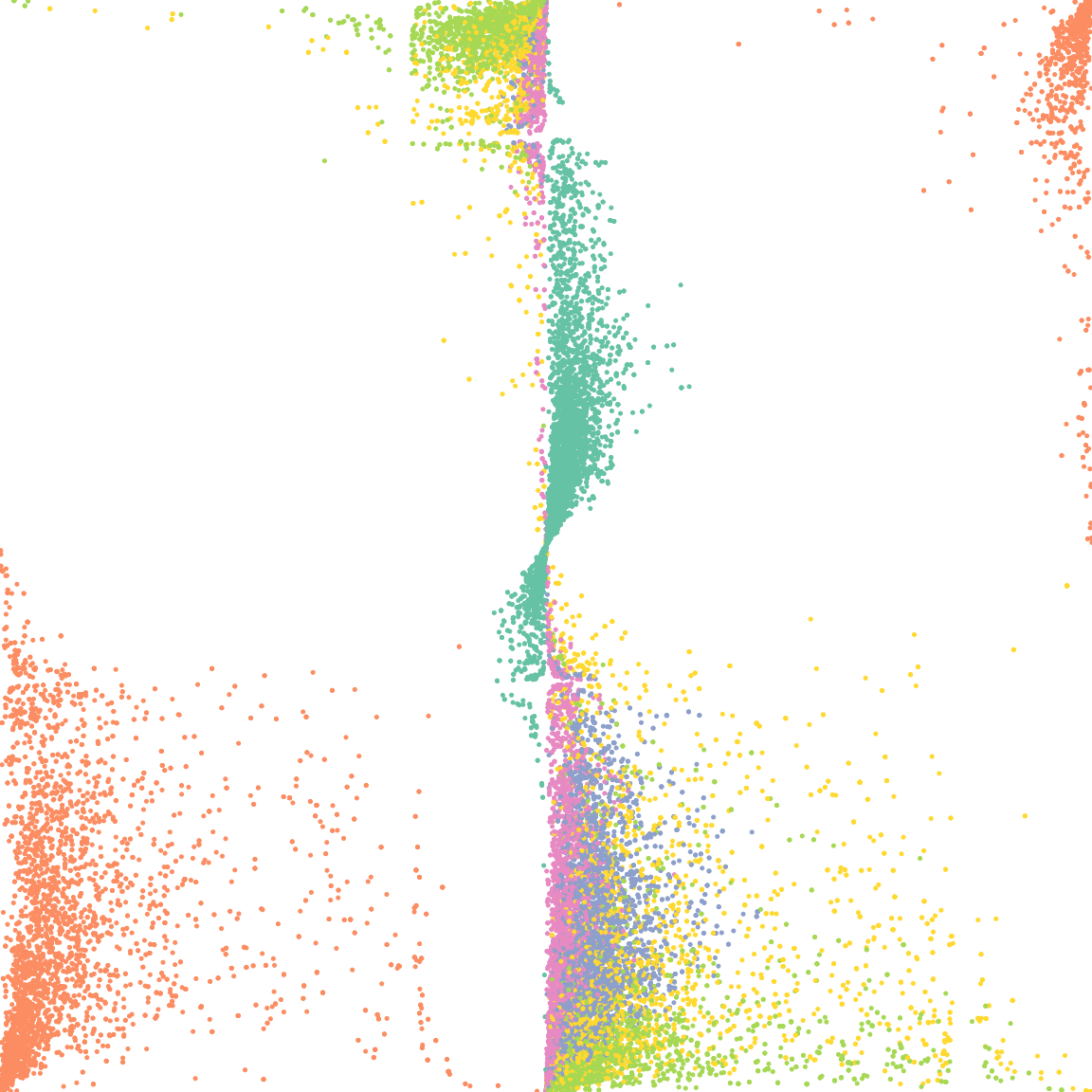,no partial trivial embedding,rc axes,no phi axes,no highlight symmetries,no symmetric games,no symmetries}
        \caption{Best team 1 vs worst team 2}
        \label{fig:tiny_bridge_polymatrix_visualization_mixed}
    \end{subfigure}
    \caption{Visualization of local polymatrix approximation of the 2vs2 team extensive-form Tiny Bridge with two different background policies sets. There are two teams, 1 and 2, each with two players, A and B. Key: \colorsquare{tab1color}~1A vs 1B, \colorsquare{tab2color}~2A vs 2B, \colorsquare{tab3color}~1A vs 2A, \colorsquare{tab4color}~1A vs 2B, \colorsquare{tab5color}~1B vs 2A, and \colorsquare{tab6color}~1B vs 2B.}
    \label{fig:tiny_bridge_polymatrix_visualization}
\end{figure}

\section{Discussion}

In this work we derived distance metrics and embeddings for n-player general-sum normal-form games, such that they are equilibrium-invariant, equilibrium-symmetric, and better-response-invariant. Similar metrics and embeddings could be readily derived for other succinct representations of games such as polymatrix, symmetric, zero-sum, or common-payoff.

To explore equilibrium-invariance, we focused on the most popular equilibrium solution concepts including NEs, CEs, and CCEs. Other equilibria such as quantal response equilibrium (QRE)~\citep{mckelvey1995_qre_lle_nf} may also be compatible with the equilibrium-invariant embedding we derived in this work. There may also be interesting connections to evolutionary stable strategies (ESS). We leave verification of other equilibrium concepts to future work.

Theorem~\ref{theorem:reversible_gains} highlighted the direct connection between the deviation gains and the equilibrium-invariant embedding. Since the deviation gains directly describe the space of equilibria, if one knows the equilibria of a system, one could estimate the deviation gains, and hence the equilibrium-invariant embeddings. Therefore the equilibrium-invariant embedding could be a useful tool in inverse game theory.

A metric-space can be useful in performing perturbation analysis (as hinted in Section~\ref{sec:equi_distance}). It is not uncommon for payoffs to be estimated from data, such as in empirical game-theoretic analysis (EGTA) \citep{walsh2002_egta,wellman2006_egta}. Such payoffs have uncertainty, and small changes in the payoffs can cause large changes to the resulting equilibria. A notion of distance between games can help answer questions about how a game's equilibria may change with a different estimate. For example, is it near an equivalence class boundary, or about to switch from being invariant-zero-sum to invariant-common-payoff?

Popular game-theoretic rating methods, like Nash Average \citep{balduzzi2018_nashaverage}, are used to rate and rank strategies in normal-form games. These rankings are invariant to affine transforms of the payoffs (Section~\ref{sec:game_theoretic_rating}). Therefore the equilibrium-invariant embedding preserves the game-theoretic ranking of strategies, which is further evidence of its fundamental nature.

Furthermore, we discovered a novel per-strategy scale better-response-invariant transform, that can be used to derive a set of 2×2 best-response-invariant embeddings. After symmetry, this results in 15 equivalence classes. This same set has been identified previously by studying best-responses of 2×2 games \citep{borm1987_classification_of_2x2_games}. There is a deep connection between best-response-invariance and equilibrium-invariance \citep{morris2004_best_response_equivalence}. However, this work also provides an equilibrium-invariant embedding, distance measure, efficient parameterization, graph representation, and naming scheme for the best-response-invariant embeddings.

We note that the majority of the games in the best-response-invariant embedding have indifferences (11 out of 15). Indifferences cannot be captured in ordinal payoffs: partially ordinal payoffs are required. However, popular topologies and classifications for 2×2 games only focus on ordinal games. We argue that such a choice both a) limits the space of interesting games that are studied, and b) places too much prominence on games without indifferences which are highly redundant, comprising of only 4 out of 15 of the best-response-invariant equivalence classes.

The 2×2 best-response-invariant embeddings can be used quickly to calculate equilibria for any 2×2 game, simply by storing the extreme points of the (C)CE polytope in a lookup table for the 15 best-response-invariant embeddings. The equilibria can be calculated for any 2×2 game by a) calculating the equilibrium-invariant embedding b) identifying the equivalence class the game belongs to, and c) scaling the equilibria in the lookup table according to Theorem~\ref{thm:fundamental_set}.

Summarizing and visualizing large datasets with high dimensionality is an important area of research in machine learning. Principled techniques like PCA \citep{pearson1901_pca,hotelling1936_pca} that reduce dimensionality, while maintaining the maximum amount of information, are ubiquitous. PCA allows data to be visualized in fewer dimensions. Other less principled techniques like t-SNE \citep{vandermaaten2008_tsne,hintonroweis2003_sne} are also a very popular tool for inspecting datasets. While the fields of statistics and machine learning have benefited from such tools for decades, game theory has lacked such analysis tools, although some attempts have been made to visualize games \citep{czarnecki2020_spinning,omidshafiei2022_multiagent_behaviour_analysis,omidshafiei2020_landscape_of_multiplayer_games}. Large games are extremely complex: they have many possible equilibria and complicated better-response dynamics. The visualization tools in this work build on fundamental principles and could be an important first step to developing such analysis techniques for large games.

\paragraph{Limitations} \hfill \\
To motivate the metric spaces and embeddings we focused on 2×2 games. We derived an efficient parameterization of the 2×2 equilibrium-invariant embedding, which requires only two variables. This is achieved by making an assumption: only the equilibria, or similarly, the strategic interactions (best-response dynamics), of games are important. Most solution concepts are equilibrium or best-response based and the majority of the study of games is devoted to finding these solutions. However, this assumption does have a consequence: the preference ordering of \emph{joint strategy} payoffs is not necessarily maintained after equilibrium-invariant transforms. This means that equilibrium selection methods such as maximum welfare could select for different equilibria in the equilibrium-invariant embedding\footnote{Although, with knowledge of the mean and scale used to make the transformation, a trivial modification to the objective of the linear program could preserve the selection}. Furthermore, the narrative of games that are motivated based on the ordering of joint payoffs may unravel. For example the most studied game, Prisoner's Dilemma, has dominant strategies that do not result in a welfare maximizing equilibrium. \ordinalgame{4231}{4321}~Prisoner's Dilemma is transformed to \dominantgame~Dominant which is intuitive from an equilibrium preserving perspective, but less so from a welfare or social dilemma perspective. We stress that this is a feature of the methods described in this work. If one is only concerned with the resulting rational behaviour of players, simplifying analysis of the game to only its equilibrium-invariant embedding or best-response-invariant embedding is a valuable way of removing redundant features of games.

The 2×2 equilibrium-invariant embedding can be leveraged to produce visualizations of larger games. Payoff structure, equilibrium properties, and other details can be read from the visualizations by observing where the point cloud spreads and how the density of points land in best-response-invariant equivalence classes. Although we provide evidence that visualizations can give an insight into how cyclic a game is, it is only capable of showing cycles of length two, longer cycles may not be captured. For example, Rock-Paper-Scissors visualized using a 2×2 point cloud would only show ~\dominantgame~Dominant points.

\section{Conclusion}

We study payoff transforms that reduce the degrees of freedom in the space of games. This makes them a) easier to understand, b) easier to sample from, and c) easier to visualize. In this work we defined an equilibrium inspired metric-space, an equilibrium-invariant embedding, equilibrium-symmetric embedding, and better-response-invariant embedding for games. In order to explore these concepts we focused on 2×2 games which are well studied in the literature. We uncovered an efficient two variable parameterization of the 2×2 equilibrium-invariant embedding. These variables geometrically represent angles on unit circles which allows them to be clearly visualized in two dimensions. Several properties, including equilibrium support, cyclicness, competitiveness, distances, and symmetries, can be read from this visualization, which we explore. We applied a new equilibrium-payoff transform that enabled further simplification, and rediscovered a set of 15 equivalence classes of games. We argued this set is fundamental, and covers all interesting 2×2 strategic interactions. We name and explore the properties of these classes. Finally, we develop visualization tools for 2×2 and $|\mathcal{A}_1|$×$|\mathcal{A}_2|$ normal-form games and arbitrary dimensional polymatrix games. Since all extensive-form games have a normal-form representation and a local polymatrix approximation representation, we explored visualizing the strategic space of large games that have until now been considered intractable to visualize. We hope this work builds intuition for normal-form games, champions the set of 15 2×2 game classes, provides useful visualization tools for game theorists, and computational insights for game theoretic algorithm developers. 

\section*{Acknowledgements}

Thanks to Georg Ostrovski, David Parkes, Joel Leibo, Thore Graepel, Karl Tuyls, Ed Hughes, Harshnira Patani, and Kevin McKee for helpful discussion and feedback.




\bibliography{bibtex,colab}

\clearpage
\appendix

\section{Deviation Gains as Linear Operators}
\label{sec:gains_as_linear_operators}

The deviation gains (Equations~\eqref{eq:wsce_def}, \eqref{eq:ce_def}) and~\eqref{eq:cce_def}) can be found from linear transforms of each player's payoff.

\begin{lemma} \label{lemma:deviation_gain_rank}
    The WSCE, CE and CCE deviation gains can be written as linear operations on the payoffs, $A^\text{WSCE}(a'_p, a''_p, a_{-p}) = \sum_{\tilde{a}} T^\text{CCE}_p(a'_p, a''_p, a_{-p}, \tilde{a}) G_p(\tilde{a})$, $A^\text{CE}(a'_p, a''_p, a) = \sum_{\tilde{a}} T^\text{CE}_p(a'_p, a''_p, a, \tilde{a}) G_p(\tilde{a})$, and $A^\text{CCE}(a'_p, a) = \sum_{\tilde{a}} T^\text{CCE}_p(a'_p, a, \tilde{a}) G_p(\tilde{a})$. The rank of the linear operations is $|\mathcal{A}| - |\mathcal{A}_{-p}|$.
\end{lemma}
\begin{proof}
    Flatten the payoff into a vector, $G_p(a''')$, of length $|\mathcal{A}|$, flatten the gain into vectors, $A^\text{CCE}_p(a'_p \otimes a)$, $A^\text{WSCE}_p(a'_p \otimes a''_p \otimes a_{-p})$ and $A^\text{CE}_p(a'_p \otimes a''_p \otimes a)$, of length $|\mathcal{A}_p||\mathcal{A}|$, $|\mathcal{A}_p|^2|\mathcal{A}_p|$ and $|\mathcal{A}_p|^2|\mathcal{A}|$, and use matrix linear operators, $T^\text{CCE}_p(a'_p \otimes a, a)$, $T^\text{WSCE}_p(a'_p \otimes a''_p \otimes a_{-p}, a''')$ and $T^\text{CE}_p(a'_p \otimes a''_p \otimes a, a''')$, with shapes $|\mathcal{A}_p||\mathcal{A}| \times |\mathcal{A}|$, $|\mathcal{A}_p|^2|\mathcal{A}_{-p}| \times |\mathcal{A}|$ and $|\mathcal{A}_p|^2|\mathcal{A}| \times |\mathcal{A}|$. Inspect the block matrix structure of $T^\text{CCE}_1$, where $I$ is the identity matrix of shape $|\mathcal{A}_{-1}| \times |\mathcal{A}_{-1}|$, which has $|\mathcal{A}_1|$ block columns and $|\mathcal{A}_1|^2$ block rows. Note the property that the definition of the WSCE is just a reshaped version of the CCE, $A^\text{WSCE}_p(a'_p, a''_p, a_{-p})= A^\text{CCE} (a'_p, a''_p, a_{-p})$ . A similar inspection of $T^\text{CE}_1$ can be made, which has $|\mathcal{A}_1|$ block columns and $|\mathcal{A}_1|^3$ block rows.
    
    \noindent\begin{subequations}
        \begin{minipage}[c]{0.48\textwidth}
            \begin{align}
                \tilde{T}^\text{WSCE}_1 = T^\text{CCE}_1  = \begin{bmatrix}
                     $0$ &  $0$ & ... &  $0$ &  $0$ \\
                     $I$ & $-I$ & ... &  $0$ &  $0$ \\
                     \vdots & \vdots & \ddots &  \vdots &  \vdots \\
                     $I$ &  $0$ & ... & $-I$ &  $0$ \\
                     $I$ &  $0$ & ... &  $0$ & $-I$ \\ \hline
                    $-I$ &  $I$ & ... &  $0$ &  $0$ \\
                     $0$ &  $0$ & ... &  $0$ &  $0$ \\
                     \vdots & \vdots & \ddots &  \vdots &  \vdots \\
                     $0$ &  $I$ & ... & $-I$ &  $0$ \\
                     $0$ &  $I$ & ... &  $0$ & $-I$ \\  \hline
                     \vdots & \vdots & \vdots &  \vdots & \vdots \\  \hline
                    $-I$ &  $0$ & ... &  $0$ &  $I$ \\
                     $0$ & $-I$ & ... &  $0$ &  $I$ \\
                     \vdots & \vdots & \ddots &  \vdots &  \vdots \\
                     $0$ &  $0$ & ... & $-I$ &  $I$ \\
                     $0$ &  $0$ & ... &  $0$ &  $0$
                \end{bmatrix}
            \end{align}
        \end{minipage}
        \hfill
        \begin{minipage}[c]{0.48\textwidth}
            \begin{align}
                T^\text{CE}_1 = \begin{bmatrix}
                     $0$ &  $0$ & ... &  $0$ &  $0$ \\
                     $0$ &  $0$ & ... &  $0$ &  $0$ \\
                     \vdots & \vdots & \ddots &  \vdots &  \vdots \\
                     $0$ &  $0$ & ... &  $0$ &  $0$ \\
                     $0$ &  $0$ & ... &  $0$ &  $0$ \\  \hline
                     $0$ &  $0$ & ... &  $0$ &  $0$ \\
                     $I$ & $-I$ & ... &  $0$ &  $0$ \\
                     \vdots & \vdots & \ddots &  \vdots &  \vdots \\
                     $0$ &  $0$ & ... &  $0$ &  $0$ \\
                     $0$ &  $0$ & ... &  $0$ &  $0$ \\  \hline
                     \vdots & \vdots & \vdots &  \vdots & \vdots \\  \hline
                     $0$ &  $0$ & ... &  $0$ &  $0$ \\
                     $0$ &  $0$ & ... &  $0$ &  $0$ \\
                     \vdots & \vdots & \ddots &  \vdots &  \vdots \\
                     $0$ &  $0$ & ... &  $0$ &  $0$ \\
                     $I$ &  $0$ & ... &  $0$ & $-I$ \\  \hline \hline
                    $-I$ &  $I$ & ... &  $0$ &  $0$ \\
                     $0$ &  $0$ & ... &  $0$ &  $0$ \\
                     \vdots & \vdots & \ddots &  \vdots &  \vdots \\
                     $0$ &  $0$ & ... &  $0$ &  $0$ \\
                     $0$ &  $0$ & ... &  $0$ &  $0$ \\  \hline
                     $0$ &  $0$ & ... &  $0$ &  $0$ \\
                     $0$ &  $0$ & ... &  $0$ &  $0$ \\
                     \vdots & \vdots & \ddots &  \vdots &  \vdots \\
                     $0$ &  $0$ & ... &  $0$ &  $0$ \\
                     $0$ &  $0$ & ... &  $0$ &  $0$ \\ \hline
                     \vdots & \vdots & \vdots &  \vdots &  \vdots
                \end{bmatrix}
            \end{align}
        \end{minipage}
    \end{subequations}
    
    The first block column can be constructed from the negative sum of the remaining block columns. Therefore there are $|\mathcal{A}_{-1}|$ redundant columns. The remaining are linearly independent, resulting in a rank of $|\mathcal{A}| - |\mathcal{A}_{-1}|$. Similar construction patterns can be made for $T^\text{CCE}_p(a'_p \otimes a, a)$. Again, one block column, or $|\mathcal{A}_{-p}|$ are not linearly independent, so the rank is $|\mathcal{A}| - |\mathcal{A}_{-1}|$.
\end{proof}

\section{Game Theoretic Rating}
\label{sec:game_theoretic_rating}

Game-theoretic rating is a method for rating strategies in a normal-form game. Let a strategy rating, $r_p(a_p)$, be a numerical scalar for each player's strategies, and the strategy rank be the sorting order $w_p(a_p) = \arg \sort r_p(a_p)$ defined such that equal ratings are given equal rank (for example, $[1.1, -0.2, 0.3, 0.3] \to [3, 0, 1, 1]$). A very simple (but not game-theoretic) way of rating strategies in a game would be to examine their average payoff.
\begin{definition}[Average Rating]
    \begin{align}
        r^\text{avg}_p(a_p) = \frac{1}{|\mathcal{A}_{-p}|} \sum_{a_{-p}} G_p(a_p, a_{-p})
    \end{align}
\end{definition}

A popular game-theoretic rating method, Nash averaging (NA) \citep{balduzzi2018_nashaverage}, weights the payoffs by mixing over the maximum entropy Nash equilibrium solution. It is most commonly applied in two-player zero-sum where the solution is unique.
\begin{definition}[Nash Average Rating]
    \begin{align}
        r^\text{NA}_p(a_p) = \sum_{a_{-p}} G_p(a_p, a_{-p}) \left(\otimes_{q \in -p} \sigma_p^\text{MENE}(a_q) \right)
    \end{align}
\end{definition}

\begin{theorem}[Rank-Invariance]
    Nash-average is rank-invariant to affine transformations, $G_p(a) \to \hat{G}_p(a) = s_p G_p(a) + b_p(a_{-p})$, where $s_p > 0$ and $b_p(a_{-p})$ is arbitrary.
\end{theorem}
\begin{proof}
    Consider the ranking after an affine transformation
    \begin{subequations}
    \begin{align}
        \hat{r}^\text{NA}_p(a_p) &= \left( \sum_{a_{-p}} s_p G_p(a_p, a_{-p}) + b_p(a_{-p}) \right) \left(\otimes_{q \in -p} \sigma_p^\text{MENE}(a_q) \right) \\
        &= s_p r^\text{NA}_p(a_p) +  \underbrace{\sum_{a_{-p}} b_p(a_{-p}) \left(\otimes_{q \in -p} \sigma_p^\text{MENE}(a_q) \right)}_{\text{Constant: does not depend on $a_p$.}}
    \end{align}
    \end{subequations}
    It is easy to see that $w_p(a_p) = \hat{w}_p(a_p)$.
\end{proof}

The equilibrium-invariant embedding therefore preserves game-theoretic ranking of strategies. Furthermore, the equilibrium-invariant embedding may provide a more natural normalization of payoffs than the approach originally suggested by \cite{balduzzi2018_nashaverage}.

\section{Additional Figures}

\begin{figure}
    \centering
    \includegraphics[width=1.0\textwidth]{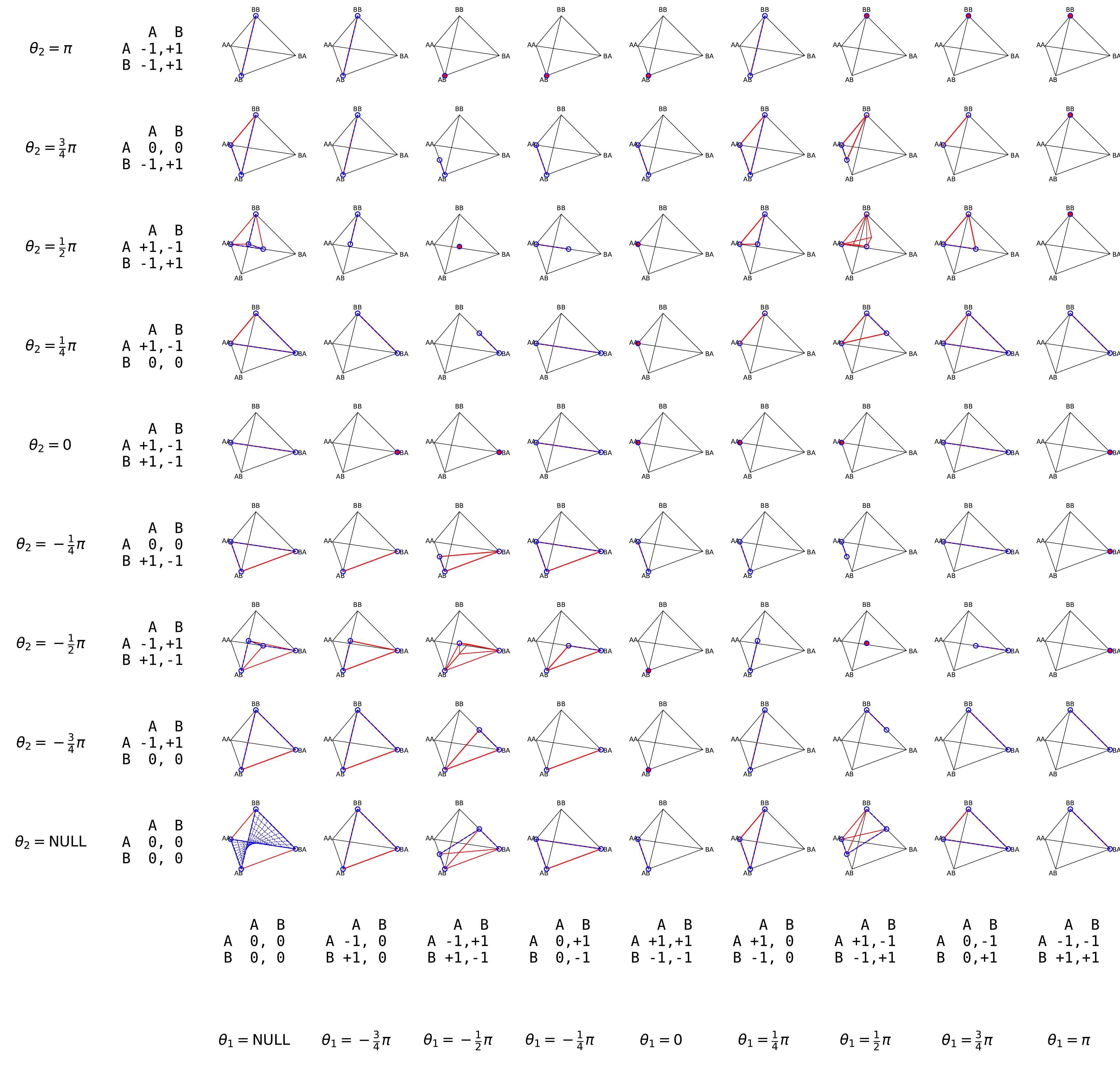}
    \caption{The space of NEs (dashed blue lines) and (C)CEs (solid red polytopes) equilibria for 2×2 games. For two-strategy games CEs and CCEs are identical. The left column contains games where player 1 has trivial payoff. The bottom row contains games where player 2 has trivial payoff. The intersection of the polytopes in these games results in the polytopes of the non-null games. In \nullgame~Null, all joints are (C)CEs and all factorizable joints are NEs. The NEs of the diagonal quadrants have convex (interchangeable) solutions because they can be mapped onto zero-sum games. The off-diagonal games are non-interchangeable because they can be mapped onto common-payoff games.}
    \label{fig:better_response_invariant_eqiulibria}
\end{figure}

\begin{figure}[h]
    \centering
    \input{tex/figure_equilibrium_support}
    \caption{Shows the joint strategies that can have support in equilibrium over the space of 2×2 games. The tiling indicates the support of a joint distribution (square) with four optionally shaded quadrants that correspond to the four joint strategies. All nonzero combinations $2^4 - 1 = 15$ of supports are possible. For example, \cyclegame~Cycle and \coordinationgame Coordination have full support \cyclesupport, \dominantgame~Dominant has pure support \dominantsupport, \fossickgame~Fossick has diagonal support \fossicksupport, and \daredevilgame~Daredevil has support over all but one joint \daredevilsupport. The majority of space either has a pure joint strategy or permits full-support mixed equilibrium. Other equilibria are possible but are measure-zero and exist on the boundaries between the pure and full-support games.}
    \label{fig:equilibrium_support}
\end{figure}


\begin{figure}[h]
    \centering
    \begin{tikzpicture}[
    /pgfplots/y=7cm, /pgfplots/x=7cm 
]
\begin{axis}[
    anchor=origin, 
    axis lines = middle,
    xmin=0.0, xmax=+1.3,
    ymin=-1.1, ymax=+1.15,
    xticklabels={}, xtick={+1},
    yticklabels={,}, ytick={-1,+1},
    xlabel = \(\phi_\text{cycle}\),
    ylabel = {\(\phi_\text{coord}\)},
    xmajorgrids=false,
    ymajorgrids=false,
    grid style=dashed,
]

\addplot[color=black] coordinates {(0.0,1.0) (1.0, 0.0) (0.0,-1.0)};

\addplot[color=gray, dotted] coordinates {(0.25,-0.75) (0.0,-0.5) (0.5,0.0)  (0.75,-0.25)};
\addplot[color=gray, dotted] coordinates {(0.5,0.0)  (0.75,0.25)};
\addplot[color=gray, dotted] coordinates {(0.0,0.5)  (0.25,0.75)};

\addplot[color=gray, dashed] coordinates {(0, 0) (0.5, 0.5)};
\addplot[color=gray, dashed] coordinates {(0, 0) (0.5, -0.5)};

\end{axis}

\matrix (m_coordination) at (0.0*7cm,1.0*7cm) [table,text width=1.6em,label={[align=center]below:\footnotesize \equigame{+-}{-+}{+-}{-+}~Coordination}]
{
\node[fill=gray!10,fill opacity=.5,text opacity=1]{$+1,+1$}; & \node[fill=gray!10,fill opacity=.5,text opacity=1]{$-1,-1$}; \\
\node[fill=gray!10,fill opacity=.5,text opacity=1]{$-1,-1$}; & \node[fill=gray!10,fill opacity=.5,text opacity=1]{$+1,+1$}; \\
};

\matrix (m_coordination) at (0.0*7cm,0.5*7cm) [table,text width=1.6em,label={[align=center]below:\footnotesize \equigame{+-}{00}{+-}{00}~Fossick}]
{
\node[fill=gray!10,fill opacity=.5,text opacity=1]{$+1,+1$}; & $~0,-1$ \\
$-1,~0$ & \node[fill=gray!10,fill opacity=.5,text opacity=1]{$~0,~0$}; \\
};

\matrix (t1) at (0.0*7cm,0.0*7cm) [table,text width=1.6em,label={[align=center]below:\footnotesize \equigame{+-}{+-}{+-}{+-}~Dominant}]
{
\node[fill=gray!10,fill opacity=.5,text opacity=1]{$+1,+1$}; & $+1,-1$ \\
$-1,+1$ & $-1,-1$ \\
};

\matrix (tr1) at (0.0*7cm,-0.5*7cm) [table,text width=1.6em,label={[align=center]below:\footnotesize  \equigame{00}{+-}{00}{+-}~Daredevil}]
{
\node[fill=gray!10,fill opacity=.5,text opacity=1]{$\phantom{+}0,\phantom{+}0$}; & \node[fill=gray!10,fill opacity=.5,text opacity=1]{$+1,\phantom{+}0$}; \\
\node[fill=gray!10,fill opacity=.5,text opacity=1]{$\phantom{+}0,+1$}; & $-1,-1$ \\
};

\matrix (m_coordination) at (0.0*7cm,-1.0*7cm) [table,text width=1.6em,label={[align=center]below:\footnotesize\color{gray!20} \equigame{-+}{+-}{-+}{+-}~Coordination}]
{
\node[fill=gray!10,opacity=.25]{$-1,-1$}; & \node[fill=gray!10,opacity=.25]{$+1,+1$}; \\
\node[fill=gray!10,opacity=.25]{$+1,+1$}; & \node[fill=gray!10,opacity=.25]{$-1,-1$}; \\
};

\matrix (tr1) at (0.25*7cm,0.75*7cm) [table,text width=1.6em,label=below:\footnotesize  \equigame{+-}{-+}{+-}{00}~Safety]
{
\node[fill=gray!10,fill opacity=.5,text opacity=1]{$+1,+1$}; & $-1,-1$ \\
\node[fill=gray!10,fill opacity=.5,text opacity=1]{$-1,~0$}; & \node[fill=gray!10,fill opacity=.5,text opacity=1]{$+1,~0$}; \\
};

\matrix (tr1) at (0.25*7cm,0.25*7cm) [table,text width=1.6em,label=below:\footnotesize  \equigame{+-}{00}{+-}{+-}~Picnic]
{
\node[fill=gray!10,fill opacity=.5,text opacity=1]{$+1,+1$}; & $~0,-1$ \\
$-1,+1$ & $~0,-1$ \\
};

\matrix (tr1) at (0.25*7cm,-0.25*7cm) [table,text width=1.6em,label=below:\footnotesize  \equigame{+-}{+-}{00}{+-}~Aidos]
{
\node[fill=gray!10,fill opacity=.5,text opacity=1]{$+1,~0$}; & \node[fill=gray!10,fill opacity=.5,text opacity=1]{$+1,~0$}; \\
$-1,+1$ & $-1,-1$ \\
};

\matrix (tr1) at (0.25*7cm,-0.75*7cm) [table,text width=1.6em,label=below:\footnotesize \color{gray!20} \equigame{00}{+-}{-+}{+-} Safety]
{
\node[fill=gray!10,opacity=.25]{$~0,-1$}; & \node[fill=gray!10,opacity=.25]{$+1,+1$}; \\
\node[fill=gray!10,opacity=.25]{$~0,+1$}; & \node[opacity=.25]{$-1,-1$}; \\
};

\matrix (tr1) at (0.5*7cm,0.5*7cm) [table,text width=1.6em,label=below:\footnotesize  \equigame{+-}{-+}{+-}{+-}~Samaritan]
{
\node[fill=gray!10,fill opacity=.5,text opacity=1]{$+1,+1$}; & $-1,-1$ \\
$-1,+1$ & $+1,-1$ \\
};

\matrix (tr1) at (0.5*7cm,0.0*7cm) [table,text width=1.6em,label=below:\footnotesize  \equigame{+-}{00}{00}{+-}~Heist]
{
\node[fill=gray!10,fill opacity=.5,text opacity=1]{$+1,~0$}; & \node[fill=gray!10,fill opacity=.5,text opacity=1]{$~0,~0$}; \\
$-1,+1$ & $~0,-1$ \\
};

\matrix (tr1) at (0.5*7cm,-0.5*7cm) [table,text width=1.6em,label=below:\footnotesize \color{gray!20} \equigame{+-}{+-}{-+}{+-}~Samaritan]
{
\node[opacity=.25]{$+1,-1$}; & \node[fill=gray!10,opacity=.25]{$+1,+1$}; \\
\node[opacity=.25]{$-1,+1$}; & \node[opacity=.25]{$-1,-1$}; \\
};

\matrix (tr1) at (0.75*7cm,0.25*7cm) [table,text width=1.6em,label=below:\footnotesize  \equigame{+-}{-+}{00}{+-}~Hazard]
{
\node[fill=gray!10,fill opacity=.5,text opacity=1]{$+1,~0$}; & \node[fill=gray!10,fill opacity=.5,text opacity=1]{$-1,~0$}; \\
$-1,+1$ & $+1,-1$ \\
};

\matrix (tr1) at (0.75*7cm,-0.25*7cm) [table,text width=1.6em,label=below:\footnotesize \color{gray!20} \equigame{+-}{00}{-+}{+-}~Hazard]
{
\node[opacity=.25]{$+1,-1$}; & \node[fill=gray!10,opacity=.25]{$~0,+1$}; \\
\node[opacity=.25]{$-1,+1$}; & \node[fill=gray!10,opacity=.25]{$~0,-1$}; \\
};

\matrix (m2) at (1.0*7cm,0.0*7cm) [table,text width=1.6em,label={[align=center]below:\footnotesize \equigame{+-}{-+}{-+}{+-}~Cycle}]
{
\node[fill=gray!10,fill opacity=.5,text opacity=1]{$+1,-1$}; & \node[fill=gray!10,fill opacity=.5,text opacity=1]{$-1,+1$}; \\
\node[fill=gray!10,fill opacity=.5,text opacity=1]{$-1,+1$}; & \node[fill=gray!10,fill opacity=.5,text opacity=1]{$+1,-1$}; \\
};

\matrix (tr1) at (1.25*7cm,-0.25*7cm) [table,text width=1.6em,label=below:\footnotesize \equigame{+-}{-+}{00}{00}~Horseplay]
{
\node[fill=gray!10,fill opacity=.5,text opacity=1]{$+1,~0$}; & \node[fill=gray!10,fill opacity=.5,text opacity=1]{$-1,~0$}; \\
\node[fill=gray!10,fill opacity=.5,text opacity=1]{$-1,~0$}; & \node[fill=gray!10,fill opacity=.5,text opacity=1]{$+1,~0$}; \\
};

\matrix (tr1) at (1.0*7cm,-0.5*7cm) [table,text width=1.6em,label=below:\footnotesize \equigame{+-}{00}{00}{00}~Dress]
{
\node[fill=gray!10,fill opacity=.5,text opacity=1]{$+1,~0$}; & \node[fill=gray!10,fill opacity=.5,text opacity=1]{$~0,~0$}; \\
$-1,~0$ & \node[fill=gray!10,fill opacity=.5,text opacity=1]{$~0,~0$}; \\
};

\matrix (tr1) at (0.75*7cm,-0.75*7cm) [table,text width=1.6em,label=below:\footnotesize \equigame{+-}{+-}{00}{00}~Ignorance]
{
\node[fill=gray!10,fill opacity=.5,text opacity=1]{$+1,~0$}; & \node[fill=gray!10,fill opacity=.5,text opacity=1]{$+1,~0$}; \\
$-1,~0$ & $-1,~0$ \\
};

\matrix (tr1) at (1.25*7cm,-0.75*7cm) [table,text width=1.6em,label=below:\footnotesize  \nullgame~Null]
{
\node[fill=gray!10,fill opacity=.5,text opacity=1]{$~0,~0$}; & \node[fill=gray!10,fill opacity=.5,text opacity=1]{$~0,~0$}; \\
\node[fill=gray!10,fill opacity=.5,text opacity=1]{$~0,~0$}; & \node[fill=gray!10,fill opacity=.5,text opacity=1]{$~0,~0$}; \\
};
\end{tikzpicture}
    \caption{Shows 15 fundamental two-player two-strategy best-response-invariant embeddings. Payoffs are shown in the table, with shaded joint strategies appearing in an equilibrium. The faded games are symmetric permutations of other games. All 2×2 games will map onto one of these games. The dashed gray lines show the quadrants: the right quadrant is zero-sum clockwise cyclic and left half-quadrants show coordination games. The dotted lines show the equilibrium solution region which, going clockwise from the top, are diagonal-coordination, top-left dominant, clockwise-coordination, top-right dominant, and off-diagonal-coordination. Games outside the equilibrium-symmetric embedding are trivial.}
\end{figure}

\end{document}